%% file: non_iso_mcl_sisc.tex
\documentclass[review,hidelinks,onefignum,onetabnum]{siamart250211}

\input{non_iso_mcl_shared}

\ifpdf
\hypersetup{
  pdftitle={Thermodynamically consistent modelling and simulation of the moving contact line problem in non-isothermal compressible two-phase flows},
  pdfauthor={Junkai Wang, and Qiaolin He}
}
\fi

\begin{document}

\maketitle

\begin{abstract}
According to the dynamic van der Waals theory, we propose a thermodynamically consistent model for non-isothermal compressible two-phase flows with contact line motion.
In this model, fluid temperature is treated as a primary variable, characterized by the proposed temperature equation instead of being obtained from intermediate variables such as total energy density, internal energy density and entropy density.
The hydrodynamic boundary conditions, which represent a generalization of the generalized Navier slip boundary condition in non-isothermal flows, are imposed on the proposed model. 
We then develop the dimensionless form of the model and prove that it rigorously satisfies the first and second laws of thermodynamics.
Two numerical schemes based on the dimensionless system are constructed: one is fully coupled and thermodynamically consistent, namely strictly satisfying the temporally discrete first and second laws of thermodynamics; the other, designed by extending the multiple scalar auxiliary variable approach to entropy production, is decoupled, linear, and unconditionally entropy-stable.
Several numerical results are presented to validate the effectiveness and stability of the proposed method.
\end{abstract}



\section{Introduction}
\label{sec:intro}
In most traditional phase transition theories, including those of dynamics, the fluid temperature is assumed to be homogeneous, i.e., independent of space and time.
However, non-isothermal two-phase (gas-liquid) flows are ubiquitous in scientific and engineering fields, such as oil reservoir, boiling, evaporation, and condensation.
In recent years, there has been growing interest in exploring the models of non-isothermal two-phase flows.
The dynamic van der Waals theory (DVDWT) \cite{van1979thermodynamic,onuki2005dynamic,onuki2007dynamic}, presented for one-component fluids, is generalized by including gradient contributions, enabling the description of two-phase hydrodynamics involving the gas-liquid transition in inhomogeneous temperature.
Building on the DVDWT, a continuum mechanics modelling framework based on the laws of thermodynamics for liquid-vapor flows is proposed in \cite{liu2015liquid}.
Recently, a non-isothermal gas-liquid flow model with the Peng-Robinson equation of state (EoS) \cite{peng1976new} has been derived in \cite{kou2018thermodynamically2,kou2018entropy}.

For modelling of non-isothermal flows, it is essential that models satisfy the fundamental laws of thermodynamics.
Specifically, the first law of thermodynamics, namely energy balance law, is a fundamental physical principle, and the second law of thermodynamics describes the entropy production of realistic irreversible processes. 
Notably, the energy dissipation law for isothermal systems, which is widely admitted in phase-field model, plays a same role as the balance equation for entropy in non-isothermal systems.
Thus, in non-isothermal systems, the entropy balance equation can be used to derive constitutive relations consistent with the second law of thermodynamics. 

In the literature, there are mainly three modelling frameworks for non-isothermal flows. 
The first framework includes the conservation equations for mass, momentum, and total energy, in which the fluid temperature is determined by the total energy density \cite{liu2015liquid,kou2018entropy,kou2025structure}.
Then, the second framework consists of the mass, momentum, and internal energy equations, where fluid temperature is characterized by the internal energy density \cite{kou2018thermodynamically2}.
As suggested by \cite{teshigawara2008droplet}, this framework leads to an artificial parasitic flow \cite{lafaurie1994modelling,jamet2002theory}.
To avoid this issue, the third framework replaces the internal energy equation with the entropy equation \cite{onuki2007dynamic,xu2010contact}, where fluid temperature is evolved by the entropy density. 
However, in numerical simulations, these frameworks introduce additional computational costs due to the use of intermediate variables.

The moving contact line (MCL) problem, where the fluid-fluid interface intersects the solid wall, is a classical problem that occurs in many physical phenomena. 
It is well known that classical hydrodynamical models with a no-slip boundary condition lead to nonphysical singularity in the vicinity of the contact line \cite{qian2006molecular}.
A phase-field model with the generalized Navier boundary condition (GNBC) \cite{qian2003molecular}, has effectively addressed this issue for immiscible flow over flat surfaces. 
Then, enormous efforts have been made to develop effective numerical schemes to approximate this model, for instance, \cite{qian2006variational,he2011least,gao2012gradient,gao2014efficient,xu2023unified} and the references therein.
Particularly, based on the DVDWT, \cite{xu2010contact} has applied the diffuse-interface modelling to the study of contact line motion in non-isothermal gas-liquid systems, with the fluid slip fully taken into account.
Based on \cite{xu2010contact}, the hydrodynamic boundary conditions have been derived in \cite{liu2012hydrodynamic,xu2012droplet,xu2014single} for the non-isothermal and heterogeneous fluid-solid interface, which are able to describe velocity slip and temperature slip (Kapitza resistance) that contribute to interfacial entropy production. 
Recently, assuming a homogeneous temperature and employing the tangential force balance in boundary layer, a compressible two-phase flow model with GNBC based on the realistic EoS has been proposed in \cite{wang2022energy,wang2024cicp}.

In the numerical simulation of non-isothermal flow models, a primary challenge is to construct efficient numerical schemes that satisfy the discrete laws of thermodynamics.
Another key challenge is the strong nonlinearity and tight coupling between variables.
For phase-field models, there has been a sizable body of literature exploring the energy-stable and linear schemes, including the convex splitting approach \cite{elliott1993global,eyre1998unconditionally,wang2025nme}, and the stabilized approach \cite{shen2010numerical}. 
However, extending these popular approaches to non-isothermal models remains challenging.
To address such problems, building on the convex splitting approach, \cite{kou2018thermodynamically2,kou2018entropy} designed a novel thermodynamically stable numerical scheme for a non-isothermal gas-liquid flow model and rigorously proved that the proposed method satisfies both the first and second laws of thermodynamics.
Although these schemes successfully decouple the discrete mass and momentum balance equations, the tight coupling between number density and fluid temperature persists, leading to the requirement of a linearized iterative method. 

The scalar auxiliary variable (SAV) approach \cite{shen2018scalar} and its extended methods \cite{cheng2018multiple,cheng2020new}, are notable works in the phase-field model for constructing efficient and accurate energy-stable schemes for nonlinear dissipative systems.
Recently, \cite{wang2022energy,wang2024cicp} have employed these methods for both Helmholtz free energy and surface free energy to design energy-stable, decoupled, and linear schemes in isothermal compressible two-phase flows with GNBC based on a realistic EoS.
In particular, \cite{liu2024efficient} has extended the multiple SAV (MSAV) approach \cite{cheng2018multiple} to study the Cahn-Hilliard equation with dynamic boundary conditions.
However, when the general slip boundary conditions are introduced into the non-isothermal flow models, especially the fluid-solid interfaces are assumed to be non-isothermal and heterogeneous, it is a challenge to design efficient, fully decoupled, and thermodynamically consistent numerical schemes. 

A main objective of this paper is to propose a new thermodynamically consistent non-isothermal two-phase model, which consists of equations for number density, fluid velocity, and fluid temperature.
The dimensionless form of proposed model is proved to satisfy the fundamental laws of thermodynamics.
Notably, the proposed temperature equation can accommodate more general fluid motions beyond MCL, which is demonstrated in numerical simulations.
For MCL problems, we introduce general slip boundary conditions on both isothermal and non-isothermal fluid-solid interfaces.

Another purpose of this paper is to construct efficient numerical schemes that satisfy the temporally discrete laws of thermodynamics.
Based on the dimensionless system, we develop two numerical schemes.
The first is fully coupled and thermodynamically consistent.
The second is designed by extending the MSAV approach for the part of bulk entropy and surface entropy, and it enjoys several advantages: decoupled, linear, and unconditionally entropy-stable.
We present ample numerical tests to show efficiency and stability of the second scheme for various fluid-solid interfaces.

The remaining sections of this paper are organized as follows.
In \cref{sec:Derivation of the mathematical model}, the DVDWT and hydrodynamic boundary conditions are first reviewed.
Then, a new thermodynamically consistent model is derived.
In \cref{sec:The dimensionless model and thermodynamical consistency}, we develop a dimensionless system for the proposed model and prove its thermodynamical consistency. 
The thermodynamically consistent numerical schemes are constructed in \cref{sec:Thermodynamically consistent numerical schemes}. 
Numerical experiments are implemented and the results are presented in \cref{sec:Numerical experiments}. 
The paper is concluded in \cref{sec:conclusions} with a few remarks.

\section{Derivation of the mathematical model}
\label{sec:Derivation of the mathematical model}

\subsection{Dynamic van der Waals theory}
\label{subsec:Dynamic van der Waals theory}
By fundamental laws of thermodynamics, for a one-component homogeneous system in equilibrium, we have the Euler equation, Gibbs equation and Gibbs-Duhem equation \cite{firoozabadi1999thermodynamics,onuki2002phase}:
\begin{equation}\label{eq:threemaineq}
    \begin{aligned}
        e-Tns+p-n\mu=0,\quad
        d(ns)=\frac{1}{T}de-\frac{\mu}{T}dn,\quad
        -nsdT+dp-nd\mu=0,
    \end{aligned}
\end{equation}
in which $n$, $e$, $s$, $T$, $p$ and $\mu$ are the number density, internal energy density, entropy per particle, fluid temperature, pressure, and chemical potential per particle, respectively. 
Choosing $n$ and $T$ as the independent state variables, we obtain the Helmholtz free energy density (defined by $f=e-Tns$) satisfies \cite{onuki2007dynamic}
\begin{equation*}
    \begin{aligned}
        df=de-nsdT-Td(ns)=-nsdT+\mu dn.
    \end{aligned}
\end{equation*}
The Helmholtz free energy density $f$ as a function of $n$ and $T$ is given by \cite{van1979thermodynamic,onuki2007dynamic}
\begin{equation*}
    \begin{aligned}
        f=\frac{3nk_{B}T}{2}-nk_{B}T\ln{\left(\left(\frac{k_{B}T}{\epsilon}\right)^{\frac{3}{2}}\left(\frac{1}{v_{0}n}-1\right)\right)}-\epsilon v_{0}n^{2},
    \end{aligned}
\end{equation*}
where $k_{B}$ is the Boltzmann constant, $v_{0}$ is the molecular volume, and $\epsilon$ is the strength of attractive interaction. 
From the definition of $f$, we can derive $e$, $s$, and $p=n\mu-f$ as functions of $n$ and $T$ \cite{onuki2005dynamic,xu2012droplet}:
\begin{equation*}
    \begin{aligned}
        e=\frac{3nk_{B}T}{2}-\epsilon v_{0}n^{2}, \quad
        s=k_{B}\ln{\left(\left(\frac{k_{B}T}{\epsilon}\right)^{\frac{3}{2}}\left(\frac{1}{v_{0}n}-1\right)\right)}, \quad
        p=\frac{nk_{B}T}{1-v_{0}n}-\epsilon v_{0}n^{2}.
    \end{aligned}
\end{equation*}
Using \eqref{eq:threemaineq}, the chemical potential $\mu$ changes as
\begin{equation*}
    \begin{aligned}
        \mu=-T\frac{\partial(ns)}{\partial n}+\frac{\partial e}{\partial n}
        =\frac{3k_{B}T}{2}-k_{B}T\ln{\left(\left(\frac{k_{B}T}{\epsilon}\right)^{\frac{3}{2}}\left(\frac{1}{v_{0}n}-1\right)\right)}+\frac{k_{B}T}{1-v_{0}n}-2\epsilon v_{0}n.
    \end{aligned}
\end{equation*}
The van der Waals theory is generalized by including gradient contributions to the Helmholtz free energy, internal energy and entropy \cite{onuki2007dynamic}:
    \begin{equation}\label{eq:feSdefine}
    \begin{aligned}
        f_{b}=f+\frac{M(n,T)}{2}\vert\nabla n\vert^{2}, \quad 
        \hat{e}=e+\frac{K(n)}{2}\vert\nabla n\vert^{2}, \quad 
        \hat{S}=ns-\frac{C(n)}{2}\vert\nabla n\vert^{2}.
    \end{aligned}
    \end{equation}
Here, $M$, $K$ and $C$ are coefficients with relation: $M=K+CT$. 
The entropy in bulk fluid region is defined as $S_{b}=\int_{\Omega} \hat{S}d\mathbf{x}$. 
Then, the fluid temperature $T$ and generalized chemical potential $\hat{\mu}$ are given by
\begin{equation}\label{eq:muTdefine}
    \begin{aligned}
        \frac{1}{T}=\left(\frac{\delta S_{b}}{\delta e}\right)_{n}, \quad 
        \hat{\mu}=-T\left(\frac{\delta S_{b}}{\delta n}\right)_{\hat{e}}=\mu+\frac{M_{n}}{2}|\nabla n|^{2}-T\nabla\cdot\left(\frac{M}{T}\nabla n\right),
    \end{aligned}
\end{equation}
where the subscripts denote $n$ or $\hat{e}$ is fixed in the derivatives and $M_{n} = (\partial M/\partial n)_{T}$.
Regarding $S_{b}$ as a functional of $n$ and $\hat{e}$, we consider small changes $\delta n$ and $\delta \hat{e}$ which yield an incremental change of $S_{b}$. 
Using \eqref{eq:threemaineq}, \eqref{eq:feSdefine} and \eqref{eq:muTdefine}, we obtain \cite{onuki2007dynamic}
\begin{equation}\label{eq:entropyvar}
    \begin{aligned}
        \delta S_{b}
        =&\int_{\Omega} \left(\frac{1}{T}\delta \hat{e} -\frac{\hat{\mu}}{T}\delta n  - \nabla\cdot\left(\frac{M}{T}\nabla n\delta n\right)\right) d\mathbf{x}.
    \end{aligned}
\end{equation}
In \eqref{eq:entropyvar}, maximizing the entropy $S_{b}$ for fixed particle number $N=\int_{\Omega} nd\mathbf{x}$ and fixed internal energy $E_{b}=\int_{\Omega}\hat{e}d\mathbf{x}$ yields the bulk equilibrium conditions: ({\romannumeral1}) the homogeneity of $T$, and ({\romannumeral2}) the homogeneity of $ \hat{\mu}$. 
When $T$ is homogeneous, the bulk Helmholtz free energy functional $F_{b}=E_{b}-TS_{b}$ is given by $F_{b}=\int_{\Omega} f_{b} d\mathbf{x}$, with $f_{b}=\hat{e}-T\hat{S}$. 
Minimizing $F_{b}$ with respect to $n$ also yields homogeneity of $\hat{\mu}$.
We then introduce the generalized pressure $\hat{p}$ by the generalized Euler equation \cite{onuki2002phase}: 
\begin{equation}\label{eq:generalizedEuler}
    \begin{aligned}
        \hat{e}-T\hat{S}+\hat{p}-n\hat{\mu}=0,
    \end{aligned}
\end{equation}
from which we obtain $\hat{p}=p-M|\nabla n|^{2}/2+nM_{n}|\nabla n|^{2}/2-nT\nabla n \cdot\nabla(M/T)-Mn\Delta n$.

\subsection{Rate of entropy production in the bulk region}
\label{subsec:Rate of entropy production in the bulk region}

For two-dimensional non-isothermal compressible two-phase flows, the hydrodynamic equations are as follows: the number density $n$ obeys
\begin{equation}\label{eq:massbalance}
    \begin{aligned}
        \frac{\partial n}{\partial t}+\nabla\cdot(n\mathbf{v})+\nabla\cdot \mathbf{J}=0,
    \end{aligned}
\end{equation}
where $\mathbf{v}=\left(v_{x}, v_{z}\right)^{T}$ is the fluid velocity field, $v_{x}$ and $v_{z}$ are the velocities along $x$- and $z$-directions, respectively. 
$\mathbf{J}=-\mathcal{G}\nabla \left(\hat{\mu}/T\right)$ is the diffusion flux and $\mathcal{G}$ is the (phenomenological) mobility coefficient \cite{kou2018entropy}. 
The mass density is defined as $\rho=mn$, with $m$ being the particle mass.
The momentum density $\rho \mathbf{v}$ obeys 
\begin{equation}\label{eq:NS}
    \begin{aligned}
        \rho\left(\frac{\partial  \mathbf{v}}{\partial t}+ \mathbf{v}\cdot\nabla\mathbf{v}\right)+m\mathbf{J}\cdot\nabla\mathbf{v}=\nabla\cdot( -\Pi+\Theta)-\rho g \mathbf{z}.
    \end{aligned}
\end{equation}
Here, $g$ is the gravitational acceleration in the direction of $-\mathbf{z}$, $\Pi=\hat{p}\mathbf{I}+M\nabla n\otimes\nabla n$ is the reversible stress tensor, $\Theta=\eta D(\mathbf{v})+(\xi-2\eta/3)(\nabla\cdot\mathbf{v})\mathbf{I}$ is the viscous stress tensor \cite{onuki2007dynamic}, $D(\mathbf{v})=\nabla\mathbf{v}+\nabla\mathbf{v}^{T}$, $\eta$ and $\xi$ denote the shear and bulk viscosities, respectively. 
Moreover, including the kinetic part, the total energy density is defined as
\begin{equation}\label{eq:energydefine}
    \begin{aligned}
        e_{T}=\hat{e}+\frac{\rho}{2}\vert\mathbf{v}\vert^{2}.
    \end{aligned}
\end{equation}
Then $e_{T}$ is governed by \cite{landau1987fluid,liu2015liquid,kou2018entropy}
\begin{equation}\label{eq:energy}
    \begin{aligned}
        \frac{\partial e_{T}}{\partial t}+\nabla\cdot(e_{T}\mathbf{v})=-\nabla\cdot\left(\left(\Pi-\Theta\right)\cdot\mathbf{v}\right)-\nabla\cdot\mathbf{q}-\rho g v_{z},
    \end{aligned}
\end{equation}
with $\mathbf{q}=-\lambda\nabla T$ being the heat current density and $\lambda$ being the thermal conductivity. 
Using \eqref{eq:massbalance}, \eqref{eq:NS}, \eqref{eq:energydefine} and \eqref{eq:energy}, the equation for $\hat{e}$ can be derived as \cite{onuki2007dynamic, kou2018thermodynamically2}
\begin{equation}\label{eq:intenal}
    \begin{aligned}
        \frac{\partial \hat{e}}{\partial t}+\nabla\cdot(\hat{e}\mathbf{v})=-\Pi:\nabla \mathbf{v}+\dot{\epsilon}_{v}-\nabla\cdot\mathbf{q},
    \end{aligned}
\end{equation}
where $\dot{\epsilon}_{v}=\Theta:\nabla \mathbf{v}\geq 0$ is the rate of viscous heat production.
For simplicity of analysis and simulation, some assumptions are imposed on the parameters: 
(\romannumeral 1) the gravity is neglected, 
(\romannumeral 2) In \eqref{eq:feSdefine}, $C$ is a positive constant and $K$ vanishes, thus $M=CT$, 
(\romannumeral 3) the mobility coefficient $\mathcal{G}$ is a positive constant, and 
(\romannumeral 4) the shear viscosity, bulk viscosity, and thermal conductivity are proportional to $n$, with $\eta=\xi=\nu mn$ and $\lambda=\nu k_{B}n$, where $\nu$ is the kinematic viscosity independent of $n$ \cite{onuki2007dynamic}.

Based on the second law of thermodynamics (the positive definiteness of entropy production rate), we provide below a brief review of the derivation of the equation for $\hat{S}$ and the hydrodynamic boundary conditions \cite{xu2010contact,xu2012droplet}. 
Firstly, we introduce the following lemma from \cite{onuki2005dynamic}:
\begin{lemma}\label{lma:gradinet_p}
    The reversible stress tensor $\Pi$, internal energy $\hat{e}$, and chemical potential $\hat{\mu}$ satisfy the relation
    \begin{equation}\label{eq:gradinet_p}
    \begin{aligned}
        \nabla \cdot\frac{\Pi}{T}=n\nabla\frac{\hat{\mu}}{T}-\hat{e}\nabla \frac{1}{T},\quad\text{i.e.,}\quad\nabla\frac{\hat{p}}{T}+\nabla\cdot\left(C\nabla n \otimes \nabla n\right)=n\nabla\frac{\hat{\mu}}{T}-\hat{e}\nabla \frac{1}{T}.
    \end{aligned}
\end{equation}
\end{lemma}
\begin{proof}
Using \eqref{eq:threemaineq} and $\nabla \cdot(\nabla n\otimes \nabla n)=\nabla(\left\vert\nabla n\right\vert^{2})/2+\Delta n\nabla n$, we obtain \eqref{eq:gradinet_p}.
\end{proof}
\begin{remark}
    \cref{lma:gradinet_p} is a generalization of $d(p/T)=-ed(1/T)+nd(\mu/T)$ which follows from Gibbs-Duhem equation. 
    Moreover, in our model, the relation $\nabla p=n\nabla \mu+ns\nabla T$ can be derived and generalized to $\nabla\cdot\Pi=n\nabla\hat{\mu}+\hat{S}\nabla T+C|\nabla n|^{2}\nabla T$.
\end{remark}
From \eqref{eq:entropyvar}, the time derivative of $S_{b}$ becomes
\begin{equation}\label{eq:dSbdt}
    \begin{aligned}
       \frac{\partial S_{b}}{\partial t}= \int_{\Omega}\frac{\partial \hat{S}}{\partial t}d\mathbf{x}
       =\int_{\Omega} \left(\frac{1}{T}\frac{\partial \hat{e}}{\partial t} -\frac{\hat{\mu}}{T}\frac{\partial n}{\partial t} -\nabla\cdot\left(C\nabla n\frac{\partial n}{\partial t}\right)\right)d\mathbf{x}.
    \end{aligned}
\end{equation}
Substituting \eqref{eq:generalizedEuler}, \eqref{eq:massbalance}, \eqref{eq:intenal}, \eqref{eq:gradinet_p} into \eqref{eq:dSbdt}, $\hat{S}$ can be found to satisfy
\begin{align}
    \frac{\partial \hat{S}}{\partial t}
    =&-\nabla\cdot\left(\hat{S}\mathbf{v}\right)-\nabla\cdot\left(C\nabla n\left(\frac{\partial n}{\partial t}+\mathbf{v}\cdot\nabla n\right)+\frac{\mathbf{q}}{T}\right)+\frac{\dot{\epsilon}_{v}+\dot{\epsilon}_{\theta}}{T}+\frac{\hat{\mu}}{T}\nabla\cdot \mathbf{J}, \label{eq:entropydensitysatisfy}
\end{align}
in which $\dot{\epsilon}_{\theta}=\lambda(\nabla T)^{2}/T\geq 0$ is the rate of thermal heat production, and $(\dot{\epsilon}_{v}+\dot{\epsilon}_{\theta})/T\geq 0$ is the rate of entropy production per unit volume in the bulk.
Substituting \eqref{eq:massbalance} into \eqref{eq:entropydensitysatisfy}, we obtain the entropy density equation \cite{onuki2007dynamic}:
\begin{equation}\label{eq:entropydensity}
    \begin{aligned}
        \frac{\partial \hat{S}}{\partial t}+\nabla\cdot\left(\hat{S}\mathbf{v}\right)=\nabla\cdot\left(C\nabla n\left(n\nabla\cdot\mathbf{v}+\nabla\cdot \mathbf{J}\right)-\frac{\mathbf{q}}{T}\right)+\frac{\dot{\epsilon}_{v}+\dot{\epsilon}_{\theta}}{T}+\frac{\hat{\mu}}{T}\nabla\cdot \mathbf{J}.
    \end{aligned}
\end{equation}

\subsection{Derivation of the hydrodynamic boundary conditions}
\label{subsec:Derivation of the hydrodynamic boundary conditions}
In this paper, the solid substrates are assumed to be flat, rigid, impermeable, insoluble and non-isothermal (denoted by $\Gamma$).
The surface energy $E_{s}=\int_{\Gamma} e_{s}(n)dA$ and surface entropy $S_{s}=\int_{\Gamma}\sigma_s(n)dA$ are introduced, respectively, where the densities $e_{s}$ and $\sigma_{s}$ only depend on $n$ at the solid surface \cite{xu2012droplet}. 
Similar to the Helmholtz bulk free energy, we define the Helmholtz surface free energy and density as
\begin{equation}\label{eq:fsdefine}
    \begin{aligned}
         F_{s}=\int_{\Gamma} f_{s}(n,T)dA, \quad 
         f_{s}(n,T)=e_{s}-T\sigma_{s}.
    \end{aligned}
\end{equation}
By using \eqref{eq:fsdefine}, $\sigma_{s}=-(\partial f_{s}/\partial T)_{n}$ can be derived, and hence
\begin{align}
    df_{s}=-\sigma_{s}dT+\left(\frac{\partial f_{s}}{\partial n}\right)_{T}dn, \quad  
    d\sigma_{s}=\frac{1}{T}de_{s}-\frac{1}{T}\left(\frac{\partial f_{s}}{\partial n}\right)_{T}dn. \label{eq:Gibbs-typeequation}
\end{align}
Notice that the second equation in \eqref{eq:Gibbs-typeequation} yields a Gibbs-type equation on fluid-solid interface.
Similarly, minimizing the total Helmholtz free energy (defined as $F_{\text{tot}}=F_{b}+F_{s}$) with respect to $n$ at the surface and using \eqref{eq:entropyvar} yields the fluid-solid interface equilibrium condition:  $L:=(\partial f_{s}/\partial n)_{T}+CT\partial_{\gamma}n=0$ \cite{qian2006molecular}, in which $\bm{\gamma}$ is the outward normal unit vector, and $L$ is a quantity defined at the fluid-solid interface \cite{xu2010contact}.

The fluid-solid interfacial tension (denoted by $\gamma_{fs}$) satisfies a Euler-type equation: $e_{s}-T\sigma_{s}-\gamma_{fs}=0$, which leads to $\gamma_{fs}=f_{s}$ \cite{rowlinson2013molecular}.
It is assumed that the surface stress $\Xi$ and surface heat flux $q_{s}$ (tangent to the surface) exist at fluid-solid interface. 
The equation for internal energy balance at interface is given by \cite{xu2012droplet}
\begin{equation}\label{eq:ebalanceinterface}
    \begin{aligned}
        \frac{\partial e_{s}}{\partial t}+\partial_{\tau}\left(e_{s}v_{\tau}\right)=\partial_{\tau}\left(\Xi_{\gamma\tau}v_{\tau}\right)-\partial_{\tau}q_{s}+\mathbf{q}\cdot\bm{\gamma}-\mathbf{q}_{w}\cdot\bm{\gamma}-\left(-\Pi_{\gamma\tau}+\Theta_{\gamma\tau}\right)v_{\tau},
    \end{aligned}
\end{equation}
where $\Pi_{\gamma\tau}=CT\partial_{\gamma}n\partial_{\tau}n$, $\Theta_{\gamma\tau}=\eta \partial_{\gamma}v_{\tau}$, $v_{\tau}=\mathbf{v}\cdot\bm{\tau}$ and $\bm{\tau}$ denoting the direction tangent to the fluid-solid interface.
Here, $\Xi$ is assumed only has the reversible part (no surface viscosity) arising from the interfacial tension $\gamma_{fs}$, i.e., $\Xi_{\gamma\tau}=\gamma_{fs}=f_{s}$.
And $\mathbf{q}_{w}=-\lambda_{w}\nabla T_{w}$ is the heat flux in the solid substrate with $\lambda_{w}$ being the heat conductivity of solid and $T_{w}$ being the temperature of solid substrate.
The surface heat flux $q_{s}$ will be derived in \eqref{eq:qsdefine} through the second law of thermodynamics.
Substituting \eqref{eq:ebalanceinterface} into \eqref{eq:Gibbs-typeequation}, the balance equation for $\sigma_{s}$ can be derived as \cite{liu2012hydrodynamic}
\begin{align}
    \frac{\partial \sigma_{s}}{\partial t}+\partial_{\tau}\left(\sigma_{s}v_{\tau}\right)
    =&-\partial_{\tau}\frac{q_{s}}{T}-\frac{\mathbf{q}_{w}\cdot\bm{\gamma}}{T_{w}}+\Psi_{b}-\frac{1}{T^{2}}q_{s}\partial_{\tau}T-\left(\frac{1}{T}-\frac{1}{T_{w}}\right)\mathbf{q}_{w}\cdot\bm{\gamma}\label{eq:dsigmasdtfinal}\\
    &+\frac{1}{T}\left(\partial_{\tau}f_{s}+\Pi_{\gamma\tau}-\Theta_{\gamma\tau}\right)v_{\tau}-\frac{L}{T}\left(\frac{\partial n}{\partial t}+v_{\tau}\partial_{\tau}n\right). \nonumber
\end{align}
Here, $\Psi_{b}=\mathbf{q}\cdot\bm{\gamma}/T+C\partial_{\gamma}n(\partial n/\partial t+v_{\tau}\partial_{\tau}n)$ is total entropy flux from bulk fluid region in \eqref{eq:entropydensitysatisfy}, and the last four terms on the right-hand side of \eqref{eq:dsigmasdtfinal} are the rate of entropy production per unit area at fluid-solid interface, which must be positive definite (the second law of thermodynamics).
Thus, the formulation of $q_{s}$ and three constitutive relations can be derived as (without cross coupling term)
\begin{subequations}
    \begin{align}
        q_{s}&=-\lambda_{s}\partial_{\tau}T, \label{eq:qsdefine} \\ 
        \kappa\lambda_{w}\partial_{\gamma}T_{w}&=\frac{1}{T}-\frac{1}{T_{w}}, \label{eq:Twbcd} \\
        \beta v_{\tau}&=\partial_{\tau}f_{s}+\Pi_{\gamma\tau}-\Theta_{\gamma\tau}, \label{eq:GNBC} \\
        -\frac{L}{T}&=\alpha\left(\frac{\partial n}{\partial t}+v_{\tau}\partial_{\tau}n\right). \label{eq:RBC}
    \end{align}
\end{subequations}
Here, $\lambda_{s}\geq 0$ is the surface heat conductivity, $\kappa\geq0$ is the interfacial parameter related to the interfacial thermal resistance (i.e., the Kapitza resistance) \cite{barrat2003kapitza}, $\beta\geq 0$ is the slip coefficient, and $\alpha\geq 0$ is a damping relaxation coefficient. 
Substituting \eqref{eq:qsdefine} and \eqref{eq:Twbcd} into \eqref{eq:ebalanceinterface}, we obtain the boundary condition applicable to $T$ as
\begin{multline}\label{eq:Tbcd}
    -\lambda \partial_{\gamma}T+\kappa^{-1}\left(\frac{1}{T}-\frac{1}{T_{w}}\right) \\
    =\frac{\partial e_{s}}{\partial t}+\partial_{\tau}\left(T\sigma_{s}v_{\tau}\right)+v_{\tau}\partial_{\tau}f_{s}-\partial_{\tau}\left(\lambda_{s}\partial_{\tau}T\right)+v_{\tau}\left(-\partial_{\tau}f_{s}-\Pi_{\gamma\tau}+\Theta_{\gamma\tau}\right).
\end{multline}

When the solid wall is moving with velocity $v_{w}$, the slip velocity is defined by $v_{\tau}^{slip}=v_{\tau}-v_{w}$ and \eqref{eq:GNBC} becomes $\beta v_{\tau}^{slip}=\partial_{\tau}f_{s}+\Pi_{\gamma\tau}-\Theta_{\gamma\tau}$. 
Summing \eqref{eq:entropydensitysatisfy} and \eqref{eq:dsigmasdtfinal}, the rate of total entropy production (defined as $S_{\text{tot}}=S_{b}+S_{s}$) is given by
\begin{align}
    \frac{dS_{\text{tot}}}{dt}
    =&\int_{\Omega}\left(\frac{\dot{\epsilon}_{v}+\dot{\epsilon}_{\theta}}{T}+\mathcal{G}\left\vert\nabla\frac{\hat{\mu}}{T}\right\vert^{2}\right)d\mathbf{x} \label{eq:dStotdt}\\
    &+\int_{\Gamma}\left(\frac{\lambda_{s}}{T^{2}}\left(\partial_{\tau}T\right)^{2}+\kappa^{-1}\left(\frac{1}{T}-\frac{1}{T_{w}}\right)^{2}+\frac{\beta}{T}\left(v_{\tau}^{slip}\right)^{2}+\frac{1}{\alpha }\left(\frac{L}{T}\right)^{2}\right)dA \nonumber\\
    &+\int_{\Gamma} \frac{\beta}{T}v_{w}v_{\tau}^{slip}dA-\int_{\Gamma}\left(\partial_{\tau}\frac{q_{s}}{T}+\frac{1}{T_{w}}\mathbf{q}_{w}\cdot\bm{\gamma}+\partial_{\tau}\left(\sigma_{s}v_{\tau}\right)\right)dA,\nonumber
\end{align}
where boundary conditions $v_{\gamma} = 0$ and $\partial_{\gamma}(\hat{\mu}/T)=0$ are also assumed at the fluid-solid interface. 
In \eqref{eq:dStotdt}, the term $\beta v_{w}v_{\tau}^{slip}/T$ is the entropy done per unit time by the flow to the wall and the last term is the outgoing entropy flux.

Furthermore, some additional assumptions are imposed on the parameters: 
(\romannumeral 1) the damping coefficient $\alpha$ is a constant,
(\romannumeral 2) the coefficients at the solid surface are extrapolated via $\zeta(n)=(\zeta_{l}-\zeta_{g})(n-n_{g})/(n_{l}-n_{g})+\zeta_{g}$, $\zeta = \beta$, $\kappa$, $\lambda_{s}$,
where $\zeta_{g}$ and $\zeta_{l}$ are the values of a surface parameter in the homogeneous gas and liquid phases, respectively.
And $n_{g}$ and $n_{l}$ are the number densities of the homogeneous gas and liquid phases in gas-liquid coexistence equilibrium, respectively.
(\romannumeral 3) the surface energy density $e_{s}=0$ and the surface entropy density is defined as $\sigma_{s} = -a_{s}(n-n_{c})$, where $n_{c}$ is the critical density and $a_{s}$ is a constant, determined by the short-range fluid-solid interaction. 
Hence, the surface free energy density is given by \cite{briant2004lattice,xu2012droplet}
\begin{equation*}
    \begin{aligned}
        f_{s}=a_{s}T(n-n_{c}).
    \end{aligned}
\end{equation*}
It is confirmed that the static contact angle $\theta_{s}^{surf}$ at the fluid-solid interface (for the liquid phase) can be determined by $a_{s}$ \cite{xu2014single}.
Notably, $f_{s}$ can be defined in other forms, such as quadratic polynomial or trigonometric functions of $n$ and $T$ \cite{qian2006molecular,gao2012gradient,wang2022energy,wang2024cicp}. 
\begin{remark}\label{rmk:tem_dirichlet_bcd}
For the case that fluid-solid interface is isothermal \cite{xu2010contact}, with no temperature discontinuity across the fluid-solid interface, \eqref{eq:Tbcd} reduces to a Dirichlet boundary condition: $T=T_{w}$ with the solid temperature $T_{w}$ being fixed and homogeneous.
Meanwhile, \eqref{eq:GNBC} reduces to
\begin{equation}
    \beta v_{\tau}^{slip}=-\eta\partial_{\gamma}v_{\tau}+L\partial_{\tau}n=-\eta\partial_{\gamma}v_{\tau}+\left(\frac{\partial f_{s}}{\partial n}+CT\partial_{\gamma}n\right)\partial_{\tau}n,
\end{equation}
which is of the same form as the GNBC \cite{qian2003molecular,qian2006molecular}.
In numerical simulations, both isothermal and non-isothermal cases of the fluid-solid interface are considered. 
\end{remark}

\subsection{A new thermodynamically consistent mathematical model}
\label{subsec:A new thermodynamically consistent mathematical model}
As previously discussed, when the total energy equation is adopted as one of governing equations \cite{liu2015liquid,kou2018entropy}, the model consists of \eqref{eq:massbalance}, \eqref{eq:NS}, \eqref{eq:energy} and definition $e_{T}=\hat{e}+\rho\vert\mathbf{v}\vert^{2}/2$ with boundary conditions \eqref{eq:RBC}, \eqref{eq:GNBC} and \eqref{eq:Tbcd}, where
\begin{equation}\label{eq:edefine}
    \begin{aligned}
        \hat{e}(n, T)=\frac{3nk_{B}T}{2}-\epsilon v_{0}n^{2}.
    \end{aligned}
\end{equation}
Obviously, this model naturally satisfies the first law of thermodynamics, but it introduces difficulties in constructing efficient and thermodynamically consistent numerical schemes.
When the internal energy equation is used as one of the governing equations \cite{kou2018thermodynamically2}, i.e., with $\hat{e}$ as the intermediate variable, the model becomes \eqref{eq:massbalance}, \eqref{eq:NS}, \eqref{eq:intenal} and \eqref{eq:edefine} with \eqref{eq:RBC}, \eqref{eq:GNBC} and \eqref{eq:Tbcd}. 
However, using the internal energy equation leads to artificial parasitic flow \cite{teshigawara2008droplet}.
Alternatively, if the entropy equation is used instead of the internal energy equation \cite{onuki2005dynamic,onuki2007dynamic,xu2010contact}, with $\hat{S}$ as the intermediate variable, the system includes \eqref{eq:massbalance}, \eqref{eq:NS}, \eqref{eq:entropydensity} and 
\begin{equation}\label{eq:hatSdefine}
    \begin{aligned}
        \hat{S}(n, \nabla n, T)=ns-\frac{C}{2}\vert \nabla n\vert^{2}
        =nk_{B}\ln{\left(\left(\frac{k_{B}T}{\epsilon}\right)^{\frac{3}{2}}\left(\frac{1}{v_{0}n}-1\right)\right)}-\frac{C}{2}\vert \nabla n\vert^{2},
    \end{aligned}
\end{equation}
with \eqref{eq:RBC}, \eqref{eq:GNBC} and \eqref{eq:Tbcd}. 
Although this model effectively avoids the issue of artificial parasitic flow and can easily be proved to satisfy the second law of thermodynamics, the complexity of coupling between $\hat{S}$ and $T$ in \eqref{eq:hatSdefine} also presents challenges in designing numerical schemes.
To overcome these issues, we propose a new thermodynamically consistent model where the equations for $e_{T}$, $\hat{e}$ and $\hat{S}$ are replaced by the equation for $T$, which directly characterizes the evolution of the fluid temperature.

In the following, we provide two approaches to derive the temperature equation.
First, based on \eqref{eq:edefine}, regarding $E_{b}$ as a functional of $n$ and $T$, we have
\begin{equation*}
    \begin{aligned}
        \delta E_{b}
        =&\int_{\Omega} \frac{3nk_{B}}{2}\delta Td\mathbf{x} +\int_{\Omega}\left(\frac{3k_{B}T}{2}-2\epsilon v_{0}n\right)\delta n d\mathbf{x}.
    \end{aligned}
\end{equation*}
This implies that the time derivative of $\hat{e}$ becomes
\begin{equation}\label{eq:dedt1}
    \begin{aligned}
        \frac{\partial \hat{e}}{\partial t}
        =\frac{3nk_{B}}{2}\frac{\partial T}{\partial t} +\left(\frac{3k_{B}T}{2}-2\epsilon v_{0}n\right)\frac{\partial n}{\partial t}.
    \end{aligned}
\end{equation}
Substituting \eqref{eq:massbalance}, \eqref{eq:intenal} and \eqref{eq:edefine} into \eqref{eq:dedt1}, the time derivative of $T$ satisfies
\begin{multline}\label{eq:tempequation}
    \frac{3nk_{B}}{2}\left(\frac{\partial T}{\partial t}+\mathbf{v}\cdot\nabla T\right)
            +\epsilon v_{0}n^{2}\nabla\cdot \mathbf{v}\\
            =\left(\frac{3k_{B}T}{2}-2\epsilon v_{0}n\right)\nabla\cdot\mathbf{J}-\Pi:\nabla\mathbf{v}+\Theta:\nabla\mathbf{v}-\nabla\cdot\mathbf{q}
\end{multline}
Similarly, \eqref{eq:tempequation} can be obtained by considering the variation of $S_{b}$ as
\begin{align}
    \delta S_{b}
        =&\int_{\Omega} \frac{3nk_{B}}{2T}\delta Td\mathbf{x} +\int_{\Omega}\frac{1}{T}\left(\frac{3k_{B}T}{2}-2\epsilon v_{0}n-\hat{\mu}\right)\delta n d\mathbf{x}- \int_{\Omega} C\nabla\cdot\left(\nabla n\delta n\right) d\mathbf{x}.\nonumber
\end{align}
Then, the time derivative of $\hat{S}$ becomes
\begin{equation}\label{eq:dsdtT}
    \begin{aligned}
        \frac{\partial \hat{S}}{\partial t}=\frac{3nk_{B}}{2T}\frac{\partial T}{\partial t} +\left(\frac{3k_{B}}{2}-\frac{2\epsilon v_{0}n}{T}-\frac{\hat{\mu}}{T}\right)\frac{\partial n}{\partial t}- C\nabla\cdot\left(\nabla n\frac{\partial n}{\partial t}\right).
    \end{aligned}
\end{equation}
By substituting \eqref{eq:massbalance} and \eqref{eq:entropydensity} into \eqref{eq:dsdtT}, we have
\begin{align}
    \frac{3nk_{B}}{2T}\frac{\partial T}{\partial t}
    =&-\nabla\cdot\left(\hat{S}\mathbf{v}\right)+\nabla\cdot\left(C\left(n\nabla\cdot\mathbf{v}+\nabla\cdot\mathbf{J}\right)\nabla n\right)+ C\nabla\cdot\left(\nabla n\frac{\partial n}{\partial t}\right)\nonumber\\
    &+\frac{\Theta}{T}:\nabla\mathbf{v}-\frac{1}{T}\nabla\cdot\mathbf{q}+\frac{\hat{\mu}}{T}\nabla\cdot\mathbf{J}-\left(\frac{3k_{B}}{2}-\frac{2\epsilon v_{0}n}{T}-\frac{\hat{\mu}}{T}\right)\frac{\partial n}{\partial t}\nonumber\\
    =&- \frac{1}{T}\nabla\cdot\left(\hat{e}\mathbf{v}\right)-\frac{\Pi}{T}:\nabla\mathbf{v}-\left(\frac{3k_{B}}{2}-\frac{2\epsilon v_{0}n}{T}\right)\frac{\partial n}{\partial t}+\frac{\Theta}{T}:\nabla\mathbf{v}-\frac{1}{T}\nabla\cdot\mathbf{q}.\label{eq:partofderivationofT}
\end{align}
Notice that \eqref{eq:partofderivationofT} is equivalent to \eqref{eq:dedt1}.

Finally, we arrive at a new non-isothermal compressible two-phase flow model for MCL problems. 
The equations are \eqref{eq:massbalance}, \eqref{eq:NS}, \eqref{eq:tempequation} with hydrodynamic boundary conditions \eqref{eq:RBC}, \eqref{eq:GNBC} and \eqref{eq:Tbcd} at top and bottom boundaries, and the system is closed by applying the periodic boundary condition at left and right boundaries.
\begin{remark}\label{rmk:tempequation}
    By using \cref{lma:gradinet_p}, we can reformulate \eqref{eq:tempequation} as
    \begin{multline}\label{eq:tempequation_2}
        \frac{3nk_{B}}{2}\frac{\partial T}{\partial t}+\frac{\epsilon v_{0}n^{2}}{T}\nabla\cdot\left(T\mathbf{v}\right)-nT\nabla\frac{\hat{\mu}}{T}\cdot\mathbf{v}\\
        =\left(\frac{3k_{B}}{2}T-2\epsilon v_{0}n\right)\nabla\cdot\mathbf{J}-T\nabla\cdot\left(\frac{\Pi}{T}\cdot\mathbf{v}\right)+\Theta:\nabla\mathbf{v}-\nabla\cdot\mathbf{q}.
    \end{multline}
    Actually, although \eqref{eq:tempequation_2} is equivalent to \eqref{eq:tempequation} for the time-continuous case, it avoids the difficulties in designing numerical schemes that are decoupled, linear and satisfy the temporally discrete second law of thermodynamics.
    Consequently, we adopt \eqref{eq:tempequation_2} as the governing equation for $T$ in the \cref{subsec:A decoupled linear and unconditionally entropy-stable scheme}.
\end{remark}

\section{The dimensionless model and thermodynamical consistency}
\label{sec:The dimensionless model and thermodynamical consistency}

\subsection{Dimensionless equations}
\label{subsec:Dimensionless equations}
To deduce the dimensionless form for the new system, we scale the length by $l=(C/2\varepsilon k_{B}v_{0})^{1/2}$, $\mathbf{v}$ by a constant speed $U_{0}$, $t$ by $l/U_{0}$, $n$ by $1/v_{0}$, $\hat{\mu}$ by $1/\epsilon$, $-\Pi+\Theta$ by $\epsilon/v_{0}$, $\hat{S}$ by $k_{B}/v_{0}$, $L$ by $\epsilon l$, and $T$ by $\epsilon/k_{B}$.
Using the same symbols for the dimensionless variables, the system becomes:
\begin{subequations}\label{eq:dimensionless_equations}
    \begin{align}
            &\frac{\partial n}{\partial t}+\nabla\cdot(n\mathbf{v})=\mathcal{L}_{d}\Delta\frac{\hat{\mu}}{T}, \label{eq:dimensionless_mass} \\
            &\mathcal{R}_{e}\mathcal{R}n\left(\frac{\partial  \mathbf{v}}{\partial t}+\mathbf{v}\cdot\nabla\mathbf{v}\right)-\mathcal{B}\nabla \frac{\hat{\mu}}{T}\cdot\nabla\mathbf{v}=\nabla\cdot\left(-\Pi+\mathcal{R}\Theta\right), \label{eq:dimensionless_NS} \\
            &\frac{3n}{2}\left(\frac{\partial T}{\partial t}+\mathbf{v}\cdot\nabla T\right)
        +n^{2}\nabla\cdot \mathbf{v}+\mathcal{L}_{d}\left(\frac{3}{2}T-2n\right)\Delta \frac{\hat{\mu}}{T} \label{eq:dimensionless_Temp} \\
        &=-\Pi:\nabla\mathbf{v}+\mathcal{R}\Theta:\nabla\mathbf{v}+\mathcal{R}_{e}^{-1}\nabla\cdot\left(n\nabla T\right), \nonumber
    \end{align}
\end{subequations}
in which $\hat{\mu}=3T/2+T/(1-n)-3T\ln{T}/2-T\ln{(1/n-1)}-2n-2\varepsilon T\Delta n$, $\Pi=\hat{p}\mathbf{I}+2\varepsilon T\nabla n\otimes\nabla n$, $\Theta=n(D(\mathbf{v})+(\nabla\cdot\mathbf{v})\mathbf{I}/3)$, $\hat{p}=nT/(1-n)-n^{2}-\varepsilon T\vert\nabla n\vert^{2}-2\varepsilon nT\Delta n$, and $\varepsilon$ is a dimensionless parameter characterizing the gas-liquid interfacial thickness. 
The dimensionless hydrodynamic boundary conditions are
\begin{subequations}\label{eq:dimensionless_bcd}
    \begin{align}
        &\frac{\partial n}{\partial t}+v_{\tau}\partial_{\tau}n=-\mathcal{L}_{b}\frac{L}{T},\quad
        \partial_{\gamma}\frac{\hat{\mu}}{T}=0,\label{eq:dimensionless_rbc} \\
        &\beta v_{\tau}^{slip}=\mathcal{L}_{s}\mathcal{R}^{-1}\left(-\mathcal{R}\Theta_{\gamma\tau}+\mathcal{W}\partial_{\tau} f_{s}+\Pi_{\gamma\tau}\right), \quad v_{\gamma}=0, \label{eq:dimensionless_gnbc} \\
        &-n\partial_{\gamma}T+\mathcal{L}_{\kappa}^{-1}\kappa^{-1}\left(\frac{1}{T}-\frac{1}{T_{w}}\right)+\mathcal{R}_{e}\mathcal{W}f_{s}\partial_{\tau}v_{\tau}+\mathcal{L}_{\lambda}\partial_{\tau}\left(\lambda_{s}\partial_{\tau}T\right) \label{eq:dimensionless_Tbcd}\\
        &=\mathcal{R}_{e}v_{\tau}\left(\mathcal{R}\Theta_{\gamma\tau}-\mathcal{W}\partial_{\tau} f_{s}-\Pi_{\gamma\tau}\right), \nonumber
    \end{align}
\end{subequations}
where $L=\mathcal{W}\partial f_{s}/\partial n+2\varepsilon T\partial_{\gamma}n$, $f_{s}=T(n-n_{c})$, $\partial_{\tau} f_{s}=T\partial_{\tau}n+(n-n_{c})\partial_{\tau}T$, $v_{\tau}^{slip}=v_{\tau}-v_{w}$, and $v_{w}$ is the wall speed.
The dimensionless surface coefficients in \eqref{eq:dimensionless_bcd} are $\zeta=(1-\zeta_{gl})(n-n_{g})/(n_{l}-n_{g})+\zeta_{gl}$, $\zeta=\beta$, $\kappa$, $\lambda_{s}$, where $\zeta_{gl}=\zeta_{g}/\zeta_{l}$.
Periodic boundary conditions are imposed at the left and right boundaries.
Other dimensionless parameters are defined as $\mathcal{L}_{d}=\mathcal{G}v_{0}k_{B}/U_{0}l$, $\mathcal{R}_{e}=U_{0}l/\nu$, $\mathcal{R}=U_{0} m\nu /\epsilon l$, $\mathcal{B}=mv_{0}\mathcal{G}U_{0}k_{B}/\epsilon l=\mathcal{L}_{d}\mathcal{R}_{e}\mathcal{R}$, $\mathcal{W}=a_{s}/k_{B}l$, $\mathcal{L}_{b}=k_{B}v_{0}l^{2}/\alpha U_{0}$, $\mathcal{L}_{s}=m\nu /\beta_{l}v_{0}l$, $\mathcal{L}_{\kappa}=\nu \kappa_{l}\epsilon^{2}/v_{0}k_{B}l$, $\mathcal{L}_{\lambda}=\lambda_{sl}v_{0}/\nu k_{B}l$. 
Furthermore, the dimensionless internal energy density $\hat{e}$ and entropy density $\hat{S}$ are given by
\begin{equation}\label{eq:dimensionless_Se}
    \begin{aligned}
        \hat{e}=\frac{3}{2}nT-n^{2}, \quad 
        \hat{S}=n\ln{\left(T^{\frac{3}{2}}\frac{1-n}{n}\right)}-\varepsilon \vert\nabla n\vert^{2}.
    \end{aligned}
\end{equation}

\subsection{Thermodynamical consistency} 
\label{subsec:Thermodynamical consistency}
We first prove that the above dimensionless model satisfies the first law of thermodynamics.
We define the dimensionless total energy $E_{\text{tot}}$, kinetic energy $E_{k}$ and internal energy $E_{b}$ over the domain as
\begin{align}
    E_{\text{tot}}=E_{k}+E_{b}, \quad 
    E_{\text{tot}}=\int_{\Omega} e_{T}d\mathbf{x}, \quad
    E_{k}=\mathcal{R}_{e}\mathcal{R}\int_{\Omega}\frac{1}{2}n|\mathbf{v}|^{2}d\mathbf{x}, \quad 
    E_{b}=\int_{\Omega}\hat{e}d\mathbf{x}. \nonumber
\end{align}
\begin{theorem}\label{thm:firstlawcontinue}
The system \eqref{eq:dimensionless_equations}-\eqref{eq:dimensionless_bcd} satisfies first law of thermodynamics as
\begin{equation*}
    \begin{aligned}
        \frac{\partial E_{\text{tot}}}{\partial t}=-\mathcal{R}_{e}^{-1}\int_{\Gamma} \left(\partial_{\tau}q_{s}+\mathbf{q}_{w}\cdot\bm{\gamma}+\mathcal{R}_{e}\mathcal{W}\partial_{\tau}\left(T\sigma_{s}v_{\tau}\right)\right)dA,
    \end{aligned}
\end{equation*}
where $\sigma_{s}=-(n-n_{c})$, $q_{s}=-\mathcal{L}_{\lambda}\lambda_{s}\partial_{\tau}T$ and $\mathbf{q}_{w}\cdot\bm{\gamma}=-\mathcal{L}_{\kappa}^{-1}\kappa^{-1}(1/T-1/T_{w})$.
\end{theorem}
\begin{proof}
From \eqref{eq:dimensionless_mass}, \eqref{eq:dimensionless_NS} and relation $\mathcal{B}=\mathcal{L}_{d}\mathcal{R}_{e}\mathcal{R}$, we have
\begin{align}
    \frac{\partial E_{k}}{\partial t}
    =&\mathcal{R}_{e}\mathcal{R}\int_{\Omega}\left( n\mathbf{v}\cdot\frac{\partial \mathbf{v}}{\partial t}+\frac{1}{2}\frac{\partial n}{\partial t}|\mathbf{v}|^{2}\right)d\mathbf{x}\nonumber\\
    =&\int_{\Omega}\left( \nabla\mathbf{v}:\Pi-\mathcal{R}\nabla\mathbf{v}:\Theta\right)d\mathbf{x}+\int_{\Gamma} v_{\tau}\left(-\Pi_{\gamma \tau}+\mathcal{R}\Theta_{\gamma \tau}\right)dA.\label{eq:partialEkpartialt}
\end{align}
Using \eqref{eq:gradinet_p}, \eqref{eq:dimensionless_mass}, \eqref{eq:dimensionless_Temp} and \eqref{eq:dimensionless_Se}, we have
\begin{align}
    \frac{\partial E_{b}}{\partial t}
    =&\int_{\Omega}\left( \frac{3n}{2}\frac{\partial T}{\partial t}+\left(\frac{3T}{2}-2n\right)\frac{\partial n}{\partial t}\right)d\mathbf{x}\nonumber\\
    =&\int_{\Omega} \left(-\nabla\cdot(\hat{e}\mathbf{v})-\Pi:\nabla \mathbf{v}+\mathcal{R}\Theta:\nabla \mathbf{v}\right)d\mathbf{x}+\mathcal{R}_{e}^{-1}\int_{\Gamma} n\partial_{\gamma} T dA,\label{eq:partialEbpartialt}
\end{align}
in which we use the relation
\begin{align}
    \nabla\cdot\left(\hat{e}\mathbf{v}\right)
    =&\left(\frac{3}{2}T-2n\right)\nabla\cdot\left(n\mathbf{v}\right)+n^{2}\nabla\cdot\mathbf{v}+\frac{3}{2}n\nabla T\cdot\mathbf{v}.\label{eq:haterelation}
\end{align}
Combining \eqref{eq:dimensionless_Tbcd}, \eqref{eq:partialEkpartialt} and \eqref{eq:partialEbpartialt} yields
\begin{equation*}
    \begin{aligned}
        \frac{\partial E_{\text{tot}}}{\partial t} 
        =&\int_{\Gamma} \left(\mathcal{W}v_{\tau}\partial_{\tau}f_{s}+\mathcal{R}_{e}^{-1}\mathcal{L}_{\kappa}^{-1}\kappa^{-1}\left(\frac{1}{T}-\frac{1}{T_{w}}\right)+\mathcal{W}f_{s}\partial_{\tau}v_{\tau}+\frac{\mathcal{L}_{\lambda}}{\mathcal{R}_{e}}\partial_{\tau}\left(\lambda_{s}\partial_{\tau}T\right)\right)dA\\
        =&-\mathcal{R}_{e}^{-1}\int_{\Gamma} \left(\partial_{\tau}q_{s}+\mathbf{q}_{w}\cdot\bm{\gamma}+\mathcal{R}_{e}\mathcal{W}\partial_{\tau}\left(T\sigma_{s}v_{\tau}\right)\right)dA.
    \end{aligned}
\end{equation*}
This ends the proof.
\end{proof}
To prove that the proposed dimensionless model satisfies the second law of thermodynamics, for the notations, we use $\Vert \cdot\Vert$ to represent the $L^{2}(\Omega)$ norm and the subscript $\Vert \cdot\Vert_{\Gamma}$ denotes the norm at the top and bottom boundaries.
By using the definition of the total entropy $S_{\text{tot}}$ in \cref{subsec:Derivation of the hydrodynamic boundary conditions}, we have
\begin{align}
    S_{\text{tot}}=S_{b}+S_{s}, \quad 
    S_{b}=\int_{\Omega} \hat{S}d\mathbf{x}, \quad
    S_{s}=\int_{\Gamma} \mathcal{W}\sigma_{s}dA.\nonumber
\end{align}
\begin{theorem}\label{thm:secondlawcontinue}
The system \eqref{eq:dimensionless_equations}-\eqref{eq:dimensionless_bcd} satisfies second law of thermodynamics as
\begin{align}
    &\frac{\partial S_{\text{tot}}}{\partial t}
    =\frac{\mathcal{R}}{\mathcal{L}_{s}}\int_{\Gamma}\frac{\beta}{T} v_{w}v_{\tau}^{slip}dA-\int_{\Gamma} \left(\mathcal{R}_{e}^{-1}\partial_{\tau}\frac{q_{s}}{T}+\frac{\mathbf{q}_{w}\cdot\bm{\gamma}}{\mathcal{R}_{e}T_{w}}+\mathcal{W}\partial_{\tau}\left(\sigma_{s}v_{\tau}\right)\right)dA \nonumber \\
    &+\mathcal{R}_{e}^{-1}\left\Vert\frac{n^{\frac{1}{2}}\nabla T}{T}\right\Vert^{2}+\frac{\mathcal{R}}{2}\left\Vert \left(\frac{n}{T}\right)^{\frac{1}{2}}D(\mathbf{v})\right\Vert^{2}+\frac{\mathcal{R}}{3}\left\Vert \left(\frac{n}{T}\right)^{\frac{1}{2}}\nabla\cdot\mathbf{v}\right\Vert^{2}+\mathcal{L}_{d} \left\Vert\nabla\frac{\hat{\mu}}{T}\right\Vert^{2}\nonumber\\
    &+\mathcal{L}_{b}\left\Vert\frac{L}{T}\right\Vert^{2}_{\Gamma}+\frac{\mathcal{R}}{\mathcal{L}_{s}}\left\Vert \left(\frac{\beta}{T}\right)^{\frac{1}{2}}v_{\tau}^{slip}\right\Vert^{2}_{\Gamma}+\mathcal{R}_{e}^{-1}\mathcal{L}_{\kappa}^{-1}\left\Vert \kappa^{-\frac{1}{2}}\left(\frac{1}{T}-\frac{1}{T_{w}}\right)\right\Vert^{2}_{\Gamma}+\frac{\mathcal{L}_{\lambda}}{\mathcal{R}_{e}}\left\Vert \lambda_{s}^{\frac{1}{2}}\frac{\partial_{\tau}T}{T}\right\Vert^{2}_{\Gamma}.\nonumber
\end{align}
Here, the first term on the right-hand side is the entropy production per unit time due to the flow to the wall, and the second term is the outgoing entropy flux per unit time.
\end{theorem}
\begin{proof}
Using \cref{lma:gradinet_p}, \eqref{eq:dimensionless_mass}, \eqref{eq:dimensionless_Temp}, \eqref{eq:dimensionless_Se} and \eqref{eq:haterelation}, we have
\begin{align}
    \frac{\partial S_{b}}{\partial t}
    =&\int_{\Omega} \left(\frac{3n}{2T}\frac{\partial T}{\partial t}+\left( \ln{\left(T^{\frac{3}{2}}\frac{1-n}{n}\right)}-\frac{1}{1-n}+2\varepsilon\Delta n\right)\frac{\partial n}{\partial t} - 2\varepsilon \nabla\cdot\left(\nabla n\frac{\partial n}{\partial t}\right)\right)d\mathbf{x} \nonumber\\
    =&\int_{\Omega} \left(-\nabla\cdot\left(\frac{\hat{e}}{T}\mathbf{v}\right)+\hat{e}\nabla\frac{1}{T}\cdot\mathbf{v}+\frac{\hat{\mu}}{T}\nabla\cdot(n\mathbf{v})\right)d\mathbf{x} \nonumber\\
    &+\int_{\Omega}\left(-\frac{\Pi}{T}:\nabla\mathbf{v}+2\varepsilon\nabla\cdot\left(\mathbf{v}\cdot\nabla n\otimes\nabla n\right)-2\varepsilon\nabla\cdot\left(\left(\frac{\partial n}{\partial t}+\mathbf{v}\cdot\nabla n\right)\nabla n\right)\right)d\mathbf{x} \nonumber\\
    &+\int_{\Omega} \left(\mathcal{R}\frac{\Theta}{T}:\nabla\mathbf{v}-\mathcal{L}_{d}\frac{\hat{\mu}}{T}\Delta\frac{\hat{\mu}}{T}+\mathcal{R}_{e}^{-1}\frac{\nabla\cdot\left(n\nabla T\right)}{T}\right)d\mathbf{x} \nonumber\\
    =&\int_{\Omega}-\nabla\cdot\left(\hat{S}\mathbf{v}\right)d\mathbf{x}+\mathcal{R}_{e}^{-1}\left\Vert\frac{n^{\frac{1}{2}}\nabla T}{T}\right\Vert^{2}+\mathcal{L}_{d} \left\Vert\nabla\frac{\hat{\mu}}{T}\right\Vert^{2}+\frac{\mathcal{R}}{2}\left\Vert \left(\frac{n}{T}\right)^{\frac{1}{2}}D(\mathbf{v})\right\Vert^{2} \nonumber\\
    &+\frac{\mathcal{R}}{3}\left\Vert \left(\frac{n}{T}\right)^{\frac{1}{2}}\nabla\cdot\mathbf{v}\right\Vert^{2}+\int_{\Gamma}\left(\mathcal{R}_{e}^{-1}\frac{n\partial_{\gamma} T}{T}-2\varepsilon\partial_{\gamma}n\left(\frac{\partial n}{\partial t}+v_{\tau}\partial_{\tau}n\right)\right)dA. \nonumber
\end{align}
Using \eqref{eq:dimensionless_bcd}, the last term on the right-hand side can be reformulated as 
\begin{align}
    rhs=&\int_{\Gamma} \frac{1}{T}\left(\frac{1}{\mathcal{R}_{e}\mathcal{L}_{\kappa}\kappa}\left(\frac{1}{T}-\frac{1}{T_{w}}\right)+\frac{\mathcal{L}_{\lambda}}{\mathcal{R}_{e}}\partial_{\tau}\left(\lambda_{s}\partial_{\tau}T\right)+\mathcal{W}f_{s}\partial_{\tau}v_{\tau}+\frac{\mathcal{R}}{\mathcal{L}_{s}}\beta v_{w}v_{\tau}^{slip}\right)dA \nonumber\\
    &-\int_{\Gamma}\frac{1}{T}\left(L-\mathcal{W}\frac{\partial f_{s}}{\partial n}\right)\left(\frac{\partial n}{\partial t}+v_{\tau}\partial_{\tau}n\right)dA+\frac{\mathcal{R}}{\mathcal{L}_{s}}\left\Vert \left(\frac{\beta}{T}\right)^{\frac{1}{2}}v_{\tau}^{slip}\right\Vert^{2}_{\Gamma} \nonumber\\
    =&\frac{\mathcal{R}}{\mathcal{L}_{s}}\int_{\Gamma} \frac{\beta}{T} v_{w}v_{\tau}^{slip}dA-\mathcal{R}_{e}^{-1}\int_{\Gamma} \left(\partial_{\tau}\frac{q_{s}}{T}+\frac{1}{T_{w}}\mathbf{q}_{w}\cdot\bm{\gamma}\right)dA\label{eq:partialSbpartialt_2}\\
    &-\int_{\Gamma} \mathcal{W}\left(\partial_{\tau}\left(\sigma_{s}v_{\tau}\right)+\frac{\partial \sigma_{s}}{\partial t}\right)dA+\frac{\mathcal{R}}{\mathcal{L}_{s}}\left\Vert \left(\frac{\beta}{T}\right)^{\frac{1}{2}}v_{\tau}^{slip}\right\Vert^{2}_{\Gamma}+\mathcal{L}_{b}\left\Vert\frac{L}{T}\right\Vert^{2}_{\Gamma}\nonumber \\
    &+\mathcal{R}_{e}^{-1}\mathcal{L}_{\kappa}^{-1}\left\Vert \kappa^{-\frac{1}{2}}\left(\frac{1}{T}-\frac{1}{T_{w}}\right)\right\Vert^{2}_{\Gamma}+\frac{\mathcal{L}_{\lambda}}{\mathcal{R}_{e}}\left\Vert \lambda_{s}^{\frac{1}{2}}\frac{\partial_{\tau}T}{T}\right\Vert^{2}_{\Gamma}.\nonumber
\end{align}
Substituting \eqref{eq:partialSbpartialt_2} into $\partial S_{b}/\partial t$ yields the desired result.
\end{proof}

\section{Thermodynamically consistent numerical schemes}
\label{sec:Thermodynamically consistent numerical schemes}
In this section, we develop efficient and thermodynamically consistent numerical schemes for the system in \cref{subsec:Dimensionless equations}.
The time interval $[0, T_{f}]$ with final time $T_{f}>0$ is uniformly divided into $M_{t}$ time steps, and the time step size $\delta t:=T_{f}/M_{t}$.
We denote $q^{k}$ or $\mathbf{q}^{k}$ as the approximation of a scalar function $q(t)$ or a vector function $\mathbf{q}(t)$ at time $t^{k}:=k\delta t$. 

\subsection{A fully coupled and thermodynamically consistent scheme}
\label{subsec:A fully coupled and thermodynamically consistent scheme}
To construct a scheme that strictly satisfies the temporally discrete laws of thermodynamics, we propose following semi-implicit scheme: find $(n^{k+1}, \mathbf{v}^{k+1}, T^{k+1})$ for $k \geq 0$ by
\begin{subequations}\label{eq:discret_eq}
    \begin{align}
        &\frac{n^{k+1}-n^{k}}{\delta t}+\nabla\cdot\left(n^{k+1}\mathbf{v}^{k+1}\right)=\mathcal{L}_{d}\Delta\frac{\hat{\mu}^{k+1}}{T^{k+1}}, \label{eq:discret_mass} \\
        &\mathcal{R}_{e}\mathcal{R}\left(n^{k}\frac{\mathbf{v}^{k+1}-\mathbf{v}^{k}}{\delta t}+n^{k+1}\mathbf{v}^{k+1}\cdot\nabla\mathbf{v}^{k+1}\right)-\mathcal{B}\nabla \frac{\hat{\mu}^{k+1}}{T^{k+1}}\cdot\nabla\mathbf{v}^{k+1} \label{eq:discret_ns}\\
        &=\nabla\cdot\left(-\Pi^{k+1}+\mathcal{R}\Theta^{k+1}\right)-\frac{\left(n^{k+1}-n^{k}\right)^{2}}{\delta t \mathbf{v}^{k+1}}, \nonumber \\
        &\frac{3}{2}\left(n^{k}\frac{T^{k+1}-T^{k}}{\delta t}+n^{k+1}\mathbf{v}^{k+1}\cdot\nabla T^{k+1}\right) +\left(n^{k+1}\right)^{2}\nabla\cdot \mathbf{v}^{k+1} \label{eq:discret_temp}\\
        &+\mathcal{L}_{d}\left(\frac{3}{2}T^{k+1}-2n^{k+1}\right)\Delta \frac{\hat{\mu}^{k+1}}{T^{k+1}}=-\Pi^{k+1}:\nabla\mathbf{v}^{k+1}+\mathcal{R}\Theta^{k+1}:\nabla\mathbf{v}^{k+1} \nonumber\\
        &+\mathcal{R}_{e}^{-1}\nabla\cdot\left(n^{k+1}\nabla T^{k+1}\right)+\frac{\mathcal{R}_{e}\mathcal{R}n^{k}}{2\delta t}\left|\mathbf{v}^{k+1}-\mathbf{v}^{k}\right|^{2}, \nonumber
    \end{align}
\end{subequations}
where $\hat{\mu}^{k+1}=3T^{k+1}/2+T^{k+1}/(1-n^{k+1})-3T^{k+1}\ln{T^{k+1}}/2-T^{k+1}\ln{(1/n^{k+1}-1)}-2n^{k+1}-2\varepsilon T^{k+1}\Delta n^{k+1}$, $\Pi^{k+1}=\hat{p}^{k+1}\mathbf{I}+2\varepsilon T^{k+1}\nabla n^{k+1}\otimes\nabla n^{k+1}$, $\Theta^{k+1}=n^{k+1}(D(\mathbf{v}^{k+1})+(\nabla\cdot\mathbf{v}^{k+1})\mathbf{I}/3)$, and $\hat{p}^{k+1}=n^{k+1}T^{k+1}/(1-n^{k+1})-(n^{k+1})^{2}-\varepsilon T^{k+1}\vert\nabla n^{k+1}\vert^{2}-2\varepsilon n^{k+1}T^{k+1}\Delta n^{k+1}$.
Here, $(n^{k+1}-n^{k})^{2}/\delta t \mathbf{v}^{k+1}$ and $\mathcal{R}_{e}\mathcal{R}n^{k}\vert\mathbf{v}^{k+1}-\mathbf{v}^{k}\vert^{2}/2\delta t$ are stabilization terms to ensure the scheme satisfy temporally discrete first law of thermodynamics.
The discrete hydrodynamic boundary conditions are
\begin{subequations}\label{eq:discret_bcd}
    \begin{align}
        &\frac{n^{k+1}-n^{k}}{\delta t}+v_{\tau}^{k+1}\partial_{\tau}n^{k+1}=-\mathcal{L}_{b}\frac{L^{k+1}}{T^{k+1}}, \quad \partial_{\gamma}\frac{\hat{\mu}^{k+1}}{T^{k+1}}=0, \label{eq:discret_rbc}\\
        &\beta^{k+1}v_{\tau}^{slip,k+1}=\mathcal{L}_{s}\mathcal{R}^{-1}\left(-\mathcal{R}\Theta_{\gamma\tau}^{k+1}+\mathcal{W}\partial_{\tau} f_{s}^{k+1}+\Pi_{\gamma\tau}^{k+1}\right),\quad v_{\gamma}^{k+1}=0 \label{eq:discret_gnbc}\\
        &-n^{k+1}\partial_{\gamma}T^{k+1}+\frac{1}{\mathcal{L}_{\kappa}\kappa^{k+1}}\left(\frac{1}{T^{k+1}}-\frac{1}{T_{w}}\right)+\mathcal{R}_{e}\mathcal{W}f_{s}^{k+1}\partial_{\tau}v_{\tau}^{k+1} \label{eq:discret_tbcd}\\
        &=-\mathcal{L}_{\lambda}\partial_{\tau}\left(\lambda_{s}^{k+1}\partial_{\tau}T^{k+1}\right)+\mathcal{R}_{e}v_{\tau}^{k+1}\left(\mathcal{R}\Theta_{\gamma\tau}^{k+1}-\mathcal{W}\partial_{\tau} f_{s}^{k+1}-\Pi_{\gamma\tau}^{k+1}\right),\nonumber
    \end{align}
\end{subequations}
in which $L^{k+1}=\mathcal{W}(\partial f_{s}/\partial n)^{k+1}+2\varepsilon T^{k+1}\partial_{\gamma}n^{k+1}$, $f_{s}^{k+1}=T^{k+1}(n^{k+1}-n_{c})$, $\partial_{\tau} f_{s}^{k+1}=T^{k+1}\partial_{\tau}n^{k+1}+(n^{k+1}-n_{c})\partial_{\tau}T^{k+1}$, $v_{\tau}^{slip,k+1}=v_{\tau}^{k+1}-v_{w}$, and $\zeta^{k+1}=(1-\zeta_{gl})(n^{k+1}-n_{g})/(n_{l}-n_{g})+\zeta_{gl}, \ \zeta=\beta, \ \kappa,\ \lambda_{s}$.
The discrete internal energy density and entropy density are defined as
\begin{equation}\label{eq:discret_hate_hatS}
    \begin{aligned}
        \hat{e}^{k}=\frac{3}{2}n^{k}T^{k}-\left(n^{k}\right)^{2}, \quad 
        \hat{S}^{k}=n^{k}\ln{\left(\left(T^{k}\right)^{\frac{3}{2}}\frac{1-n^{k}}{n^{k}}\right)}-\varepsilon\left\vert\nabla n^{k}\right\vert^{2}.
    \end{aligned}
\end{equation}
To demonstrate that the above scheme obeys the temporally discrete laws of thermodynamics, we firstly define the formulations of total energy, kinetic energy, and internal energy over the domain at the time $t^{k}$ as
\begin{align}
    E_{\text{tot}}^{k}=E_{k}^{k}+E_{b}^{k}, \ 
    E_{\text{tot}}^{k}=\int_{\Omega} e_{T}^{k}d\mathbf{x}, \
    E_{k}^{k}=\mathcal{R}_{e}\mathcal{R}\int_{\Omega}\frac{1}{2}n^{k}|\mathbf{v}^{k}|^{2}d\mathbf{x}, \ 
    E_{b}^{k}=\int_{\Omega}\hat{e}^{k}d\mathbf{x}. \label{eq:totaldiscreteenergydefine}
\end{align}
\begin{theorem}
    The scheme \eqref{eq:discret_eq}-\eqref{eq:discret_bcd} satisfies the temporally discrete first law of thermodynamics as
    \begin{align}
       \frac{E_{\text{tot}}^{k+1}-E_{\text{tot}}^{k}}{\delta t}=-\mathcal{R}_{e}^{-1}\int_{\Gamma} \left(\partial_{\tau}q_{s}^{k+1}+\mathbf{q}_{w}^{k+1}\cdot\bm{\gamma}+\mathcal{R}_{e}\mathcal{W}\partial_{\tau}\left(T^{k+1}\sigma_{s}^{k+1}v_{\tau}^{k+1}\right)\right)dA, \nonumber
    \end{align}
    where $q_{s}^{k+1}=-\mathcal{L}_{\lambda}\lambda_{s}^{k+1}\partial_{\tau}T^{k+1}$ and $\mathbf{q}_{w}^{k+1}\cdot\bm{\gamma}=-(\mathcal{L}_{\kappa}\kappa^{k+1})^{-1}(1/T^{k+1}-1/T_{w})$.
\end{theorem}
\begin{proof}
    Using \eqref{eq:discret_mass}, \eqref{eq:discret_ns} and $\mathcal{B}=\mathcal{L}_{d}\mathcal{R}_{e}\mathcal{R}$, we obtain
\begin{align}
    E_{k}^{k+1}-E_{k}^{k}
    =&\mathcal{R}_{e}\mathcal{R}\int_{\Omega}\frac{1}{2}\left(n^{k}\left\vert\mathbf{v}^{k+1}\right\vert^{2}-n^{k}\left\vert\mathbf{v}^{k}\right\vert^{2} +\left(n^{k+1}-n^{k}\right)\left\vert\mathbf{v}^{k+1}\right\vert^{2}\right)d\mathbf{x} \nonumber\\
    =&\int_{\Omega}\left(\delta t \Pi^{k+1}:\nabla\mathbf{v}^{k+1}-\delta t\mathcal{R}\Theta^{k+1}:\nabla\mathbf{v}^{k+1}-\left(n^{k+1}-n^{k}\right)^{2}\right.\label{eq:E_K^k+1-E_K^k}\\
    &\left.-\frac{\mathcal{R}_{e}\mathcal{R}}{2}n^{k}\left\vert\mathbf{v}^{k+1}-\mathbf{v}^{k}\right\vert^{2}\right)d\mathbf{x} + \delta t\int_{\Gamma} v_{\tau}^{k+1}\left( -\Pi^{k+1}_{\gamma\tau}+\mathcal{R}\Theta^{k+1}_{\gamma\tau}\right) dA.\nonumber
\end{align}
From \eqref{eq:haterelation}, \eqref{eq:discret_mass}, \eqref{eq:discret_temp} and \eqref{eq:discret_hate_hatS}, we have
\begin{align}
    E_{b}^{k+1}-E_{b}^{k}
    =&\int_{\Omega}\left(\frac{3}{2}n^{k}\left(T^{k+1}-T^{k}\right)+\frac{3}{2}T^{k+1}\left(n^{k+1}-n^{k}\right)\right)d\mathbf{x} \nonumber\\
    &+\int_{\Omega}\left(-2n^{k+1}\left(n^{k+1}-n^{k}\right)+\left(n^{k+1}-n^{k}\right)^{2}\right)d\mathbf{x} \nonumber \\
    =&\delta t\int_{\Omega}\left(-\nabla\cdot\left(\hat{e}^{k+1}\mathbf{v}^{k+1}\right)-\Pi^{k+1}:\nabla\mathbf{v}^{k+1}+\mathcal{R}\Theta^{k+1}:\nabla\mathbf{v}^{k+1}\right)d\mathbf{x} \label{eq:E_b^k+1-E_b^k}\\
    &+\int_{\Omega}\left(\frac{\mathcal{R}_{e}\mathcal{R}}{2}n^{k}\left\vert\mathbf{v}^{k+1}-\mathbf{v}^{k}\right\vert^{2}+\left(n^{k+1}-n^{k}\right)^{2}\right)d\mathbf{x} \nonumber\\
    &+\delta t\mathcal{R}_{e}^{-1}\int_{\Gamma} n^{k+1}\partial_{\gamma} T^{k+1} dA. \nonumber
\end{align}
    Combining \eqref{eq:discret_tbcd}, \eqref{eq:E_K^k+1-E_K^k} and \eqref{eq:E_b^k+1-E_b^k} yields the result.
\end{proof}
Next, we define the formulations of total entropy, bulk entropy and surface entropy over the domain at time $t^{k}$ as
\begin{equation}\label{eq:totaldiscreteentropydefine}
    \begin{aligned}
        S_{\text{tot}}^{k}=S_{b}^{k}+S_{s}^{k}, \quad
        S_{b}^{k}=\int_{\Omega} \hat{S}^{k}d\mathbf{x}, \quad 
        S_{s}^{k}=\int_{\Gamma}\mathcal{W} \sigma_{s}^{k}dA.
    \end{aligned}
\end{equation}
\begin{theorem}\label{thm:entropyoriginal}
    The scheme \eqref{eq:discret_eq}-\eqref{eq:discret_bcd} satisfies the temporally discrete second law of thermodynamics as
    \begin{align}
       &\frac{ S_{\text{tot}}^{k+1}-S_{\text{tot}}^{k}}{\delta t}
        \geq\frac{\mathcal{R}}{\mathcal{L}_{s}}\int_{\Gamma}\frac{\beta^{k+1}}{T^{k+1}} v_{w}v_{\tau}^{slip,k+1}dA \nonumber\\
    &-\int_{\Gamma}\left(\mathcal{R}_{e}^{-1}\partial_{\tau}\frac{q_{s}^{k+1}}{T^{k+1}}+\frac{\mathbf{q}_{w}^{k+1}\cdot\bm{\gamma}}{\mathcal{R}_{e}T_{w}}+\mathcal{W}\partial_{\tau}\left(\sigma_{s}^{k+1}v_{\tau}^{k+1}\right)\right)dA +\mathcal{R}_{e}^{-1}\left\Vert\frac{\left(n^{k+1}\right)^{\frac{1}{2}}\nabla T^{k+1}}{T^{k+1}}\right\Vert^{2}\nonumber\\
    &+\frac{\mathcal{R}_{e}\mathcal{R}}{2\delta t}\left\Vert\left(\frac{n^{k}}{T^{k+1}}\right)^{\frac{1}{2}}\left(\mathbf{v}^{k+1}-\mathbf{v}^{k}\right)\right\Vert^{2}+\mathcal{L}_{d} \left\Vert\nabla\frac{\hat{\mu}^{k+1}}{T^{k+1}}\right\Vert^{2} +\frac{\mathcal{R}}{2}\left\Vert \left(\frac{n^{k+1}}{T^{k+1}}\right)^{\frac{1}{2}}D(\mathbf{v}^{k+1})\right\Vert^{2}\nonumber\\
    &+\frac{\mathcal{R}}{3}\left\Vert \left(\frac{n^{k+1}}{T^{k+1}}\right)^{\frac{1}{2}}\nabla\cdot\mathbf{v}^{k+1}\right\Vert^{2} +\mathcal{L}_{b}\left\Vert\frac{L^{k+1}}{T^{k+1}} \right\Vert^{2}_{\Gamma}+\frac{\mathcal{R}}{\mathcal{L}_{s}}\left\Vert \left(\frac{\beta^{k+1}}{T^{k+1}}\right)^{\frac{1}{2}}v_{\tau}^{slip,k+1}\right\Vert^{2}_{\Gamma}\nonumber\\
    & +\mathcal{R}_{e}^{-1}\mathcal{L}_{\kappa}^{-1}\left\Vert \left(\frac{1}{\kappa^{k+1}}\right)^{\frac{1}{2}}\left(\frac{1}{T^{k+1}}-\frac{1}{T_{w}}\right)\right\Vert^{2}_{\Gamma}+\frac{\mathcal{L}_{\lambda}}{\mathcal{R}_{e}}\left\Vert \left(\lambda_{s}^{k+1}\right)^{\frac{1}{2}}\frac{\partial_{\tau}T^{k+1}}{T^{k+1}}\right\Vert^{2}_{\Gamma}.\nonumber
    \end{align}
\end{theorem}
\begin{proof}
Using \eqref{eq:discret_hate_hatS} and the concavity and convexity of $\hat{S}$, we obtain
\begin{align}
    S_{b}^{k+1}-S_{b}^{k}
    =&\int_{\Omega}\left(\frac{3}{2}n^{k}\left(\ln{T^{k+1}}-\ln{T^{k}}\right)+\frac{3}{2}\left(n^{k+1}-n^{k}\right)\ln{T^{k+1}}\right)d\mathbf{x} \nonumber\\
    &+\int_{\Omega} \left(n^{k+1}\ln{\left(\frac{1-n^{k+1}}{n^{k+1}}\right)}-n^{k}\ln{\left(\frac{1-n^{k}}{n^{k}}\right)}\right)d\mathbf{x} \nonumber\\
    &-\int_{\Omega} \varepsilon\left(2\nabla n^{k+1}\cdot\nabla\left( n^{k+1}-n^{k}\right)-\left\vert\nabla\left( n^{k+1}-n^{k}\right)\right\vert^{2}\right)d\mathbf{x} \nonumber\\
    \geq&\int_{\Omega}\left(\frac{3}{2}\frac{n^{k}}{T^{k+1}}\left(T^{k+1}-T^{k}\right)+ \left(\frac{3}{2}-\frac{2n^{k+1}}{T^{k+1}}-\frac{\hat{\mu}^{k+1}}{T^{k+1}}\right)\left( n^{k+1}-n^{k}\right)\right)d\mathbf{x} \label{eq:Sbk+1-Sbk}\\
    &-\int_{\Gamma} 2\varepsilon\partial_{\gamma} n^{k+1}\left( n^{k+1}-n^{k}\right)dA. \nonumber
\end{align}
Similar to \cref{thm:secondlawcontinue}, substituting \eqref{eq:discret_mass} and \eqref{eq:discret_temp} into \eqref{eq:Sbk+1-Sbk} yields
\begin{align}
    &\frac{S_{b}^{k+1}-S_{b}^{k}}{\delta t}
    \geq-\int_{\Omega}\nabla\cdot\left(\hat{S}^{k+1}\mathbf{v}^{k+1}\right)d\mathbf{x}+\frac{\mathcal{R}_{e}\mathcal{R}}{2\delta t}\left\Vert\left(\frac{n^{k}}{T^{k+1}}\right)^{\frac{1}{2}}\left(\mathbf{v}^{k+1}-\mathbf{v}^{k}\right)\right\Vert^{2} \nonumber\\
    &+\frac{1}{\mathcal{R}_{e}}\left\Vert\frac{\left(n^{k+1}\right)^{\frac{1}{2}}\nabla T^{k+1}}{T^{k+1}}\right\Vert^{2}+\frac{\mathcal{R}}{2}\left\Vert \left(\frac{n^{k+1}}{T^{k+1}}\right)^{\frac{1}{2}}D\left(\mathbf{v}^{k+1}\right)\right\Vert^{2}+\frac{\mathcal{R}}{3}\left\Vert \left(\frac{n^{k+1}}{T^{k+1}}\right)^{\frac{1}{2}}\nabla\cdot\mathbf{v}^{k+1}\right\Vert^{2} \nonumber\\
    &+\mathcal{L}_{d}\left\Vert\nabla \frac{\hat{\mu}^{k+1}}{T^{k+1}}\right\Vert^{2}+\int_{\Gamma}\left(\frac{n^{k+1}\partial_{\gamma} T^{k+1}}{\mathcal{R}_{e}T^{k+1}}- 2\varepsilon\partial_{\gamma} n^{k+1}\left(\frac{n^{k+1}-n^{k}}{\delta t}+\partial_{\tau} n^{k+1} v_{\tau}^{k+1}\right)\right)dA.\nonumber
\end{align}
Using \eqref{eq:discret_bcd}, the last term on the right-hand side can be derived as
    \begin{align}
        rhs=&\int_{\Gamma}\left( \frac{\mathcal{R}}{\mathcal{L}_{s}}\frac{\beta^{k+1}}{T^{k+1}} v_{w}v_{\tau}^{slip,k+1}dA- \mathcal{W}\partial_{\tau}\left(\sigma_{s}^{k+1}v_{\tau}^{k+1}\right)-\mathcal{R}_{e}^{-1}\partial_{\tau}\frac{q_{s}^{k+1}}{T^{k+1}}-\frac{\mathbf{q}_{w}^{k+1}\cdot\bm{\gamma}}{\mathcal{R}_{e}T_{w}}\right)dA\nonumber\\
        &-\int_{\Gamma}\mathcal{W}\frac{\sigma_{s}^{k+1}-\sigma_{s}^{k}}{\delta t}dA+\mathcal{L}_{b}\left\Vert\frac{L^{k+1}}{T^{k+1}} \right\Vert^{2}_{\Gamma}+\frac{\mathcal{R}}{\mathcal{L}_{s}}\left\Vert \left(\frac{\beta^{k+1}}{T^{k+1}}\right)^{\frac{1}{2}}v_{\tau}^{slip,k+1}\right\Vert^{2}_{\Gamma} \nonumber\\
        &+\mathcal{R}_{e}^{-1}\mathcal{L}_{\kappa}^{-1}\left\Vert \left(\frac{1}{\kappa^{k+1}}\right)^{\frac{1}{2}}\left(\frac{1}{T^{k+1}}-\frac{1}{T_{w}}\right)\right\Vert^{2}_{\Gamma}+\frac{\mathcal{L}_{\lambda}}{\mathcal{R}_{e}}\left\Vert \left(\lambda_{s}^{k+1}\right)^{\frac{1}{2}}\frac{\partial_{\tau}T^{k+1}}{T^{k+1}}\right\Vert^{2}_{\Gamma}.\nonumber
    \end{align}
Substituting this term into $S_{b}^{k+1}-S_{b}^{k}$ leads to the result.
\end{proof}

\subsection{A decoupled, linear, and unconditionally entropy-stable scheme}
\label{subsec:A decoupled linear and unconditionally entropy-stable scheme}
Notice that the above scheme suffers from nonlinearity and coupling between variables, which lead to high computational costs. 
As noted in \cref{rmk:tempequation}, the tight coupling can be avoided by designing the scheme to satisfy only the second law of thermodynamics.
To address the nonlinearity, we first separate $\hat{S}$ into three parts:
\begin{equation*}
    \begin{aligned}
        \hat{S} =n\ln{\left(T^{\frac{3}{2}}\frac{1-n}{n}\right)}-\varepsilon \vert\nabla n\vert^{2}=\frac{3}{2}n\ln{T}+n\ln{\left(\frac{1-n}{n}\right)}-\varepsilon \vert\nabla n\vert^{2}:=\sum_{i=1}^{3}\hat{S}_{i},
    \end{aligned}
\end{equation*}
where $\hat{S}_{1}=3n\ln{T}/2$, $\hat{S}_{2}=n\ln{(1/n-1)}$, and $\hat{S}_{3}=-\varepsilon \vert\nabla n\vert^{2}$. 
Thus, the bulk entropy consists of three components: $S_{b}=\sum_{i=1}^{3}\int_{\Omega}\hat{S}_{i}d\mathbf{x}:=\sum_{i=1}^{3}S_{bi}$.
The main challenge in linearizing \eqref{eq:dimensionless_equations}-\eqref{eq:dimensionless_bcd} lies in term $\delta S_{b2} / \delta n$.
However, both the stabilized approach and the convex splitting approach cannot really overcome this.
Furthermore, when $f_{s}$ is defined as a nonlinear function of $n$ and $T$, a stabilization term is also required in these approaches.
Hence, we introduce an entropy-stable MSAV approach \cite{cheng2018multiple,liu2024efficient}.

For the system in \cref{subsec:Dimensionless equations}, we define two SAVs for the part of bulk entropy and surface entropy separately: $r(t)=\sqrt{S_{b2}+C_{r}}$ and $q(t)=\sqrt{S_{s}+C_{q}}$ \cite{wang2022energy,liu2024efficient}, where constants $C_{r}$ and $C_{q}$ are chosen such that $S_{b2}+C_{r},S_{s}+C_{q}>0$.
By the definitions of two SAVs, \eqref{eq:dimensionless_mass} and chemical potential can be reformulated as
\begin{subequations}\label{eq:dimensionless_sav1}
\begin{align}
    &\frac{\partial n}{\partial t}+\nabla\cdot(n\mathbf{v})=\mathcal{L}_{d}\Delta\tilde{\mu},\\
    &\tilde{\mu}=\frac{3}{2}-\frac{3}{2}\ln{T}-\frac{2n}{T}+\frac{r}{\sqrt{S_{b2}+C_{r}}}\left(\frac{1}{1-n}-\ln{\left(\frac{1-n}{n}\right)}\right)-2\varepsilon \Delta n, \label{eq:dimensionless_mu_sav}\\
    &\frac{dr}{dt}=\frac{1}{2\sqrt{S_{b2}+C_{r}}}\int_{\Omega} \left(\ln{\left(\frac{1-n}{n}\right)}-\frac{1}{1-n}\right)\frac{\partial n}{\partial t}d\mathbf{x}.
\end{align}
\end{subequations}
Meanwhile, \eqref{eq:dimensionless_rbc} can be expanded as
\begin{subequations}\label{eq:dimensionless_sav2}
\begin{align}
    & \frac{\partial n}{\partial t}+v_{\tau}\partial_{\tau}n=-\mathcal{L}_{b}\tilde{L}, \ \partial_{\gamma}\tilde{\mu}=0, \\
    & \tilde{L}=\frac{q}{\sqrt{S_{s}+C_{q}}}\mathcal{W}+2\varepsilon \partial_{\gamma}n,\\
    &\frac{dq}{dt}=\frac{1}{2\sqrt{S_{s}+C_{q}}}\int_{\Gamma} -\mathcal{W}\frac{\partial n}{\partial t}dA.
\end{align}
\end{subequations}
Thus, a decoupled, linear, and unconditionally entropy-stable scheme can be constructed as follows: for $k \geq 0$, find $(n^{k+1}, r^{k+1}, q^{k+1})$ by
\begin{subequations}\label{eq:discret_nrq_dc}
    \begin{align}
        &\frac{n^{k+1}-n^{k}}{\delta t}+\nabla\cdot\left(n^{k}\mathbf{v}^{k}\right)=\mathcal{L}_{d}\Delta \tilde{\mu}^{k+1},\label{eq:discret_mass_dc}\\
        &\tilde{\mu}^{k+1}
        =\tilde{\mu}_{b1}^{k}+\frac{r^{k+1}+r^{k}}{2\sqrt{S_{b2}^{k}+C_{r}}}\tilde{\mu}_{b2}^{k}-2\varepsilon \Delta n^{k+1}, \label{eq:discret_mu_dc}\\
        &r^{k+1}-r^{k}=\frac{1}{2\sqrt{S_{b2}^{k}+C_{r}}}\int_{\Omega} -\tilde{\mu}_{b2}^{k}\left(n^{k+1} - n^{k}\right)d\mathbf{x},\label{eq:discret_sav1_dc}\\
        &\frac{n^{k+1}-n^{k}}{\delta t}+v_{\tau}^{k}\partial_{\tau}n^{k+1}=-\mathcal{L}_{b}\tilde{L}^{k+1}, \quad \partial_{\gamma}\tilde{\mu}^{k+1}=0,\label{eq:discret_rbc_dc}\\
        &\tilde{L}^{k+1}=\frac{q^{k+1}+q^{k}}{2\sqrt{S_{s}^{k}+C_{q}}}\mathcal{W}+2\varepsilon \partial_{\gamma}n^{k+1},\label{eq:discret_L_dc}\\
        &q^{k+1}-q^{k}=\frac{1}{2\sqrt{S_{s}^{k}+C_{q}}}\int_{\Gamma} -\mathcal{W}\left(n^{k+1} - n^{k}\right)dA.\label{eq:discret_sav2_dc}
    \end{align}
\end{subequations}
Here, $\tilde{\mu}_{b1}^{k}:=3(1-\ln{T^{k}})/2-2n^{k}/T^{k}$ and $\tilde{\mu}_{b2}^{k}:=1/(1-n^{k})-\ln{(1/n^{k}-1)}$.
Then, compute $\mathbf{v}^{k+1}$ by
\begin{subequations}\label{eq:discret_v_dc}
    \begin{align}
        &\mathcal{R}_{e}\mathcal{R}n^{k}\left(\frac{\mathbf{v}^{k+1}-\mathbf{v}^{k}}{\delta t}+\mathbf{v}^{k}\cdot\nabla\mathbf{v}^{k+1}\right)-\mathcal{B}\nabla \tilde{\mu}^{k+1}\cdot\nabla\mathbf{v}^{k+1}\label{eq:discret_ns_dc}\\
        &=\nabla\cdot\left(-\Pi^{k+1}+\mathcal{R}\Theta^{k+1}\right),\nonumber \\
        &\beta^{k+1}v_{\tau}^{slip,k+1}=\frac{\mathcal{L}_{s}}{\mathcal{R}}\left(-\mathcal{R}\Theta_{\gamma\tau}^{k+1}+T^{k}\tilde{L}^{k+1}\partial_{\tau}n^{k+1}+\mathcal{W}\left(n^{k+1}-n_{c}\right)\partial_{\tau}T^{k}\right),\label{eq:discret_gnbc_dc}\\
        &v_{\gamma}^{k+1}=0, \nonumber
    \end{align}
\end{subequations}
where $\Pi^{k+1}=\hat{p}^{k+1}\mathbf{I}+2\varepsilon T^{k}\nabla n^{k+1}\otimes\nabla n^{k+1}$, $\Theta^{k+1}=n^{k+1}(D(\mathbf{v}^{k+1})+(\nabla\cdot\mathbf{v}^{k+1})\mathbf{I}/3)$ and $v_{\tau}^{slip,k+1}=v_{\tau}^{k+1}-v_{w}$.
Find $T^{k+1}$ by
\begin{subequations}\label{eq:discret_t_dc}
    \begin{align}
        &\frac{3n^{k+1}}{2}\frac{T^{k+1}-T^{k}}{\delta t}
        +T^{k+1}\left(\frac{n^{k}}{T^{k}}\right)^{2}\nabla\cdot \left(T^{k}\mathbf{v}^{k}\right)\label{eq:discret_temp_dc}\\
        &+\mathcal{L}_{d}T^{k+1}\left(\frac{3}{2}-\frac{2n^{k}}{T^{k}}\right)\Delta  \tilde{\mu}^{k+1}=n^{k}T^{k+1}\nabla\tilde{\mu}^{k+1}\cdot\mathbf{v}^{k}\nonumber\\
        &-T^{k+1}\nabla\cdot\left(\frac{\Pi^{k+1}}{T^{k}}\cdot\mathbf{v}^{k}\right)+\mathcal{R}\Theta^{k+1}:\nabla\mathbf{v}^{k+1}+\mathcal{R}_{e}^{-1}\nabla\cdot\left(n^{k+1}\nabla T^{k+1}\right),\nonumber\\
        &-n^{k+1}\partial_{\gamma}T^{k+1}+\frac{T^{k+1}}{\mathcal{L}_{\kappa}\kappa^{k}T^{k}}\left(\frac{1}{T^{k}}-\frac{1}{T_{w}}\right) \label{eq:discret_tbcd_dc}\\
        &+\mathcal{R}_{e}q^{k\ast}\mathcal{W}T^{k+1}\left(n^{k+1}-n_{c}\right)\partial_{\tau}v_{\tau}^{k}+\mathcal{L}_{\lambda}\partial_{\tau}\left(\lambda_{s}^{k+1}\partial_{\tau}T^{k+1}\right)\nonumber\\
        &=\mathcal{R}_{e}v_{\tau}^{k+1}\left(\mathcal{R}\Theta_{\gamma\tau}^{k+1}-T^{k}\tilde{L}^{k+1}\partial_{\tau}n^{k+1}-\mathcal{W}\left(n^{k+1}-n_{c}\right)\partial_{\tau}T^{k}\right),\nonumber
    \end{align}
\end{subequations}
where $q^{k\ast}:=(q^{k+1}+q^{k})/2\sqrt{S_{s}^{k}+C_{q}}$ and $\zeta^{k+1}=(1-\zeta_{gl})(n^{k+1}-n_{g})/(n_{l}-n_{g})+\zeta_{gl}, \ \zeta=\beta, \ \kappa,\ \lambda_{s}$.
After Step 1, we update $\hat{p}$ by the generalized Euler equation \eqref{eq:generalizedEuler}, which is discretized as
    \begin{align}
        \hat{p}^{k+1}&=n^{k}T^{k}\tilde{\mu}^{k+1}-\hat{e}^{k}+T^{k}\tilde{S}^{k+1} \nonumber\\
        &=\frac{n^{k}T^{k}}{1-n^{k}}-\left(n^{k}\right)^{2}-\varepsilon T^{k}\left\vert\nabla n^{k+1}\right\vert^{2}-2\varepsilon n^{k}T^{k}\Delta n^{k+1},\label{eq:pressure_dc}
    \end{align}
where $\tilde{S}^{k+1}$ is the SAV approximation of $\hat{S}^{k+1}$ and is only used to update $\hat{p}$:
\begin{equation*}
    \begin{aligned} 
        \tilde{S}^{k+1}=\left(1-r^{k\ast}\right)\frac{n^{k}}{1-n^{k}}+\frac{3}{2}n^{k}\ln{T^{k}}+r^{k\ast}n^{k}\ln{\left(\frac{1-n^{k}}{n^{k}}\right)}-\varepsilon \left\vert\nabla n^{k+1}\right\vert^{2},
    \end{aligned}
\end{equation*}
where $r^{k\ast}:=(r^{k+1}+r^{k})/2\sqrt{S_{b2}^{k}+C_{r}}$. 
Then, we define the modified total entropy (denoted by $S_{\text{sav}}$) at the time $t^{k}$ as
\begin{align}
    S_{\text{sav}}^{k}=S_{b1}^{k}+\left\vert r^{k}\right\vert^{2}-C_{r}+S_{b3}^{k}+\left\vert q^{k}\right\vert^{2}-C_{q}, \label{eq:totaldiscreteentropydefinesav}
\end{align}
where $S_{b1}^{k}=\int_{\Omega}3n^{k}\ln{T^{k}}/2 d\mathbf{x}$ and $ S_{b3}^{k}=\int_{\Omega}-\varepsilon \vert\nabla n^{k}\vert^{2}d\mathbf{x}$.
\begin{theorem}\label{thm:secondlaw_dc}
    The scheme \eqref{eq:discret_nrq_dc}-\eqref{eq:discret_t_dc} satisfies the temporally discrete second law of thermodynamics as  
    \begin{align}
        &\frac{S_{\text{sav}}^{k+1}-S_{\text{sav}}^{k}}{\delta t}
        \geq\frac{\mathcal{R}}{\mathcal{L}_{s}}\int_{\Gamma}\frac{\beta^{k+1}}{T^{k+1}} v_{w}v_{\tau}^{slip,k+1}dA-\int_{\Gamma}\left( \mathcal{R}_{e}^{-1}\partial_{\tau}\frac{q_{s}^{k+1}}{T^{k+1}}+\frac{\mathbf{q}_{w}^{k+1}\cdot\bm{\gamma}}{\mathcal{R}_{e}T_{w}}\right. \nonumber\\
        &\left. +q^{k\ast}\mathcal{W}\partial_{\tau}\left(\sigma_{s}^{k+1}v_{\tau}^{k}\right)\right)dA+\mathcal{R}_{e}^{-1}\left\Vert\frac{\left(n^{k+1}\right)^{\frac{1}{2}}\nabla T^{k+1}}{T^{k+1}}\right\Vert^{2}+\frac{\mathcal{R}}{2}\left\Vert \left(\frac{n^{k+1}}{T^{k+1}}\right)^{\frac{1}{2}}D(\mathbf{v}^{k+1})\right\Vert^{2} \nonumber\\
        &+\frac{\mathcal{R}}{3}\left\Vert \left(\frac{n^{k+1}}{T^{k+1}}\right)^{\frac{1}{2}}\nabla\cdot\mathbf{v}^{k+1}\right\Vert^{2}+\mathcal{L}_{d} \left\Vert\nabla\tilde{\mu}^{k+1}\right\Vert^{2} +\frac{\mathcal{R}}{\mathcal{L}_{s}}\left\Vert \left(\frac{\beta^{k+1}}{T^{k+1}}\right)^{\frac{1}{2}}v_{\tau}^{slip,k+1}\right\Vert^{2}_{\Gamma}\nonumber\\
        &+\mathcal{L}_{b}\left\Vert \tilde{L}^{k+1}\right\Vert^{2}_{\Gamma}+\mathcal{R}_{e}^{-1}\mathcal{L}_{\kappa}^{-1}\left\Vert \left(\frac{1}{\kappa^{k}}\right)^{\frac{1}{2}}\left(\frac{1}{T^{k}}-\frac{1}{T_{w}}\right)\right\Vert^{2}_{\Gamma}+\frac{\mathcal{L}_{\lambda}}{\mathcal{R}_{e}}\left\Vert \left(\lambda_{s}^{k+1}\right)^{\frac{1}{2}}\frac{\partial_{\tau}T^{k+1}}{T^{k+1}}\right\Vert^{2}_{\Gamma}, \nonumber
    \end{align}
    where $q_{s}^{k+1}=-\mathcal{L}_{\lambda}\lambda_{s}^{k+1}\partial_{\tau}T^{k+1}$ and $\mathbf{q}_{w}^{k}\cdot\bm{\gamma}=-(\mathcal{L}_{\kappa}\kappa^{k})^{-1}(1/T^{k}-1/T_{w})$.
\end{theorem}
\begin{proof}
    From \eqref{eq:discret_mu_dc}, \eqref{eq:discret_sav1_dc}, \eqref{eq:discret_L_dc} and \eqref{eq:discret_sav2_dc}, we have
    \begin{align}
        \left\vert r^{k+1}\right\vert^{2}-\left\vert r^{k}\right\vert^{2}
        =&\int_{\Omega}\left(\tilde{\mu}_{b1}^{k}-2\varepsilon \Delta n^{k+1}-\Tilde{\mu}^{k+1}\right)\left(n^{k+1}-n^{k}\right)d\mathbf{x}, \nonumber\\
        \left\vert q^{k+1}\right\vert^{2}-\left\vert q^{k}\right\vert^{2}
        =&\int_{\Gamma} \left(2\varepsilon\partial_{\gamma} n^{k+1}-\tilde{L}^{k+1}\right)\left( n^{k+1}-n^{k}\right)dA. \nonumber
    \end{align}
    Using the similar derivation in \cref{thm:entropyoriginal}, we can obtain
    \begin{align}
       S_{\text{sav}}^{k+1}-S_{\text{sav}}^{k}\geq
       &\int_{\Omega}\left(\frac{3}{2}-\frac{2n^{k}}{T^{k}}-\Tilde{\mu}^{k+1}\right)\left(n^{k+1}-n^{k}\right)d\mathbf{x} \label{eq:Ssavb1k+1-Ssavb1k}\\
       &+\int_{\Omega}\frac{3}{2}\frac{n^{k+1}}{T^{k+1}}\left(T^{k+1}-T^{k}\right)d\mathbf{x}-\int_{\Gamma} \tilde{L}^{k+1}\left( n^{k+1}-n^{k}\right)dA.\nonumber
    \end{align}
    Substituting \eqref{eq:discret_mass_dc} and \eqref{eq:discret_temp_dc} into \eqref{eq:Ssavb1k+1-Ssavb1k} and using boundary conditions \eqref{eq:discret_rbc_dc}, \eqref{eq:discret_gnbc_dc} and \eqref{eq:discret_tbcd_dc}, we can obtain the desired result.
\end{proof}
\begin{remark}
    (\romannumeral 1) As mentioned previously, since the discrete \cref{lma:gradinet_p} cannot be satisfied, the scheme \eqref{eq:discret_nrq_dc}-\eqref{eq:discret_t_dc} cannot satisfy the temporally discrete first law of thermodynamics.
    Consequently, the variables $n^{k+1}$, $\mathbf{v}^{k+1}$ and $T^{k+1}$ can be easily decoupled in this scheme.
    (\romannumeral 2) To obtain unconditionally original entropy-stable numerical schemes, as opposed to the modified entropy in this scheme, the Lagrange multiplier approach \cite{cheng2020new,wang2024cicp} can be extended in our model.
    However, this approach requires solving a nonlinear algebraic equation for the Lagrange multiplier which leads to additional computational costs and theoretical difficulties for its analysis \cite{cheng2024unique}.
\end{remark}
\textbf{The fully discrete scheme:}
The resulting scheme is discretized by a finite element method in space.
We introduce a discontinuous Galerkin (DG) method for the number density equation \cite{kay2009discontinuous,chen2024energy}, while the remaining two equations are discretized using continuous Galerkin (CG) approximations.
Let $\mathcal{T}_{h}$ denote a tessellation of $\Omega$ with quadrilateral elements $K$ and mesh size $h$, and let $\mathcal{E}_{h}$ be the set of all edges of $\mathcal{T}_{h}$ distinguishing between the interior edges $\mathcal{E}_{h}^{I}$ and the solid boundary edges $\Gamma_{h}$.
For two neighboring elements $K_{1}$, $K_{2}\in\mathcal{T}_{h}$, let $E=\partial K_{1}\cap \partial K_{2}$ be a shared edge with the unit normal vector $\bm{\gamma}_{E}$ directed outward from $K_{1}$.
We then define the average $\{\cdot\}$ and jump $[\cdot]$ operators as $\{\phi\}=\left((\phi\vert_{K_{1}})\vert_{E}+(\phi\vert_{K_{2}})\vert_{E}\right)/2$ and $[\phi]=(\phi\vert_{K_{1}})\vert_{E}-(\phi\vert_{K_{2}})\vert_{E}$, respectively.
Let $\mathbb{Q}^{p}(K)$ be the space of polynomials of degree $\leq p$ in each space direction on $K$.
The finite element spaces are denoted by
    \begin{align}
        \Phi_{h}&=\left\lbrace \phi\in L^{2}\left(\Omega\right): \phi\vert_{K}\in \mathbb{Q}^{p}(K),\ \forall K\in \mathcal{T}_{h}\right\rbrace, \nonumber\\
        W_{h}&=\left\lbrace w\in C^{0}\left(\Omega\right): w\vert_{K}\in \mathbb{Q}^{p}(K),\ \forall K\in \mathcal{T}_{h}\right\rbrace, \nonumber\\
        \mathbf{U}_{h}&=\left\lbrace\mathbf{u}\in \mathbf{C}^{0}\left(\Omega\right): \mathbf{u}\vert_{K}\in \left(\mathbb{Q}^{p}(K)\right)^{2},\ \forall K\in \mathcal{T}_{h}\right\rbrace. \nonumber
    \end{align}
Note that functions in $\Phi_{h}$ are allowed to be completely discontinuous across element interfaces.
We denote by $(\cdot,\cdot)_{K}$ the $L^{2}(K)$ inner product and by $\langle\cdot,\cdot\rangle_{E}$ the $L^{2}(E)$ inner product.
We first decouple the computation of two SAVs $r^{k+1}$ and $q^{k+1}$ as
\begin{align}
    \phi^{k+1}=\phi_{1}^{k+1}+r^{k+1}\phi_{2}^{k+1}+q^{k+1}\phi_{3}^{k+1}, \ \phi=n, \tilde{\mu}, \tilde{L}. \label{eq:savdecouple}
\end{align}
Substituting \eqref{eq:savdecouple} into \eqref{eq:discret_nrq_dc}, we observe that ($n_{i}^{k+1}$, $\tilde{\mu}_{i}^{k+1}$, $\tilde{L}_{i}^{k+1}$), $i=1$, $2$, $3$ can be determined from three decoupled linear systems \cite{liu2024efficient}.
For brevity, we directly present the fully discrete scheme of \eqref{eq:discret_nrq_dc}-\eqref{eq:discret_t_dc}: for $k \geq 0$, 

\textbf{Step 1:}
(\romannumeral 1) Find $(n_{1h}^{k+1},\tilde{\mu}_{1h}^{k+1})\in \Phi_{h}\times \Phi_{h}$, such that for $\forall (\phi_{1h}, \psi_{1h})\in \Phi_{h}\times \Phi_{h}$,
\begin{subequations}\label{eq:discret_n_fullydc}
    \begin{align}
        &\sum_{K\in\mathcal{T}_{h}}\left(n_{1h}^{k+1}-n_{h}^{k},\phi_{1h}\right)_{K}-\sum_{K\in\mathcal{T}_{h}}\delta t\left(n_{h}^{k}\mathbf{v}_{h}^{k},\nabla \phi_{1h}\right)_{K} \label{eq:discret_mass1_fullydc}\\
        &+\sum_{E\in\mathcal{E}_{h}^{I}}\delta t\left\langle \left\lbrace n_{h}^{k}\mathbf{v}_{h}^{k}\right\rbrace\cdot\bm{\gamma}_{E},\left[\phi_{1h}\right]\right\rangle_{E}=-\mathcal{L}_{d}\delta t\mathcal{A}\left( \tilde{\mu}_{1h}^{k+1},\phi_{1h}\right), \nonumber\\
        &\sum_{K\in\mathcal{T}_{h}}\left(\tilde{\mu}_{1h}^{k+1},\psi_{1h}\right)_{K}-\sum_{K\in\mathcal{T}_{h}}\left(\tilde{\mu}_{b1h}^{k}+r^{k}\chi_{b}^{k}\tilde{\mu}_{b2h}^{k},\psi_{1h}\right)_{K}=2\varepsilon\mathcal{A}\left( n_{1h}^{k+1},\psi_{1h}\right)\label{eq:discret_mu1_fullydc}\\
        &+\sum_{E\in\Gamma_{h}}\mathcal{L}_{b}^{-1}\left\langle \delta t^{-1}\left(n_{1h}^{k+1}-n_{h}^{k}\right)+v_{\tau h}^{k}\partial_{\tau}n_{1h}^{k+1}+ \mathcal{L}_{b}q^{k}\chi_{s}^{k}\mathcal{W},\psi_{1h}\right\rangle_{E}. \nonumber
    \end{align}
\end{subequations}
(\romannumeral 2) Find $(n_{2h}^{k+1},\tilde{\mu}_{2h}^{k+1})\in \Phi_{h}\times \Phi_{h}$, such that for $\forall (\phi_{2h}, \psi_{2h})\in \Phi_{h}\times \Phi_{h}$,
\begin{subequations}
    \begin{align}
        &\sum_{K\in\mathcal{T}_{h}}\left(n_{2h}^{k+1},\phi_{2h}\right)_{K}=-\mathcal{L}_{d}\delta t\mathcal{A}\left( \tilde{\mu}_{2h}^{k+1},\phi_{2h}\right),\label{eq:discret_mass2_fullydc} \\
        &\sum_{K\in\mathcal{T}_{h}}\left(\tilde{\mu}_{2h}^{k+1},\psi_{2h}\right)_{K}-\sum_{K\in\mathcal{T}_{h}}\left(\chi_{b}^{k}\tilde{\mu}_{b2h}^{k},\psi_{2h}\right)_{K}=2\varepsilon\mathcal{A}\left( n_{2h}^{k+1},\psi_{2h}\right) \label{eq:discret_mu2_fullydc}\\
        &+\sum_{E\in\Gamma_{h}}\mathcal{L}_{b}^{-1}\left\langle \delta t^{-1}n_{2h}^{k+1}+v_{\tau h}^{k}\partial_{\tau}n_{2h}^{k+1},\psi_{2h}\right\rangle_{E}.\nonumber
    \end{align}
\end{subequations}
(\romannumeral 3) Find $(n_{3h}^{k+1},\tilde{\mu}_{3h}^{k+1})\in \Phi_{h}\times \Phi_{h}$, such that for $\forall (\phi_{3h}, \psi_{3h})\in \Phi_{h}\times \Phi_{h}$,
\begin{subequations}
    \begin{align}
        &\sum_{K\in\mathcal{T}_{h}}\left(n_{3h}^{k+1},\phi_{3h}\right)_{K}=-\mathcal{L}_{d}\delta t\mathcal{A}\left( \tilde{\mu}_{3h}^{k+1},\phi_{3h}\right), \label{eq:discret_mass3_fullydc}\\
        &\sum_{K\in\mathcal{T}_{h}}\left(\tilde{\mu}_{3h}^{k+1},\psi_{3h}\right)_{K}=2\varepsilon\mathcal{A}\left( n_{3h}^{k+1},\psi_{3h}\right)\label{eq:discret_mu3_fullydc}\\
        &+\sum_{E\in\Gamma_{h}}\mathcal{L}_{b}^{-1}\left\langle \delta t^{-1}n_{3h}^{k+1}+v_{\tau h}^{k}\partial_{\tau}n_{3h}^{k+1}+ \mathcal{L}_{b}\chi_{s}^{k}\mathcal{W},\psi_{3h}\right\rangle_{E}. \nonumber
    \end{align}
\end{subequations}
Here, $\chi_{b}^{k}:=1/2\sqrt{S_{b2}^{k}+C_{r}}$, $\chi_{s}^{k}:=1/2\sqrt{S_{s}^{k}+C_{q}}$ and $\mathcal{A}\left(\cdot,\cdot\right)$ is the bilinear form defined as \cite{riviere2008discontinuous}
\begin{align}
    \mathcal{A}\left(\phi,\psi\right)=&\sum_{K\in\mathcal{T}_{h}}\left(\nabla \phi,\nabla \psi\right)_{K}-\sum_{E\in\mathcal{E}_{h}^{I}}\left\langle\left\lbrace\nabla \phi\cdot\bm{\gamma}_{E}\right\rbrace, \left[\psi\right]\right\rangle_{E} \nonumber\\
    &-\sum_{E\in\mathcal{E}_{h}^{I}}\left\langle\left\lbrace\nabla\psi\cdot\bm{\gamma}_{E}\right\rbrace, \left[\phi\right]\right\rangle_{E}+\sum_{E\in\mathcal{E}_{h}^{I}}\frac{\sigma_{E}}{h_{E}}\left\langle \left[\phi\right], \left[\psi\right]\right\rangle_{E}, \ \forall\phi, \psi\in \Phi_{h},\nonumber
\end{align}
where $h_{E}=\vert E\vert$ and $\sigma_{E}$ is the penalty parameter.
Notice that the above three linear fourth-order systems share a same time-independent coefficient matrix.
Once ($n_{ih}^{k+1}$, $\tilde{\mu}_{ih}^{k+1}$, $\tilde{L}_{ih}^{k+1}$) for $i=1, 2, 3$ are solved, $r^{k+1}$ and $q^{k+1}$ can be computed from
\begin{align}
    r^{k+1}&=r^{k}+\chi_{b}^{k}\sum_{K\in\mathcal{T}_{h}}\left(-\tilde{\mu}_{b2h}^{k},n_{1h}^{k+1}+r^{k+1}n_{2h}^{k+1}+q^{k+1}n_{3h}^{k+1} - n_{h}^{k}\right)_{K}, \nonumber\\
    q^{k+1}&=q^{k}+\chi_{s}^{k}\sum_{E\in\Gamma_{h}}\left\langle -\mathcal{W},n_{1h}^{k+1}+r^{k+1}n_{2h}^{k+1}+q^{k+1}n_{3h}^{k+1} - n_{h}^{k}\right\rangle_{E}.\nonumber
\end{align}
Consequently, $n_{h}^{k+1}$ and $\tilde{\mu}_{h}^{k+1}$ can be obtained from \eqref{eq:savdecouple}.

\textbf{A simple limiter to enforce bounds}: Notice that $n_{h}^{k+1}$ is expected to strictly remain in the range $0<n_{h}^{k+1}<1$ by \eqref{eq:dimensionless_mu_sav}.
For DG methods, the bounds-preserving can be enforced by the linear scaling limiter \cite{ZHANG20103091,zhang2017positivity}.
Let $n_{h,K}^{k+1}$ be the DG approximation of $n$ at the $(k+1)$-th time step for any $K\in\mathcal{T}_{h}$.
Assume that the cell average $\bar{n}_{h,K}^{k+1}:=\int_{K}n_{h,K}^{k+1}d\mathbf{x}/\vert K\vert\in(0,1)$.
Then, the limiter can be defined as
\begin{align}
   \forall K\in\mathcal{T}_{h}: \quad \tilde{n}_{h,K}^{k+1}&=\theta(n_{h,K}^{k+1}-\bar{n}_{h,K}^{k+1})+\bar{n}_{h,K}^{k+1}, \label{eq:limiter}\\
   \theta&=\min{\left\{1,\left\vert\frac{1-\bar{n}_{h,K}^{k+1}}{\max_{K}n_{h,K}^{k+1}-\bar{n}_{h,K}^{k+1}}\right\vert,\left\vert\frac{0-\bar{n}_{h,K}^{k+1}}{\min_{K}n_{h,K}^{k+1}-\bar{n}_{h,K}^{k+1}}\right\vert\right\}}.\nonumber
\end{align}
Here, $\tilde{n}_{h,K}^{k+1}$ is the modified approximation after scaling.
It is straightforward to see that $\tilde{n}_{h,K}^{k+1}\in(0,1)$ in $K$ and the cell average of $\tilde{n}_{h,K}^{k+1}$ is still $\bar{n}_{h,K}^{k+1}$.
Moreover, this simple limiter does not destroy the approximation accuracy of $n_{h,K}^{k+1}$.
For convenience, we use $n_{h}^{k+1}$ to denote the modified approximation below.

\textbf{Step 2:} Find $\mathbf{v}_{h}^{k+1}\in \mathbf{U}_{h}$, such that for $\forall \mathbf{u}_{h}\in \mathbf{U}_{h}$,
\begin{align}
    &\left(\mathcal{R}_{e}\mathcal{R}n_{h}^{k}\delta t^{-1}\left(\mathbf{v}_{h}^{k+1}-\mathbf{v}_{h}^{k}\right),\mathbf{u}_{h}\right)+\left(\left(\mathcal{R}_{e}\mathcal{R}n_{h}^{k}\mathbf{v}_{h}^{k}-\mathcal{B}\nabla \tilde{\mu}_{h}^{k+1}\right)\cdot\nabla\mathbf{v}_{h}^{k+1},\mathbf{u}_{h}\right) \nonumber\\
    &=\left(\Pi_{h}^{k+1}-\mathcal{R}\Theta_{h}^{k+1},\nabla \mathbf{u}_{h}\right)+\left\langle -\mathcal{L}_{s}^{-1}\mathcal{R}\beta_{h}^{k+1}\left(v_{\tau h}^{k+1}-v_{w}\right),u_{\tau h}\right\rangle \nonumber \\
    &+\left\langle q^{k\ast}\mathcal{W}T_{h}^{k}\partial_{\tau}n_{h}^{k+1}+\mathcal{W}\left(n_{h}^{k+1}-n_{c}\right)\partial_{\tau}T_{h}^{k},u_{\tau h}\right\rangle.\nonumber
\end{align}

\textbf{Step 3:} Find $T_{h}^{k+1}\in W_{h}$, such that for $\forall w_{h}\in W_{h}$,
\begin{align}
   &\delta t^{-1}\left(\frac{3}{2}n_{h}^{k+1}\left(T_{h}^{k+1}-T_{h}^{k}\right),w_{h}\right)
    +\left(T_{h}^{k+1}\left(\vartheta_{1h}^{k}\right)^{2}\nabla\cdot \left(T_{h}^{k}\mathbf{v}_{h}^{k}\right),w_{h}\right) \nonumber\\
    &-\left(\mathcal{L}_{d}\vartheta_{2h}^{k}\nabla  \tilde{\mu}_{h}^{k+1}\cdot\nabla T_{h}^{k+1}+\mathcal{L}_{d}T_{h}^{k+1}\nabla  \tilde{\mu}_{h}^{k+1}\cdot\nabla \vartheta_{2h}^{k}+n_{h}^{k}T_{h}^{k+1}\nabla\tilde{\mu}_{h}^{k+1}\cdot\mathbf{v}_{h}^{k},w_{h}\right)\nonumber\\
    &-\left(\mathcal{L}_{d}T_{h}^{k+1}\vartheta_{2h}^{k}\nabla  \tilde{\mu}_{h}^{k+1}+T_{h}^{k+1}\bm{\vartheta}_{3h}^{k}\cdot\mathbf{v}_{h}^{k}-\mathcal{R}_{e}^{-1} n_{h}^{k+1}\nabla T_{h}^{k+1},\nabla w_{h}\right) \nonumber\\
    &=\left(\left(\bm{\vartheta}_{3h}^{k}\cdot\mathbf{v}_{h}^{k}\right)\cdot\nabla T_{h}^{k+1}+\mathcal{R}\Theta_{h}^{k+1}:\nabla\mathbf{v}_{h}^{k+1},w_{h}\right)-\mathcal{R}_{e}^{-1}\mathcal{L}_{\lambda}\left\langle\lambda_{sh}^{k+1}\partial_{\tau}T_{h}^{k+1},\partial_{\tau}w_{h}\right\rangle \nonumber\\
    &-\left\langle T_{h}^{k+1}\left(v_{\tau h}^{k+1}\Pi_{\gamma\tau h}^{k+1}- \left(\mathcal{R}_{e}\mathcal{L}_{\kappa}\kappa_{h}^{k}T_{h}^{k}\right)^{-1}\vartheta_{4h}^{k}-q^{k\ast}\mathcal{W} \left(n_{h}^{k+1}-n_{c}\right)\partial_{\tau}v_{\tau h}^{k}\right),w_{h}\right\rangle \nonumber\\ 
    &-\left\langle v_{\tau h}^{k+1}\left(\mathcal{R}\Theta_{\gamma\tau h}^{k+1}-T_{h}^{k}\tilde{L}_{h}^{k+1}\partial_{\tau}n_{h}^{k+1}-\mathcal{W}\left(n_{h}^{k+1}-n_{c}\right)\partial_{\tau}T_{h}^{k}\right),w_{h}\right\rangle, \nonumber
\end{align}
where $\vartheta_{1h}^{k}:=n_{h}^{k}/T_{h}^{k}$, $\vartheta_{2h}^{k}:=3/2-2n_{h}^{k}/T_{h}^{k}$, $\bm{\vartheta}_{3h}^{k}:=\Pi_{h}^{k+1}/T_{h}^{k}$ and $\vartheta_{4h}^{k}:=1/T_{h}^{k}-1/T_{w}$.
Furthermore, under the assumption that the particle mass $m$ is constant, the mass conservation is equivalent to the conservation of particle number in our model. 
\begin{theorem}\label{thm:mass}
   The fully discrete scheme preserves the boundedness of $n$ and the conservation of total mass, namely,
   \begin{equation*}
       \begin{aligned}
        0<n_{h}^{k+1}<1, \quad \sum_{K\in\mathcal{T}_{h}}\int_{K}\rho_{h}^{k+1}d\mathbf{x}=\sum_{K\in\mathcal{T}_{h}}\int_{K}\rho_{h}^{k}d\mathbf{x}=\cdots=\sum_{K\in\mathcal{T}_{h}}\int_{K}\rho_{h}^{0}d\mathbf{x}.
       \end{aligned}
   \end{equation*}
\end{theorem}
\begin{proof}
    Summing \eqref{eq:discret_mass1_fullydc}, $r^{k+1}$\eqref{eq:discret_mass2_fullydc}, and $q^{k+1}$\eqref{eq:discret_mass3_fullydc}, and taking $\phi_{1h}=\phi_{2h}=\phi_{3h}=1$ in any element $K\in\mathcal{T}_{h}$ yields the total mass conservation.
    On the other hand, the linear scaling limiter \eqref{eq:limiter} ensures the strict bounds-preserving, expressed as $n_{h}^{k+1}\in(0,1)$ without changing the cell average of $n_{h}^{k+1}$ in $K$, hence, it does not destroy the local mass conservation.
\end{proof}

\section{Numerical experiments}
\label{sec:Numerical experiments}
In this section, we apply the proposed numerical method in \cref{subsec:A decoupled linear and unconditionally entropy-stable scheme} to simulate several non-isothermal gas-liquid flow problems. 
The fluid-solid interfaces at top and bottom boundaries are assumed to be isothermal in the first two examples and non-isothermal in the last example.
More precisely, the Dirichlet boundary condition for $T$ (\cref{rmk:tem_dirichlet_bcd}) is imposed in Example 1 and 2, while \eqref{eq:dimensionless_bcd} is imposed in Example 3.
Periodic boundary conditions are imposed on the left and right boundaries.
We employ the fully discrete scheme on a uniform quadrilateral mesh with $h=0.01$, $p=1$, $\sigma_{E}=10$ and $\delta t = 0.1h$.
The gas and liquid densities are set to the values at gas-liquid coexistence, given by $n_{g}=0.122$ and $n_{l}=0.58$ at $T = 0.875 T_{c}$. 
Here, $T_{c}=3.32$ K is the critical temperature, which becomes $0.296$ after nondimensionalization \cite{onuki2005dynamic,xu2010contact}.
The critical density is set as $n_{c}=1/3$.
The wettability of fluid-solid surface, characterized by the static contact angle $\theta_{s}$ (for liquid phase), is adjusted by varying the wettability number $\mathcal{W}$ \cite{xu2012droplet}. 
Specifically, $\mathcal{W}<0$ and $\mathcal{W}>0$ correspond to $\theta_{s}<90^{\circ}$ (hydrophilic) and $\theta_{s}>90^{\circ}$ (hydrophobic), respectively.
In particular, $\mathcal{W} = 0$ corresponds to $\theta_s = 90^{\circ}$.
The particle mass is set to $m=1$, and the parameters in MSAV approach are chosen as $C_{r} = C_{q} = 5$.

\subsection{Equilibrium evolution of a square droplet}
\label{subsec:Equilibrium evolution of a square droplet}
To demonstrate that our model is applicable beyond MCL problems, we first simulate the square-shaped interface droplet problem \cite{kou2018thermodynamically2}.
Parameters used are $\mathcal{L}_{d}=0.0005$, $\mathcal{L}_{b}=500$, $\mathcal{R}_{e}=0.1$, $\mathcal{R}=0.004$, $\mathcal{L}_{s}=1 / n_{l}$, $\beta_{gl}=1$, $\varepsilon=0.0001$, $\mathcal{W}=v_{w}=0$, $\Omega= [0, 1] \times [0, 1]$ and $T_{f}=1$.
A square-shaped droplet is initially placed at the center of domain, with $\mathbf{v}$ set to $0$.
The initial $T=0.875 T_{c}$ and $T_{w}$ is also fixed at $0.875 T_{c}$.

\cref{fig:case1:T0.875n} shows the dynamical evolution of $n$. 
It is clearly observed that the liquid droplet gradually reshapes into a circle shape from its initial square shape.
The droplet has an evolutionary tendency to an equilibrium state, characterized by the homogeneity of $T$ and $\tilde{\mu}$, as shown in \cref{fig:case1:T0.875TmuSM}(a) and (b).
The evolution of total entropy is illustrated in \cref{fig:case1:T0.875TmuSM}(c), showing the strict increase in both original and modified total entropy over time steps.
Here, $S_{\text{tot}}$ denotes the original total entropy defined by \eqref{eq:totaldiscreteentropydefine} in MSAV approach.
The evolution of total mass error is plotted in \cref{fig:case1:T0.875TmuSM}(d), which indicates that the proposed method conserves total mass.
\begin{figure}[htbp]
\centering
\subfloat[$t=0$]{\includegraphics[width=0.24\textwidth]{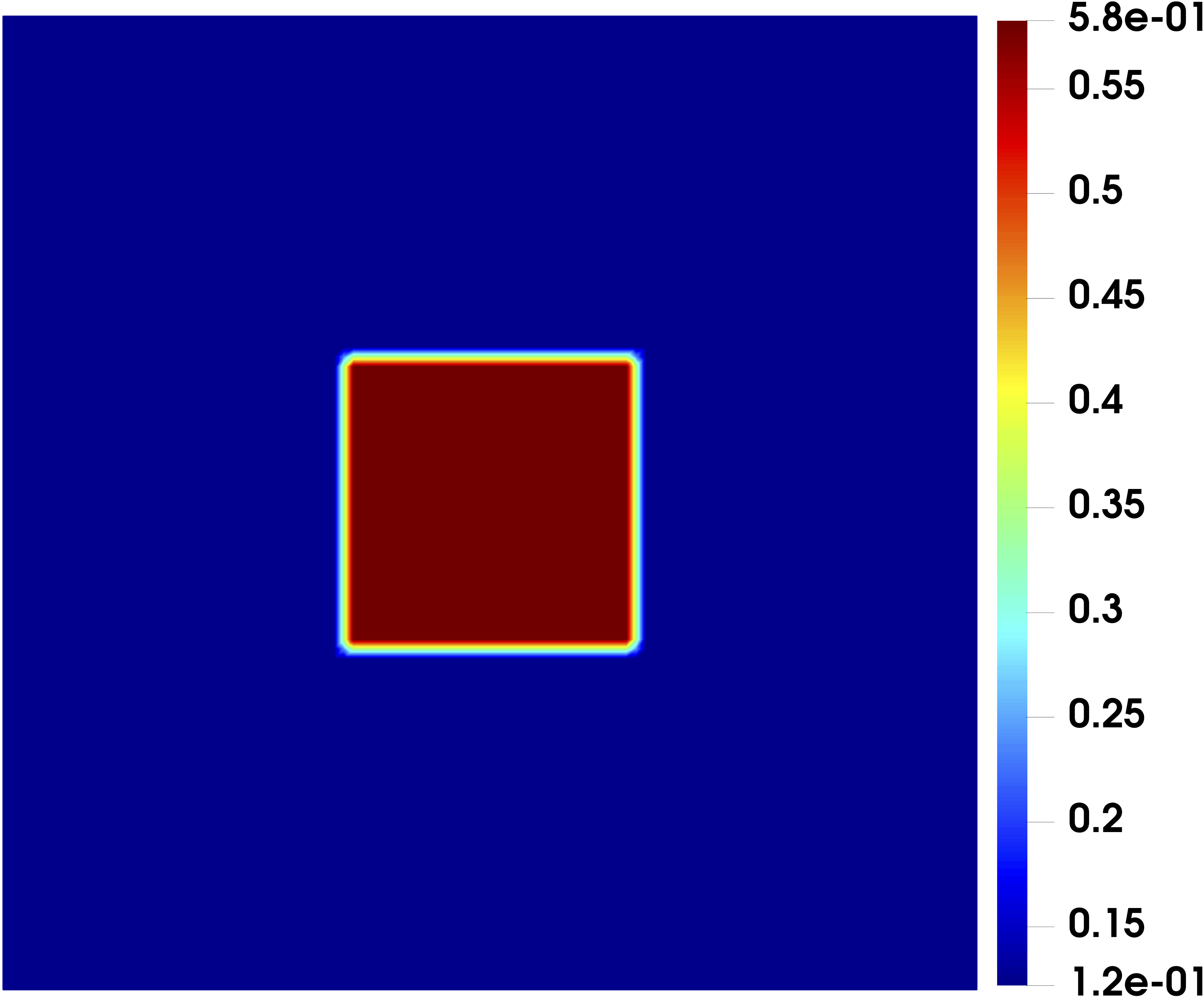}}
\hspace{5mm}
\subfloat[$t=0.5$]{\includegraphics[width=0.24\textwidth]{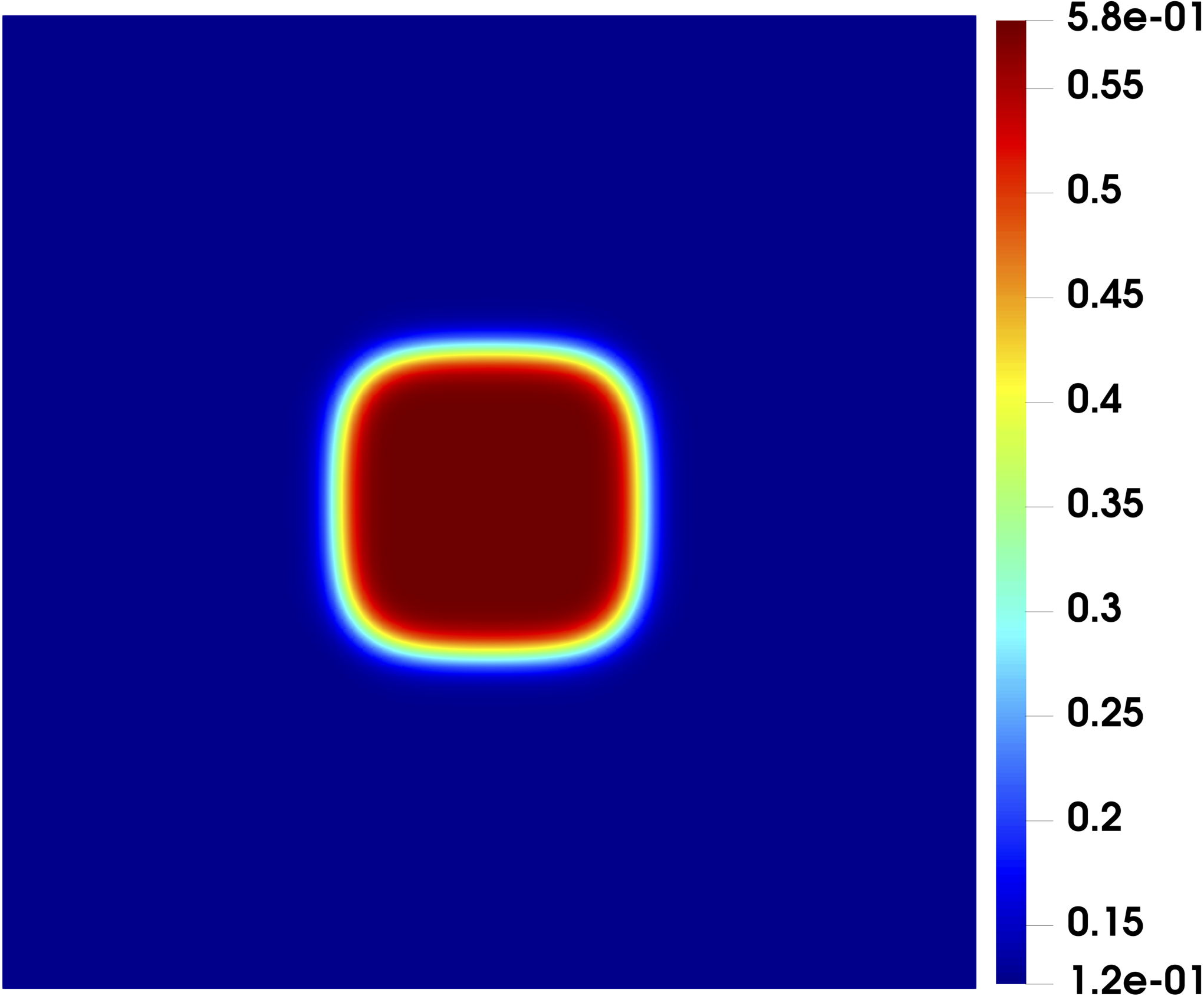}}
\hspace{5mm}
\subfloat[$t=1$]{\includegraphics[width=0.24\textwidth]{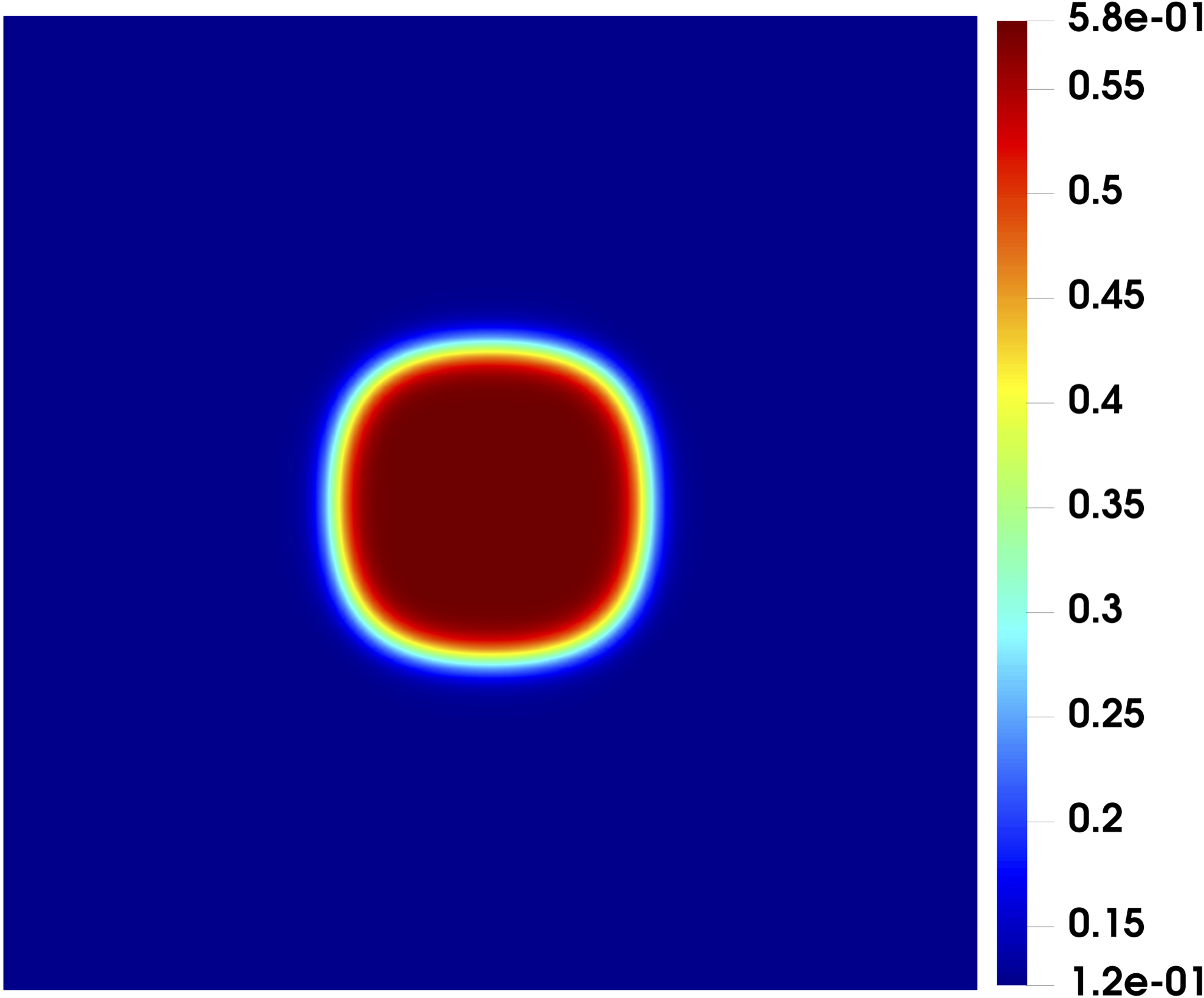}}
\caption{Dynamical evolution of $n$ for initial $T=0.875 T_{c}$ and $\varepsilon=0.0001$.}
\label{fig:case1:T0.875n}
\end{figure}
\begin{figure}[htbp]
\centering
\subfloat[$T$]{\includegraphics[width=0.24\textwidth]{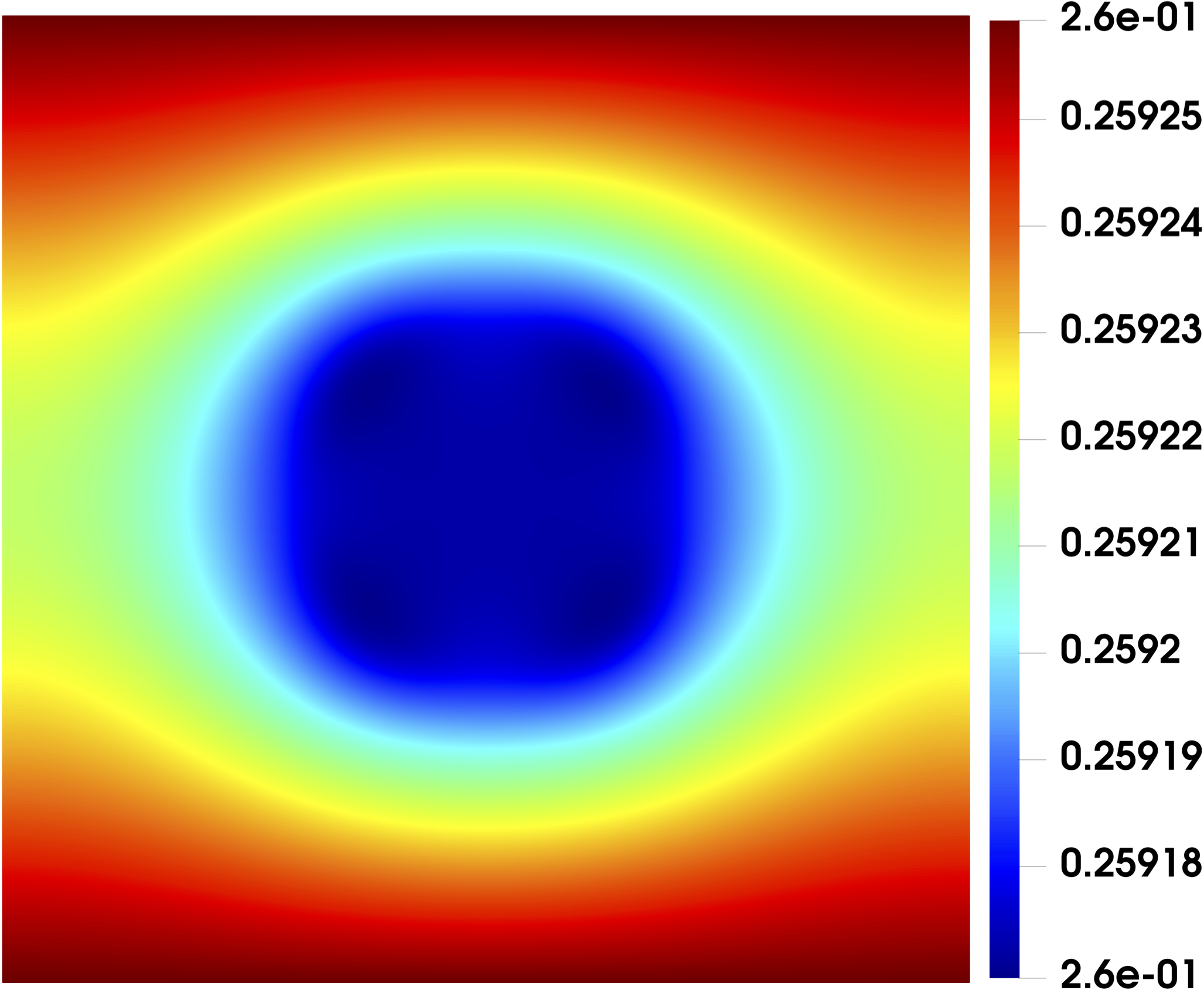}}
\hspace{0.5mm}
\subfloat[$\tilde{\mu}$]{\includegraphics[width=0.24\textwidth]{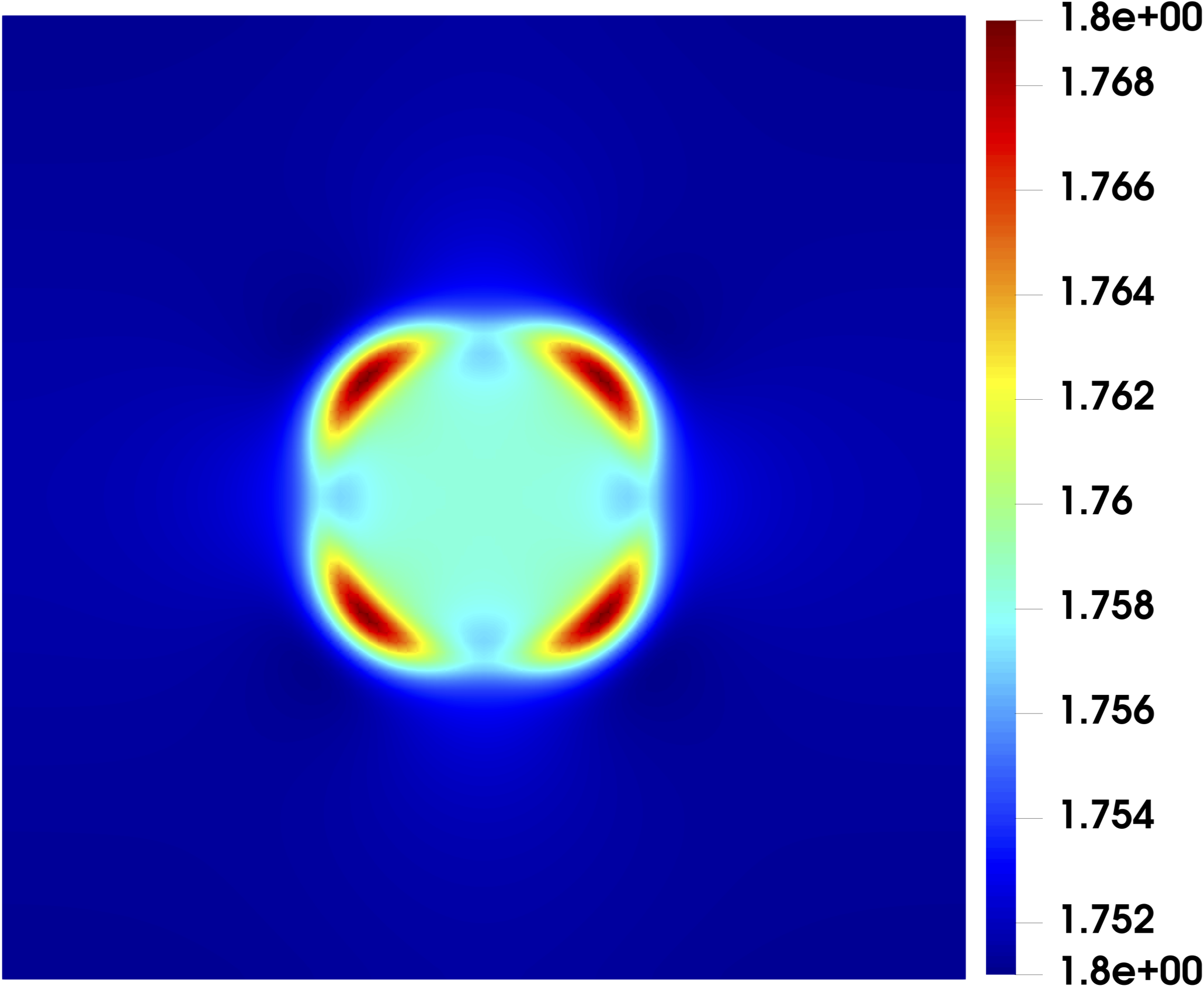}}
\hspace{0.5mm}
\subfloat[total entropy]{\includegraphics[width=0.24\textwidth]{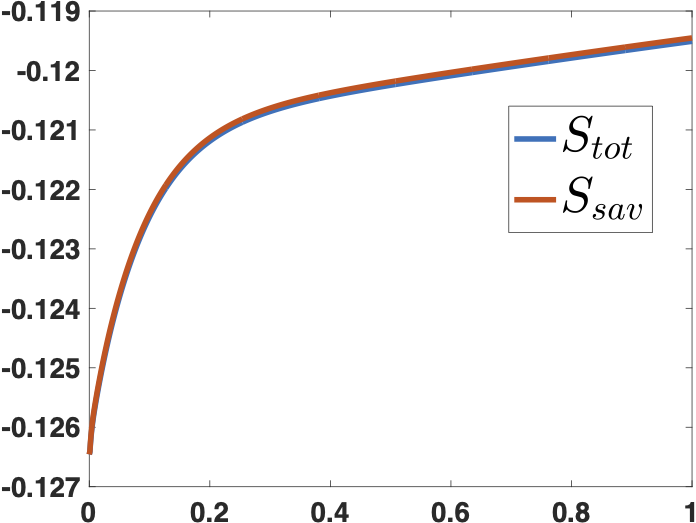}}
\hspace{0.5mm}
\subfloat[total mass error]{\includegraphics[width=0.24\textwidth]{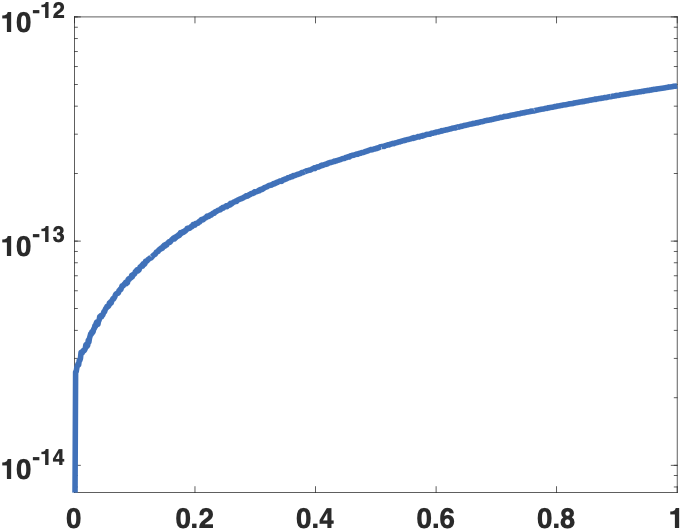}}
\caption{Final-time profiles, and evolution of total entropy and total mass error for initial $T=0.875 T_{c}$ and $\varepsilon=0.0001$.} 
\label{fig:case1:T0.875TmuSM}
\end{figure}
Furthermore, we simulate the dynamics of droplet under same parameters, except for $\varepsilon$ and initial $T$. 
\cref{fig:case1:thickness T0.95 nT} shows the final-time profiles of $n$ and $T$ for two cases: (\romannumeral 1) initial $T=0.875 T_{c}$, $\varepsilon=0.001$ and (\romannumeral 2) initial $T=0.95 T_{c}$, $\varepsilon=0.001$.
The results show that the magnitude of $\varepsilon$ affects interface thickness but has little impact on dynamics and structure of the droplet.
Compared to \cref{fig:case1:T0.875n}(c), the interface in \cref{fig:case1:thickness T0.95 nT}(c) is stabilized by a faster phase transition due to $T$ being closer to $T_{c}$.
\begin{figure}[htbp]
\centering
\subfloat[$n$]{\includegraphics[width=0.24\textwidth]{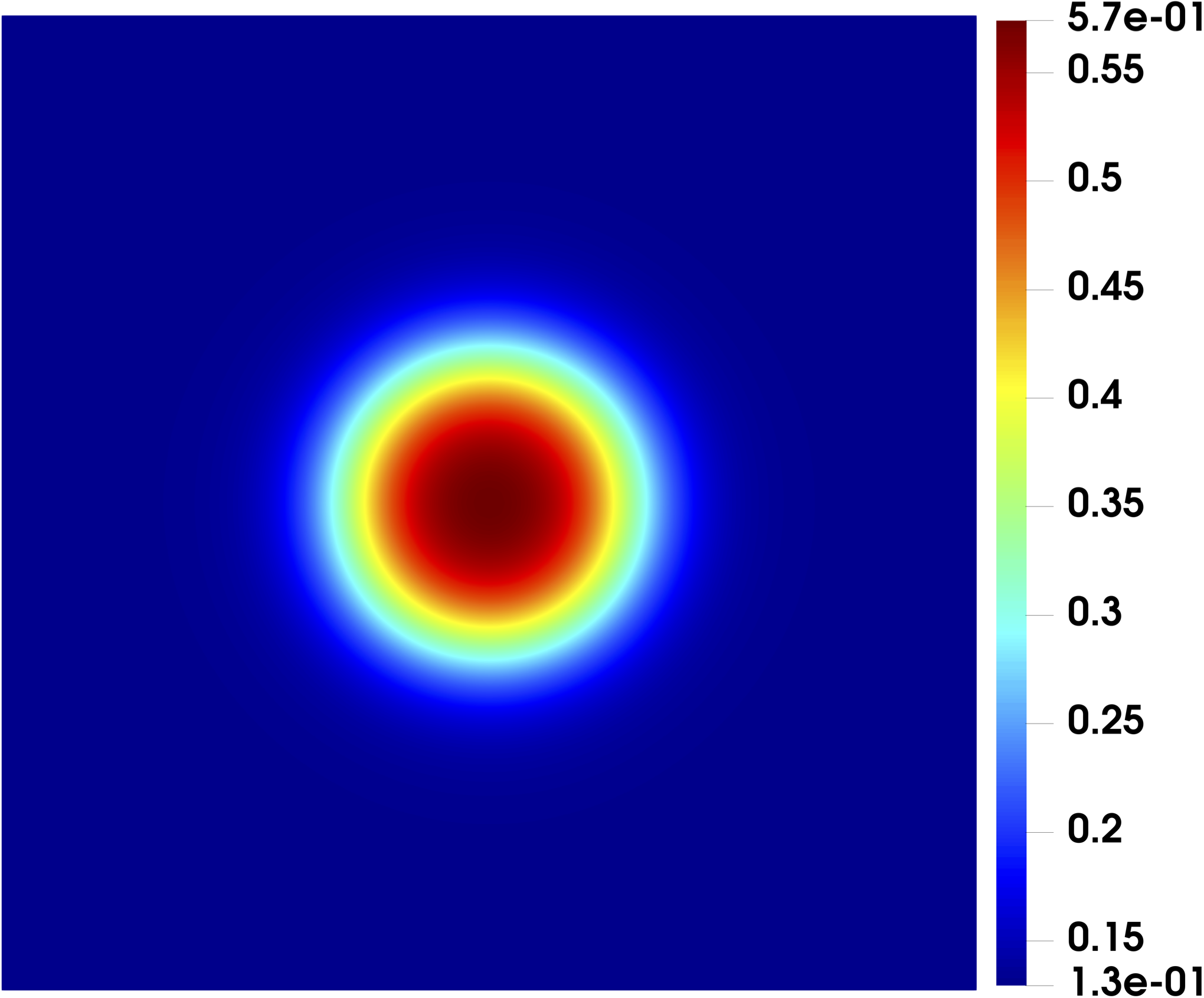}}
\hspace{0.5mm}
\subfloat[$T$]{\includegraphics[width=0.24\textwidth]{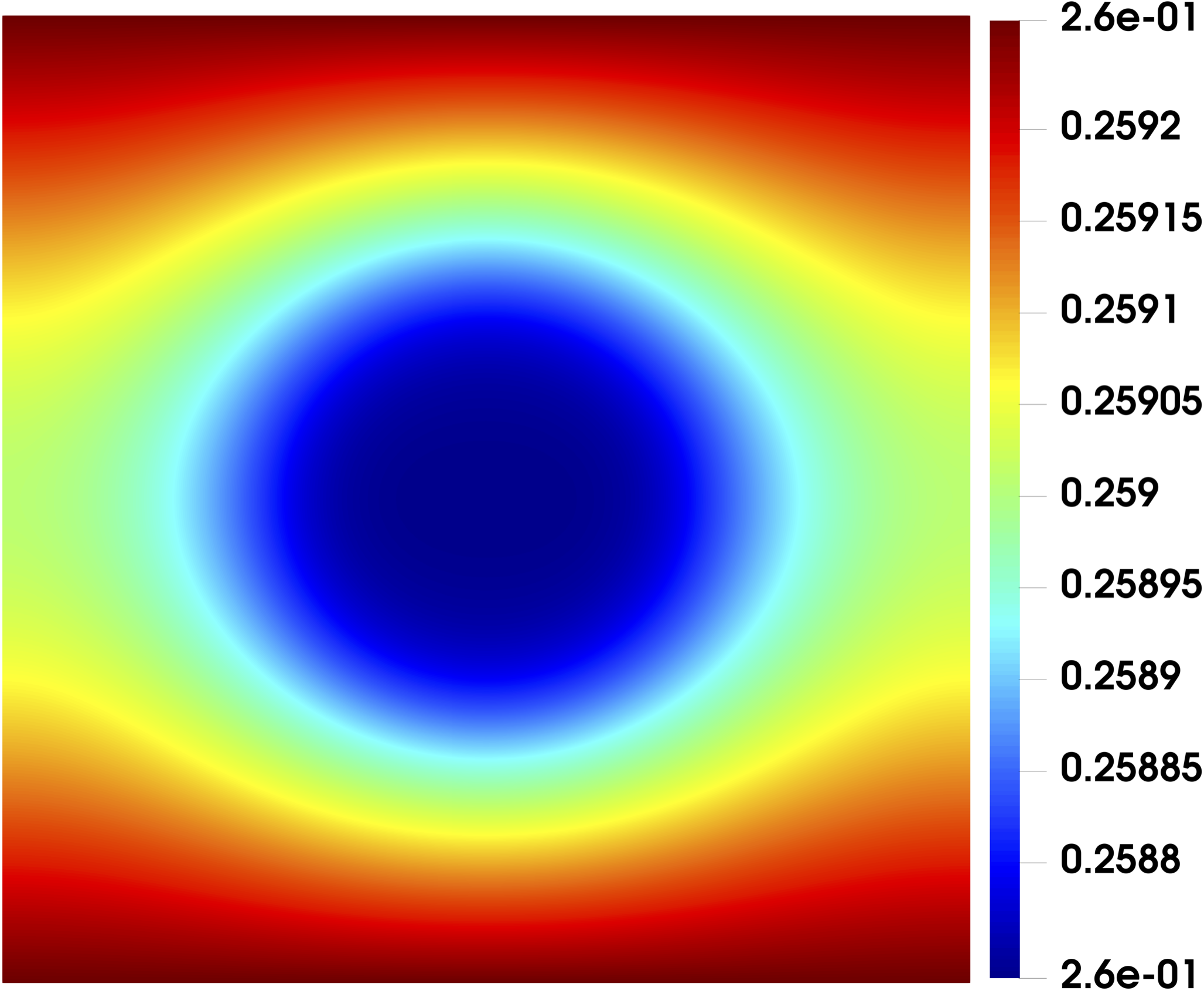}}
\hspace{0.5mm}
\subfloat[$n$]{\includegraphics[width=0.24\textwidth]{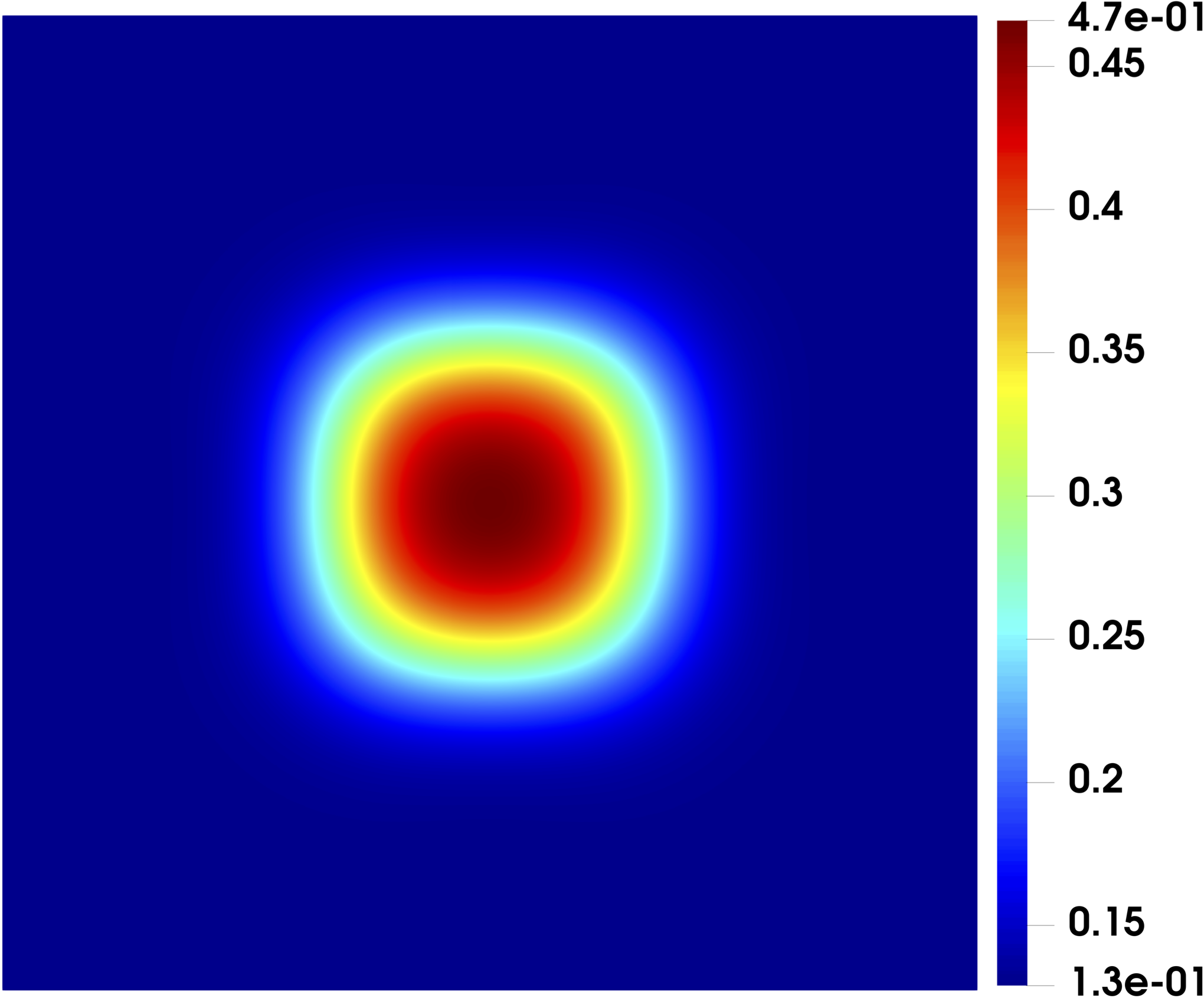}}
\hspace{0.5mm}
\subfloat[$T$]{\includegraphics[width=0.24\textwidth]{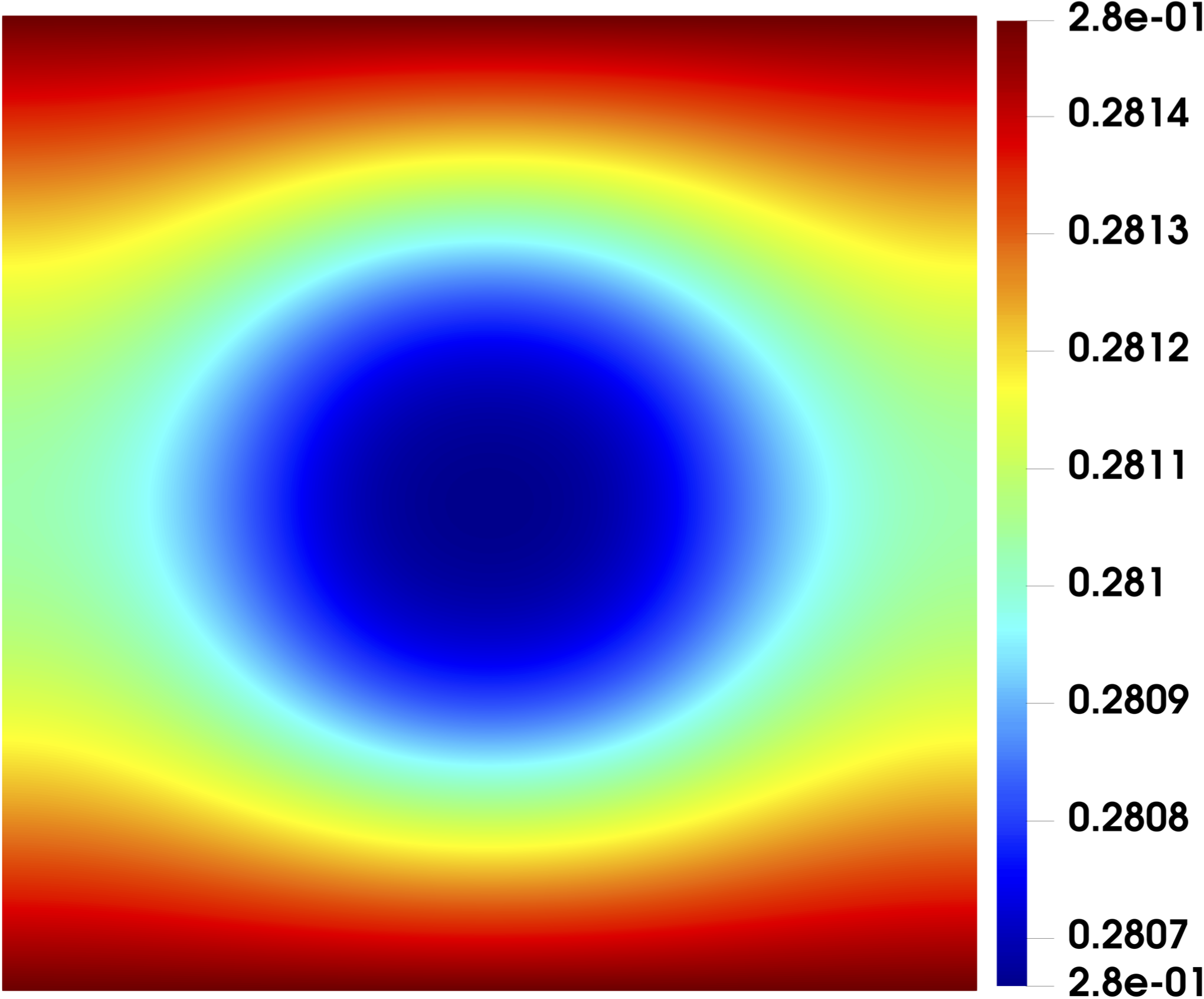}}
\caption{Final-time profiles: (a)-(b) initial $T=0.875 T_{c}$, $\varepsilon=0.001$; (c)-(d) initial $T=0.95 T_{c}$, $\varepsilon=0.001$.} 
\label{fig:case1:thickness T0.95 nT}
\end{figure}

\subsection{Contact line motion with moving walls}
\label{subsec:Contact line motion with moving walls}
In this example, we simulate a sheared gas-liquid system in a Couette-flow geometry, as shown in \cref{fig:case2:W0.00 T0.875 initial} \cite{xu2010contact,shen2020energy}.
\begin{figure}[htbp]
\centering
\includegraphics[width=0.6\textwidth]{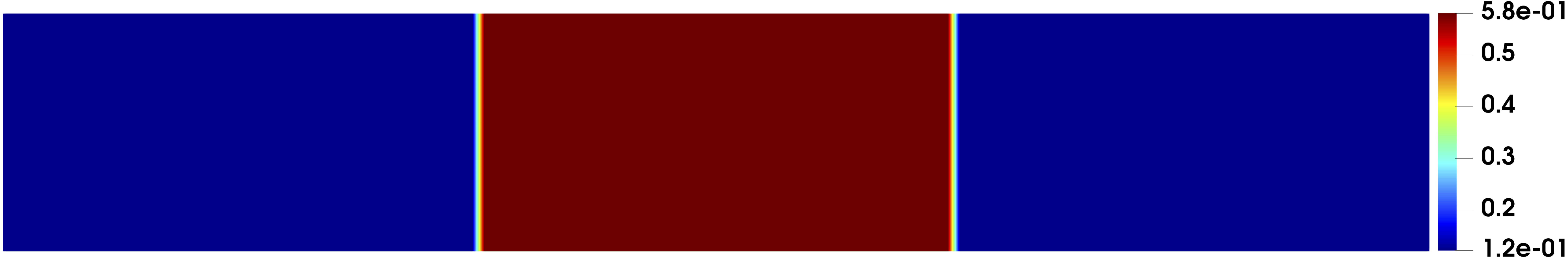}
\caption{Initial condition of $n$ for the two phase Couette flow.}
\label{fig:case2:W0.00 T0.875 initial}
\end{figure}
The Couette flow is induced by moving the top and bottom walls with a constant speed $v_{w}$ in opposite $\pm x$-directions.
Parameters are same as Example 1, except for $v_{w}=0.15$, $\Omega= [0, 2.4] \times [0, 0.4]$ and $T_{f}=2$.
The initial $T=0.875 T_{c}$, and $T_{w}$ is fixed at $0.875 T_{c}$.
\cref{fig:case2:W0.00 T0.875nTmu} shows the final-time $n$, $\mathbf{v}$, $T$ and $\tilde{\mu}$. 
It is observed that the contact line motion reaches a steady state and the gas-liquid interfaces are significantly penetrated by flow close to the solid surfaces. 
The homogeneity of $T$ and $\tilde{\mu}$ also indicate that the system arrives an equilibrium state.
\begin{figure}[htbp]
\subfloat[$n$]{\includegraphics[width=0.495\textwidth]{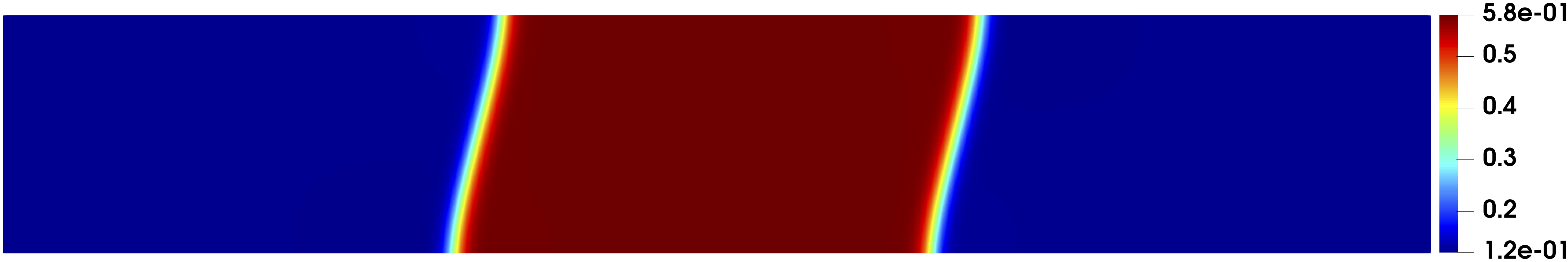}}
\subfloat[$\mathbf{v}$]{\includegraphics[width=0.495\textwidth]{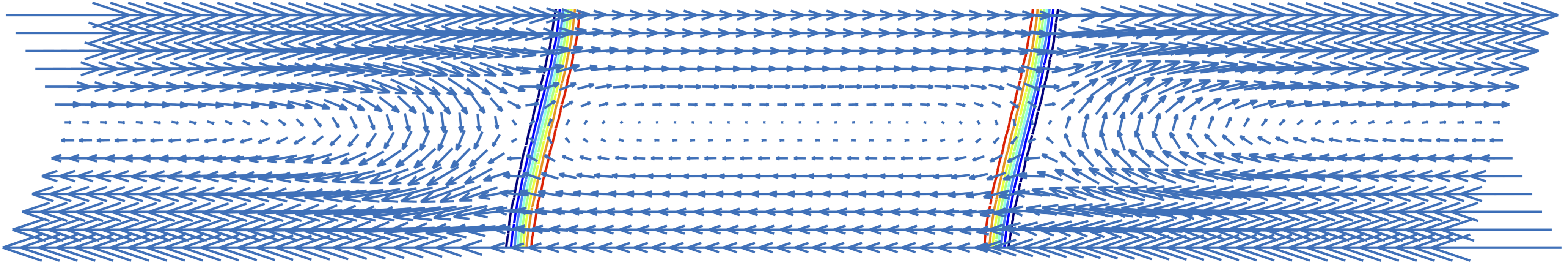}}\\
\subfloat[$T$]{\includegraphics[width=0.495\textwidth]{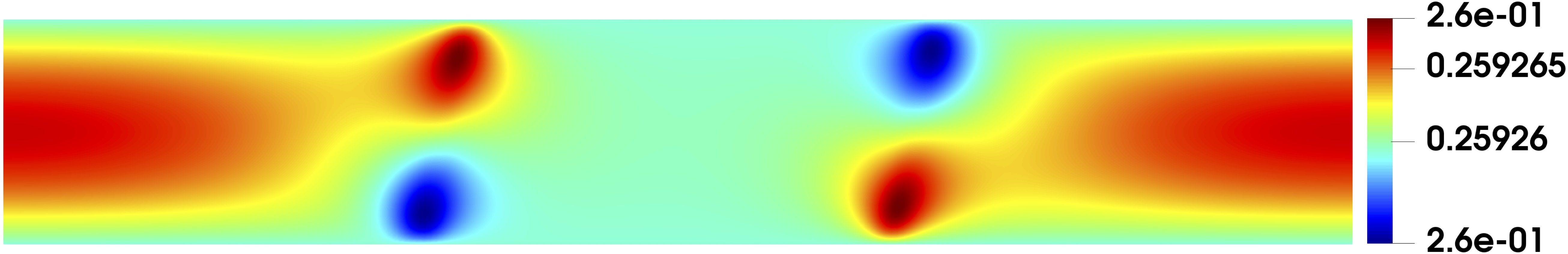}}
\subfloat[$\tilde{\mu}$]{\includegraphics[width=0.495\textwidth]{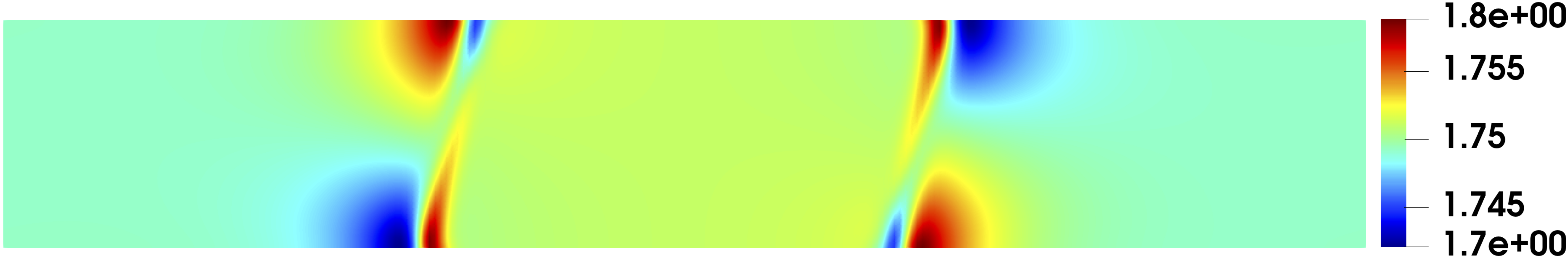}}
\caption{Final-time profiles for initial $T=0.875T_{c}$ and $\varepsilon=0.0001$.} 
\label{fig:case2:W0.00 T0.875nTmu}
\end{figure}
When the wall speed $v_{w}$ is not zero, the entropy produced by moving walls (denoted by $S_{w}$) should be considered. 
We define $S_{w}^{k}=-\sum_{i=0}^{k}\delta t\mathcal{L}_{s}^{-1}\mathcal{R}\int_{\Gamma}\beta^{i} v_{w}v_{\tau}^{slip,i}/T^{i}dA$ \cite{gao2014efficient,wang2022energy}.
As shown in \cref{thm:secondlaw_dc}, when the Dirichlet boundary condition for $T$ is imposed, $S_{\text{sav}}+S_{w}$ should increase over time.
\cref{fig:case2:entropy mass}(a) clearly demonstrates the strict increase of total entropy over time steps and \cref{fig:case2:entropy mass}(d) shows the evolution of the total mass error. 

To further test the effect of $\varepsilon$ and initial $T$, we also simulate this Couette-flow system with different $\varepsilon$ and initial $T$. 
We reset $v_{w}=0.5$ to more clearly observe the contact line motion.
The final-time profiles of $n$, $\mathbf{v}$, $T$ and $\tilde{\mu}$ for initial $T=0.875 T_{c}$ and $\varepsilon=0.001$ are shown in \cref{fig:case2:W0.00 T0.875 thickness nv T mu}.
We observe that the magnitude of $\varepsilon$ affects the interface thickness, fluid transport across the interface, and the homogeneity of $T$ and $\tilde{\mu}$.
\begin{figure}[htbp]
\centering
\subfloat[$n$]{\includegraphics[width=0.495\textwidth]{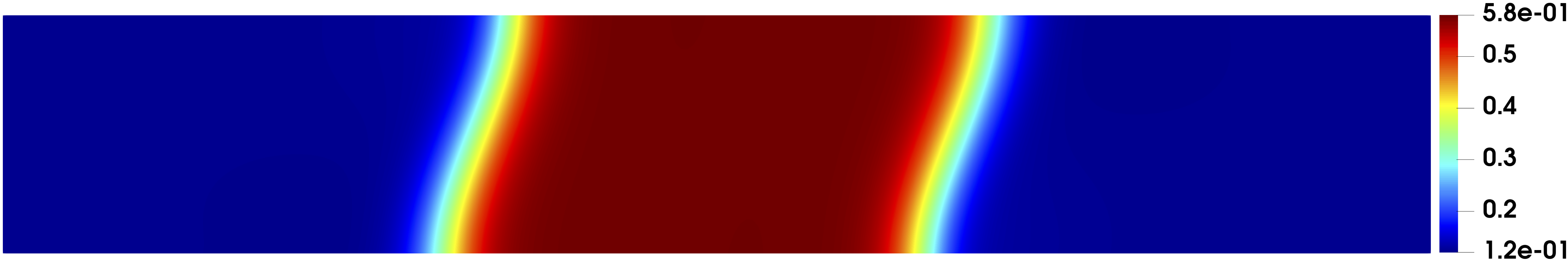}}
\subfloat[$\mathbf{v}$]{\includegraphics[width=0.495\textwidth]{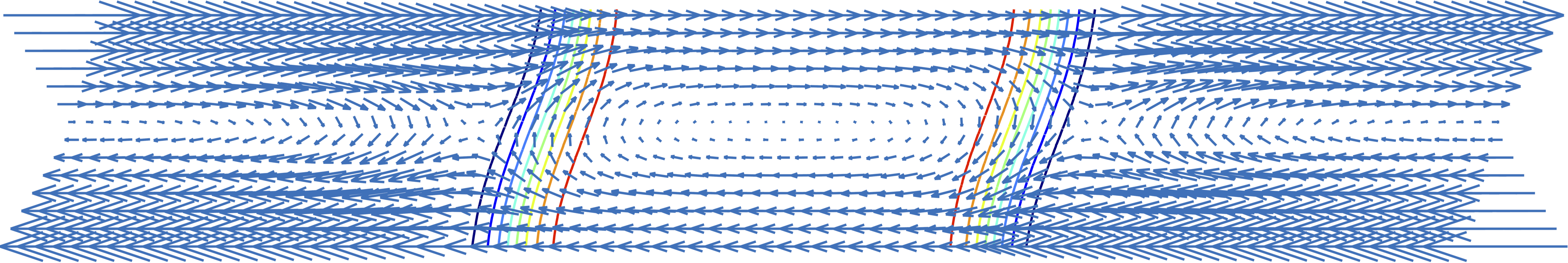}}\\
\subfloat[$T$]{\includegraphics[width=0.495\textwidth]{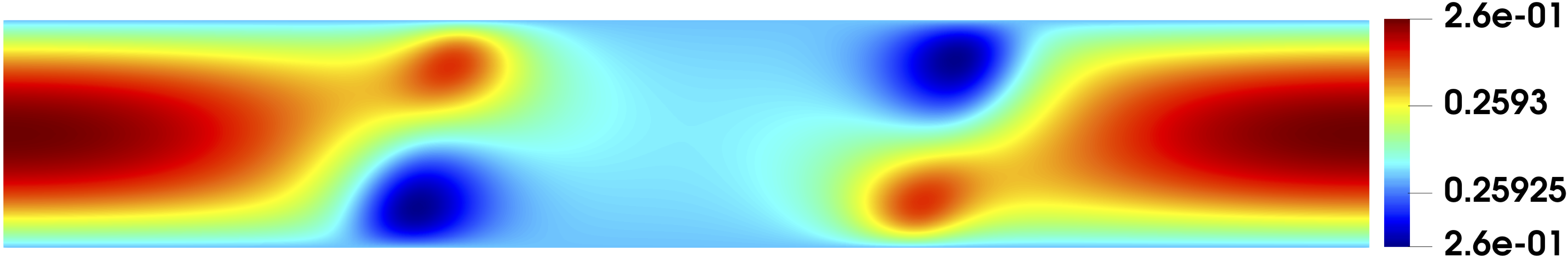}}
\subfloat[$\tilde{\mu}$]{\includegraphics[width=0.495\textwidth]{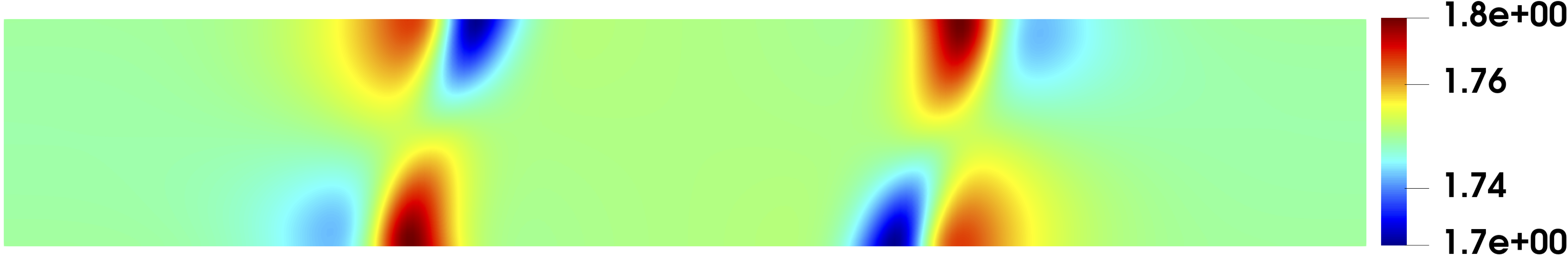}}
\caption{Final-time profiles for initial $T=0.875T_{c}$ and $\varepsilon=0.001$.} 
\label{fig:case2:W0.00 T0.875 thickness nv T mu}
\end{figure}
Moreover, the final-time profiles with initial $T=0.95 T_{c}$ and $\varepsilon=0.001$ are plotted in \cref{fig:case2:W0.00 T0.95 nv T mu}.
In this case, the interface is widened and penetrated by a larger flux of fluid.
Compared to \cref{fig:case2:W0.00 T0.875nTmu} and \cref{fig:case2:W0.00 T0.875 thickness nv T mu}, although the distributions of $T$ and $\tilde{\mu}$ in \cref{fig:case2:W0.00 T0.95 nv T mu} exhibit slight differences, their homogeneity remains well preserved.
\cref{fig:case2:entropy mass}(b) and (c) illustrate the evolution of the total entropy for above two cases.
It is observed that the initial $T$ influences not only the interface thickness but also evolution process toward equilibrium.
For brevity, we only present the mass conservation result for the first case, as the others are similar.
\begin{figure}[htbp]
\centering
\subfloat[$n$]{\includegraphics[width=0.495\textwidth]{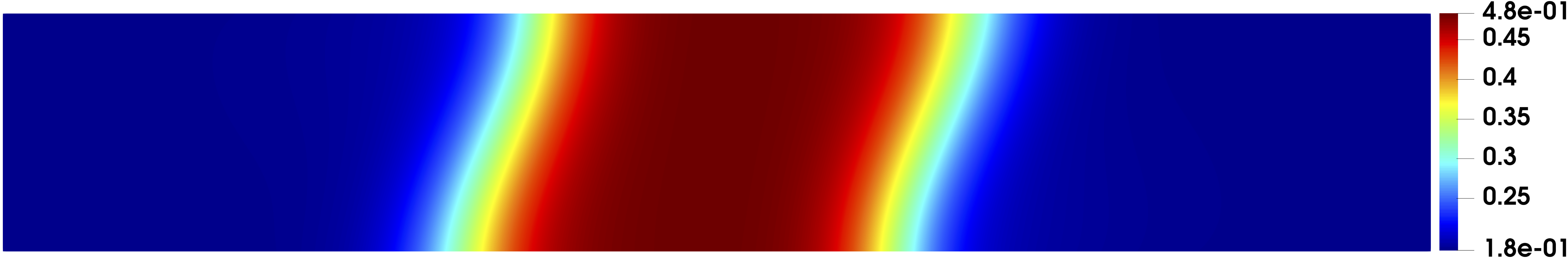}}
\subfloat[$\mathbf{v}$]{\includegraphics[width=0.495\textwidth]{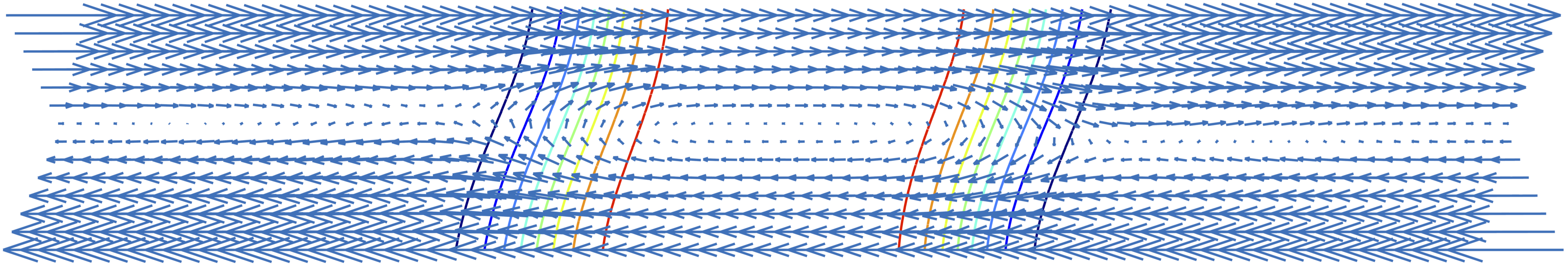}}\\
\subfloat[$T$]{\includegraphics[width=0.495\textwidth]{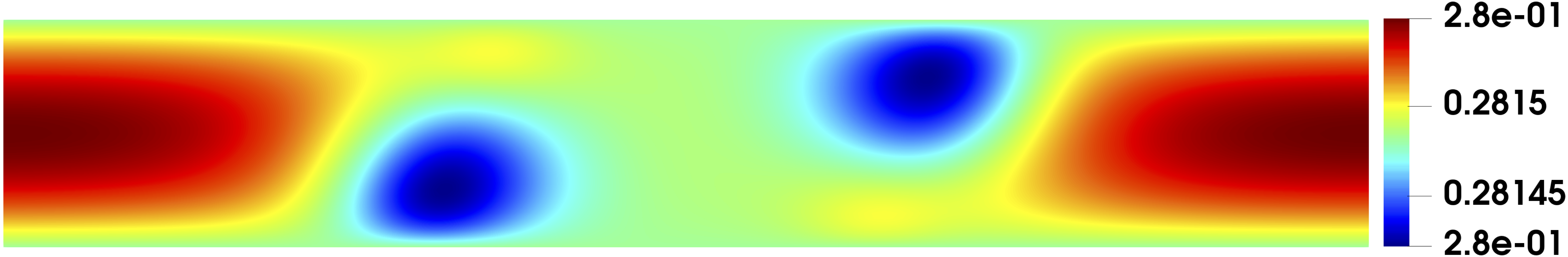}}
\subfloat[$\tilde{\mu}$]{\includegraphics[width=0.495\textwidth]{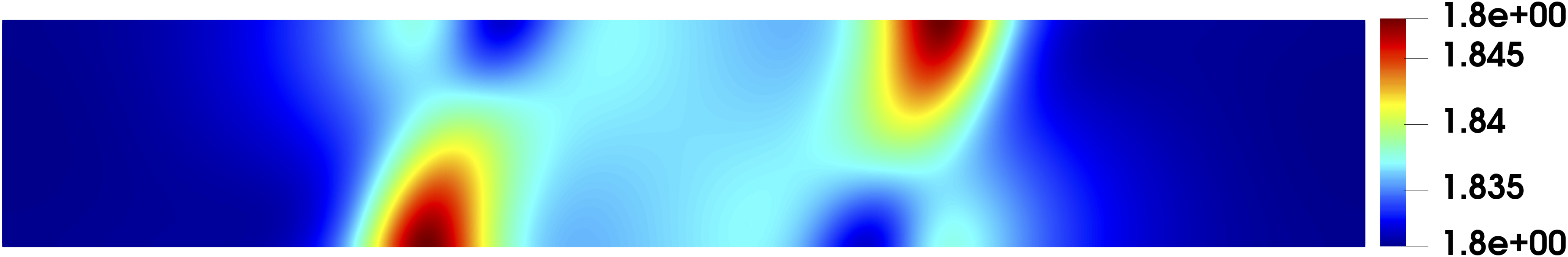}}
\caption{Final-time profiles for initial $T=0.95T_{c}$ and $\varepsilon=0.001$.} 
\label{fig:case2:W0.00 T0.95 nv T mu}
\end{figure}
\begin{figure}[htbp]
\centering
\subfloat[initial $T=0.875T_{c}$, $\varepsilon=0.0001$]{\includegraphics[width=0.24\textwidth]{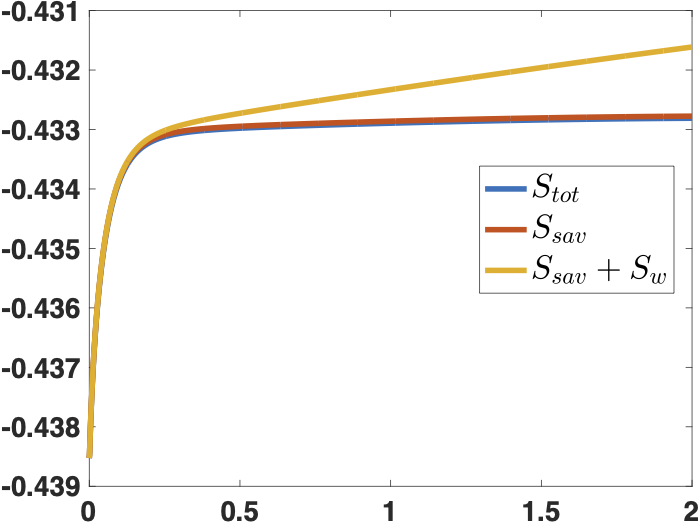}}
\hspace{0.5mm}
\subfloat[initial $T=0.875T_{c}$, $\varepsilon=0.001$]{\includegraphics[width=0.24\textwidth]{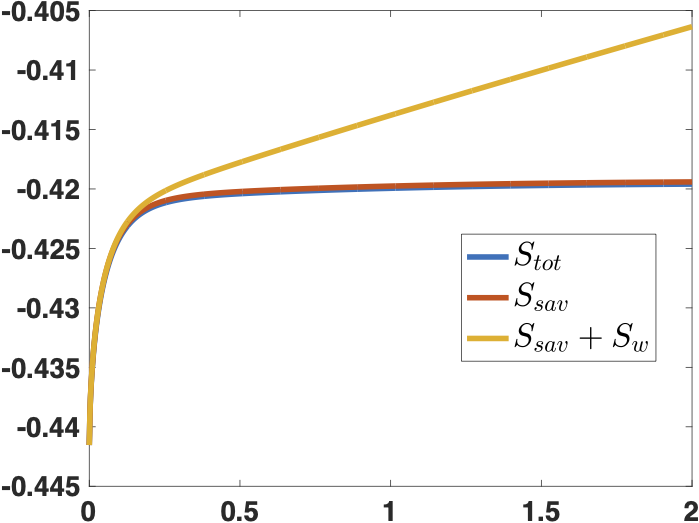}}
\hspace{0.5mm}
\subfloat[initial $T=0.95T_{c}$, $\varepsilon=0.001$]{\includegraphics[width=0.24\textwidth]{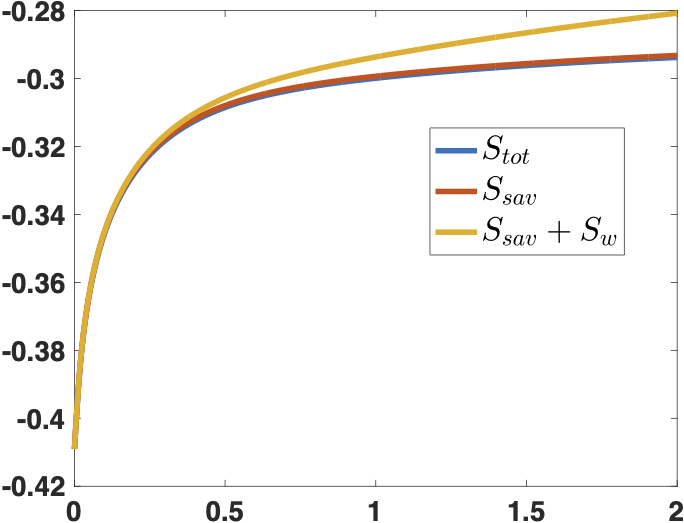}}
\hspace{0.5mm}
\subfloat[total mass error]{\includegraphics[width=0.24\textwidth]{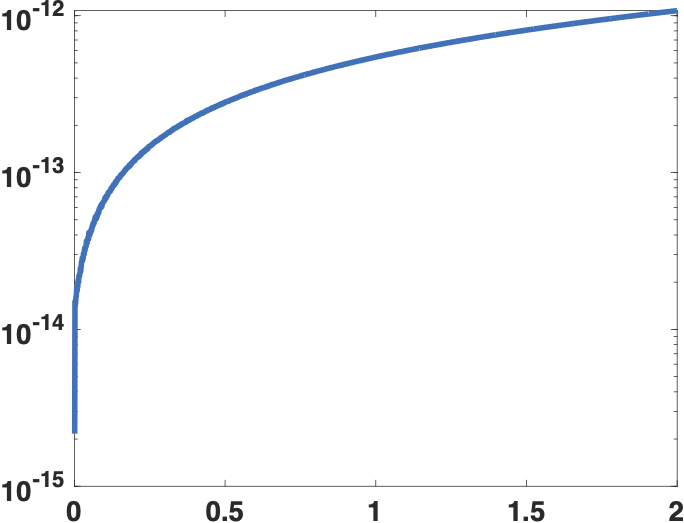}}
\caption{Evolution of total entropy and total mass error.} 
\label{fig:case2:entropy mass}
\end{figure}
To more clearly demonstrate the impact of initial $T$ and $\mathcal{W}$, \cref{fig:case2:angle} shows the final-time level curves of $(n_g + n_l)/2$ for initial $T = 0.875T_c$ and $0.95T_c$ with $\mathcal{W} = 0$, and $\mathcal{W} = -0.005$ and $0.005$ with initial $T = 0.875T_c$.
It is observed that the initial $T$ affects the width of the interface, while $\mathcal{W}$ influences the static contact angle of the liquid phase.
\begin{figure}[htbp]
\centering
\subfloat[different initial $T$]{\includegraphics[width=0.4\textwidth]{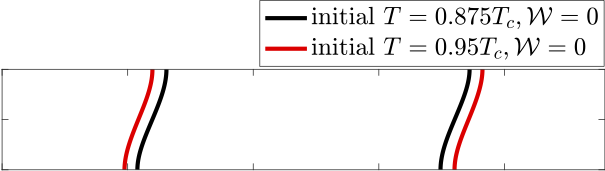}}
\hspace{5mm}
\subfloat[different $\mathcal{W}$]{\includegraphics[width=0.4\textwidth]{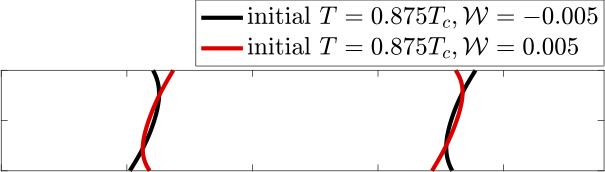}}
\caption{Level curves of $(n_g + n_l)/2$ at final time: (a) different initial $T$ with same $\mathcal{W}$; (b) different $\mathcal{W}$ with same initial $T$.}
\label{fig:case2:angle}
\end{figure}

\subsection{Droplet motion on solid substrates}
\label{subsec:Droplet motion on solid substrates}
In this example, we focus on cases where the droplet motion is induced by (\romannumeral 1) cooled, ordinary, and heated solid substrates, (\romannumeral 2) a wettability gradient at solid substrate, (\romannumeral 3) a thermal gradient at solid substrate.
A droplet is initially placed on the bottom solid surface, surrounded by gas, with no body forces introduced.
The hydrodynamic boundary conditions \eqref{eq:dimensionless_bcd} are imposed at top and bottom boundaries in this example.
Parameters used are $\Omega= [0, 1] \times [0, 0.5]$, $\mathcal{L}_{d}=0.0005$, $\mathcal{L}_{b}=50$, $\mathcal{R}_{e}=1$, $\mathcal{R}=0.06$, $\mathcal{L}_{s}=1 / n_{l}$, $\mathcal{L}_{\kappa}=1 / (n_{l}T_{c}^{2})$, $\mathcal{L}_{\lambda}=0.001n_{l}$, $\varepsilon = 0.0001$, $v_{w}=0$, $\zeta_{gl}=0.5$, $\zeta=\beta$, $\kappa$, $\lambda_{s}$.
\begin{figure}[htbp]
\centering
\subfloat[example 3(\romannumeral 1)]{\includegraphics[width=0.25\textwidth]{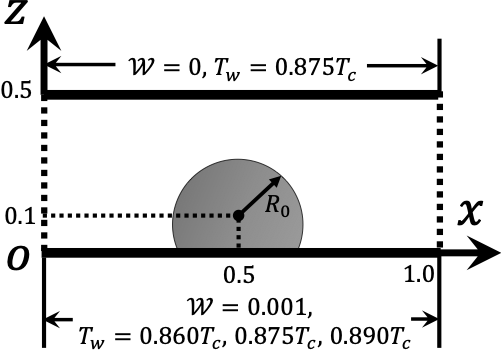}}
\subfloat[example 3(\romannumeral 2)]{\includegraphics[width=0.25\textwidth]{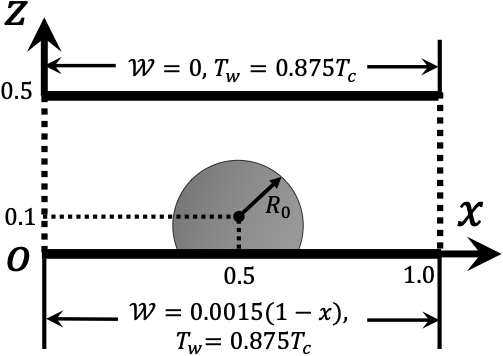}}
\subfloat[example 3(\romannumeral 2)]{\includegraphics[width=0.25\textwidth]{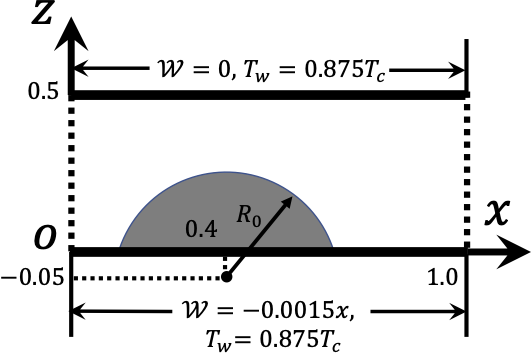}}
\subfloat[example 3(\romannumeral 3)]{\includegraphics[width=0.25\textwidth]{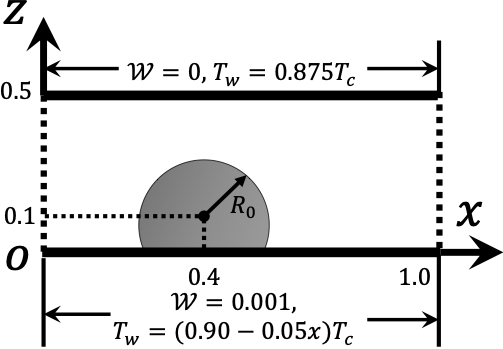}}
\caption{Schematic illustration of the system with different solid substrates.} 
\label{fig:case3:solid}
\end{figure}

\textbf{(\romannumeral 1) Cooled, ordinary, and heated solid substrates:}
First, we simulate cases where the bottom substrate temperature differs from the gas-liquid coexistence temperature \cite{xu2012droplet}.
As shown in \cref{fig:case3:solid}(a), three different bottom substrate temperatures are considered: $0.86T_{c}$, $0.875T_{c}$ and $0.89T_{c}$ for the cooled, ordinary and heated substrate, respectively.
For all cases, the top substrate temperature is fixed at $0.875T_{c}$, $\mathcal{W}=0$ for the top substrate and $\mathcal{W}=0.001$ for the bottom substrate.
The droplet radius $R_{0} = 0.1$, and initial $T = 0.875T_{c}$.
The final time $T_{f}=2$.

\cref{fig:case3:droplet T nv} shows the final-time profiles of $n$ and $\mathbf{v}$ (left column), $T$ (middle column) and $\tilde{\mu}$ (right column) for three cases.
It can be observed in \cref{fig:case3:droplet T nv}(a) that for droplet on the cooled substrate, condensation, characterized by a converging velocity field, occurs in a narrow region near the contact line.
In comparison, as shown in \cref{fig:case3:droplet T nv}(g), for droplet on the heated substrate, evaporation, characterized by a diverging velocity field, is prominent in a narrow region near the contact line.
For droplet on the ordinary substrate, local evaporation is observed from \cref{fig:case3:droplet T nv}(d), though it is less prominent compared to that for the heated substrate.
The effect of bottom substrate temperature can be observed in \cref{fig:case3:droplet T nv}(b), (e), and (h).
\cref{fig:case3:droplet T nv}(c), (f), and (i) show the homogeneity of $\tilde{\mu}$ at final time.
\begin{figure}[htbp]
\centering
\subfloat[$n$ (color) and $\mathbf{v}$ (arrow)]{\includegraphics[width=0.325\textwidth]{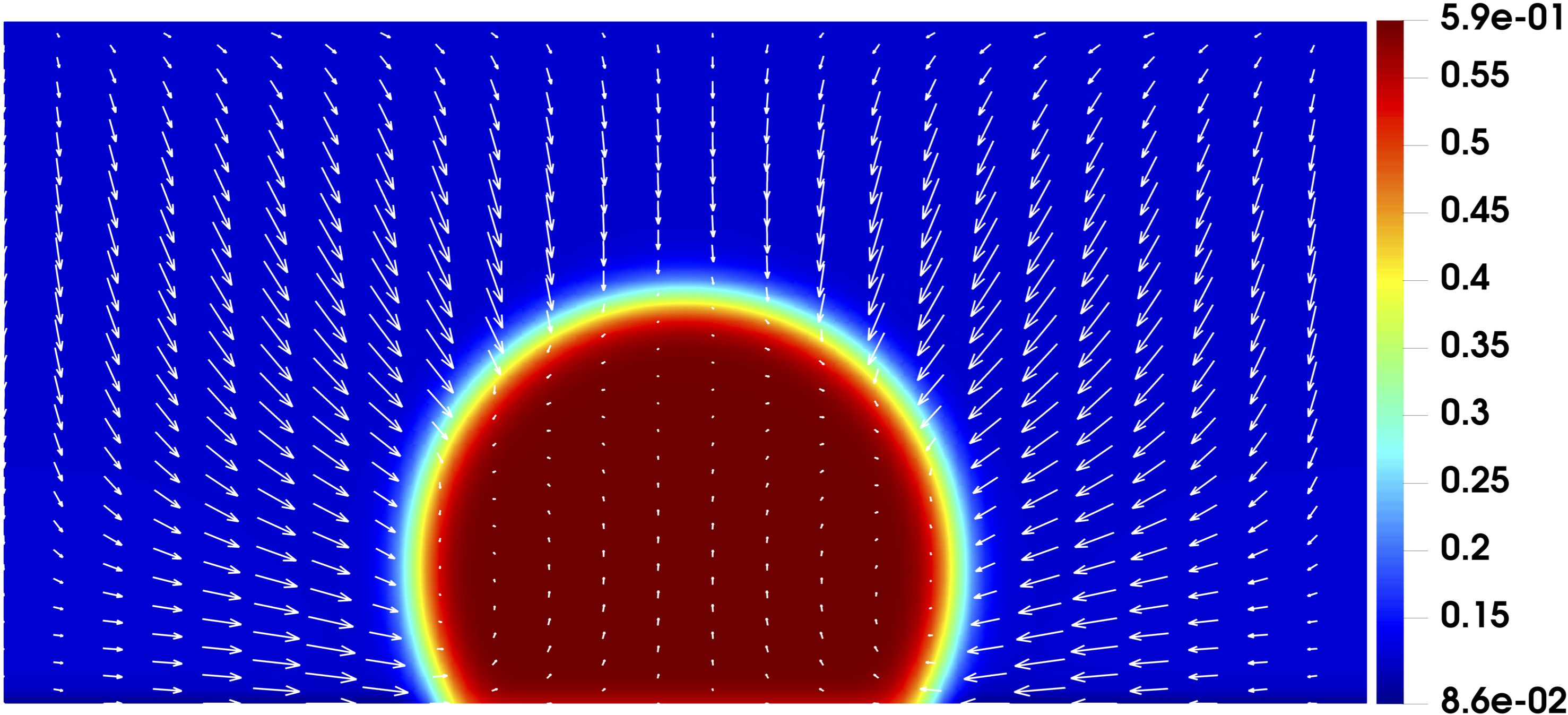}}
\hspace{0.5mm}
\subfloat[$T$]{\includegraphics[width=0.325\textwidth]{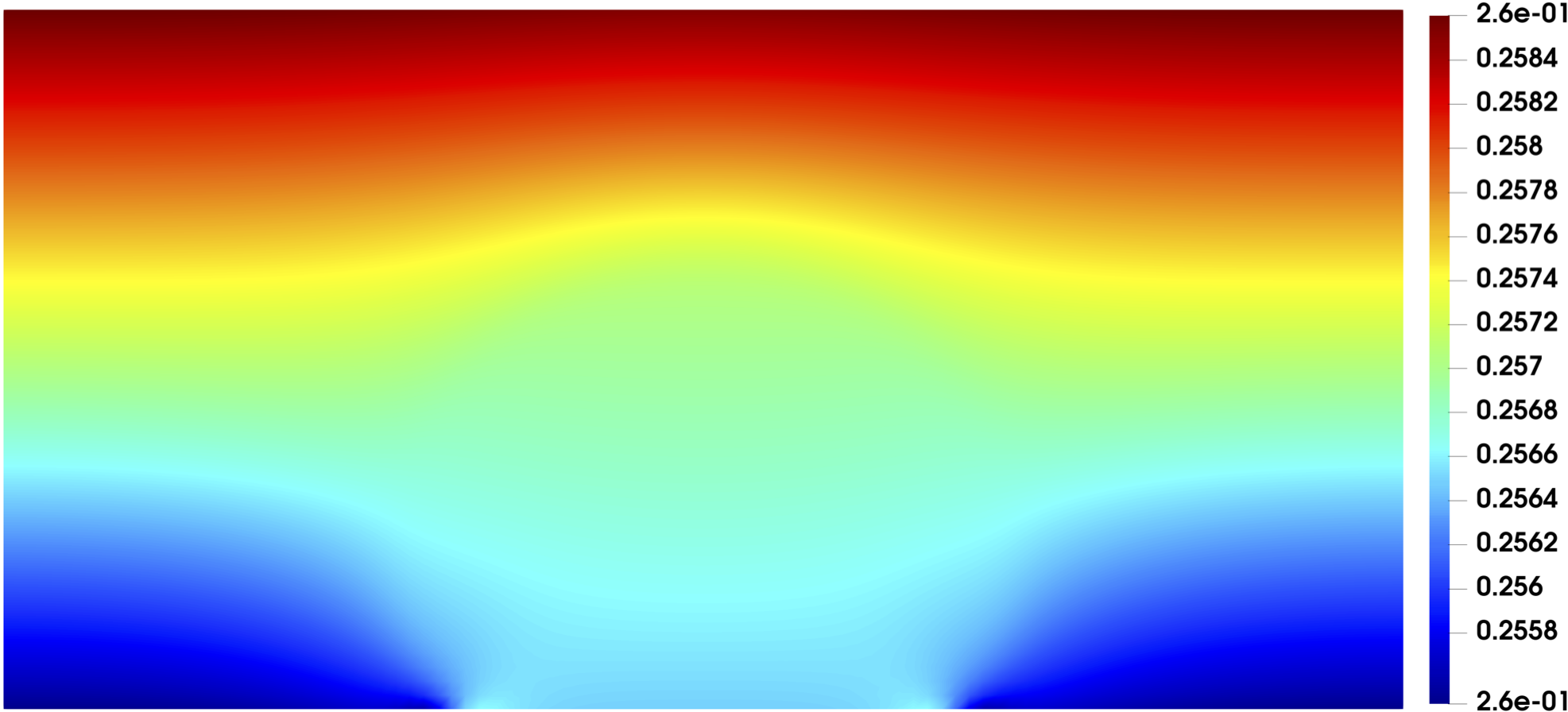}}
\hspace{0.5mm}
\subfloat[$\tilde{\mu}$]{\includegraphics[width=0.325\textwidth]{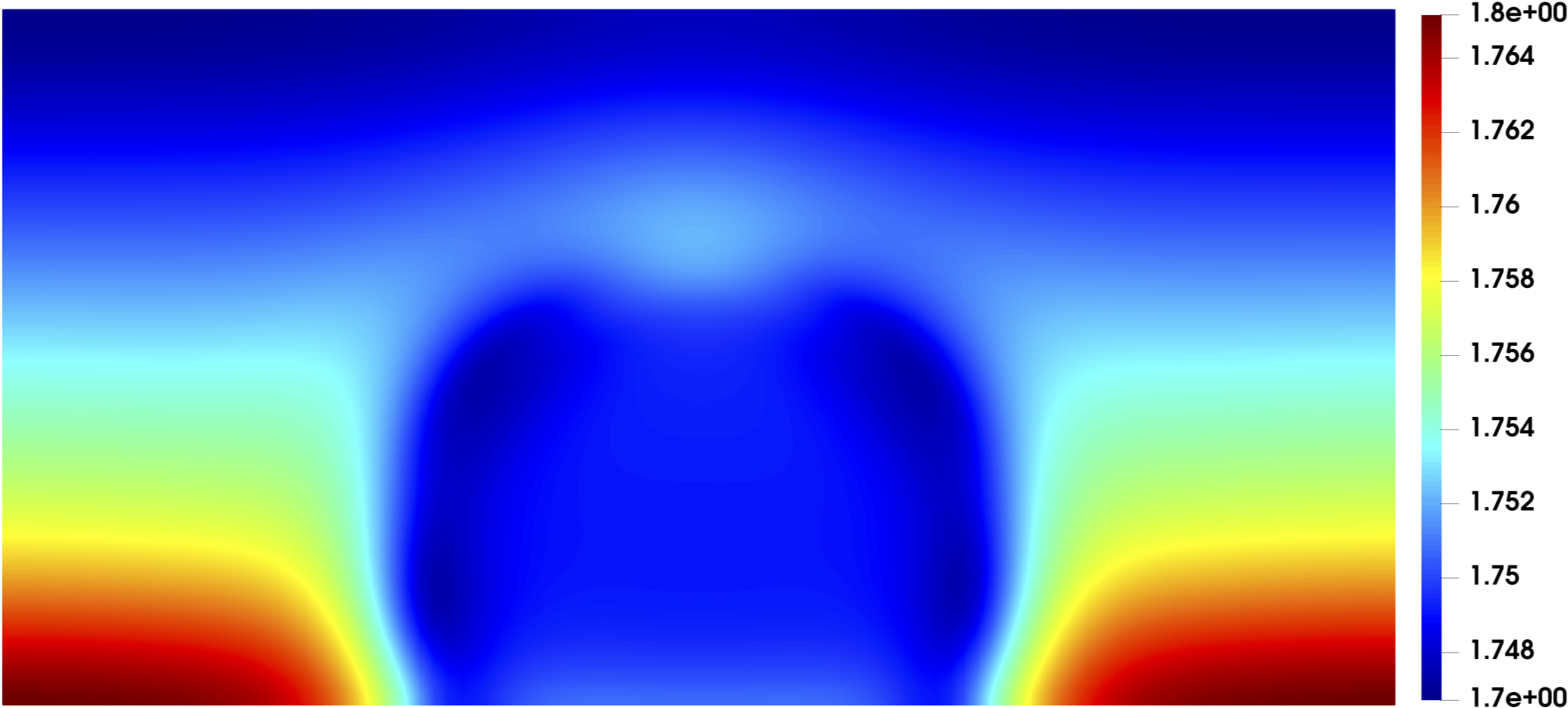}}\\
\subfloat[$n$ (color) and $\mathbf{v}$ (arrow)]{\includegraphics[width=0.325\textwidth]{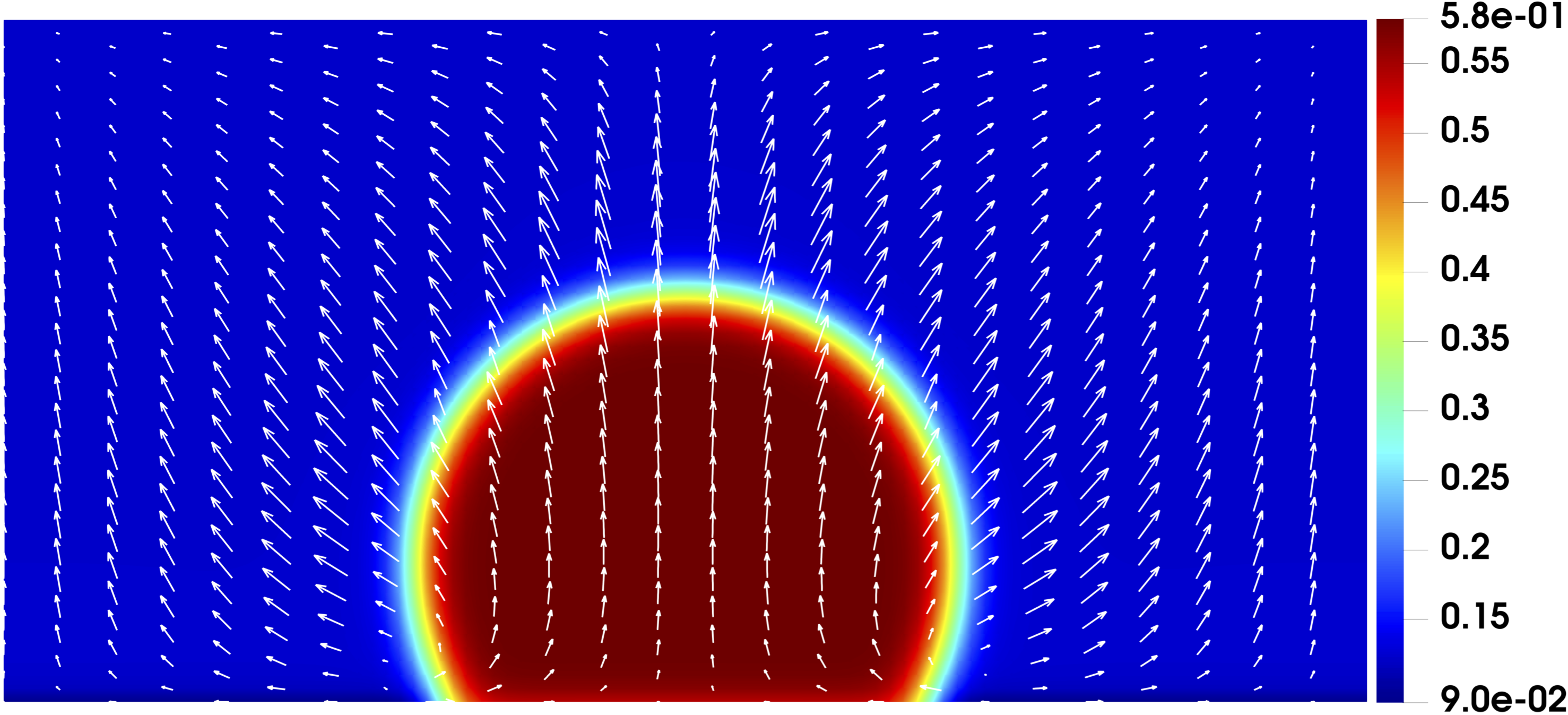}}
\hspace{0.5mm}
\subfloat[$T$]{\includegraphics[width=0.325\textwidth]{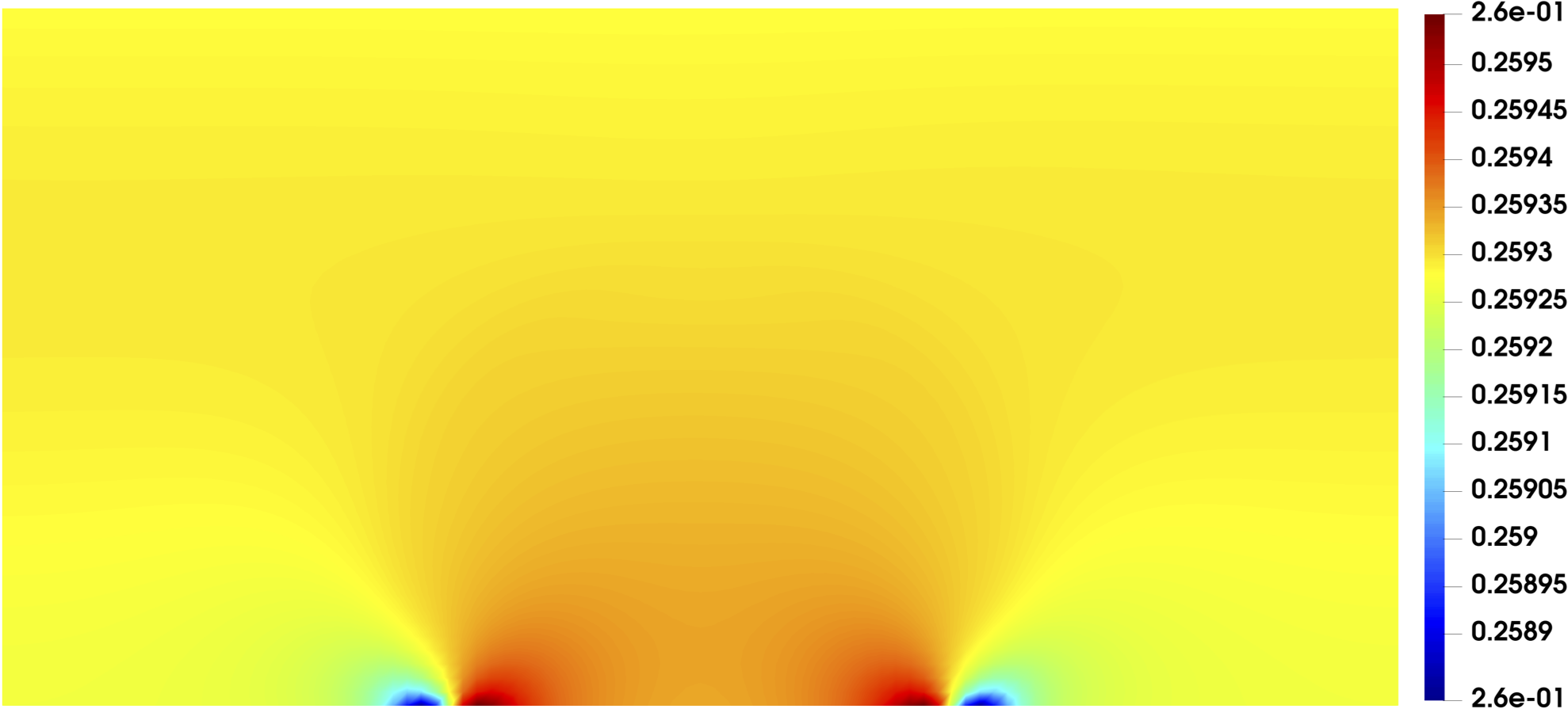}}
\hspace{0.5mm}
\subfloat[$\tilde{\mu}$]{\includegraphics[width=0.325\textwidth]{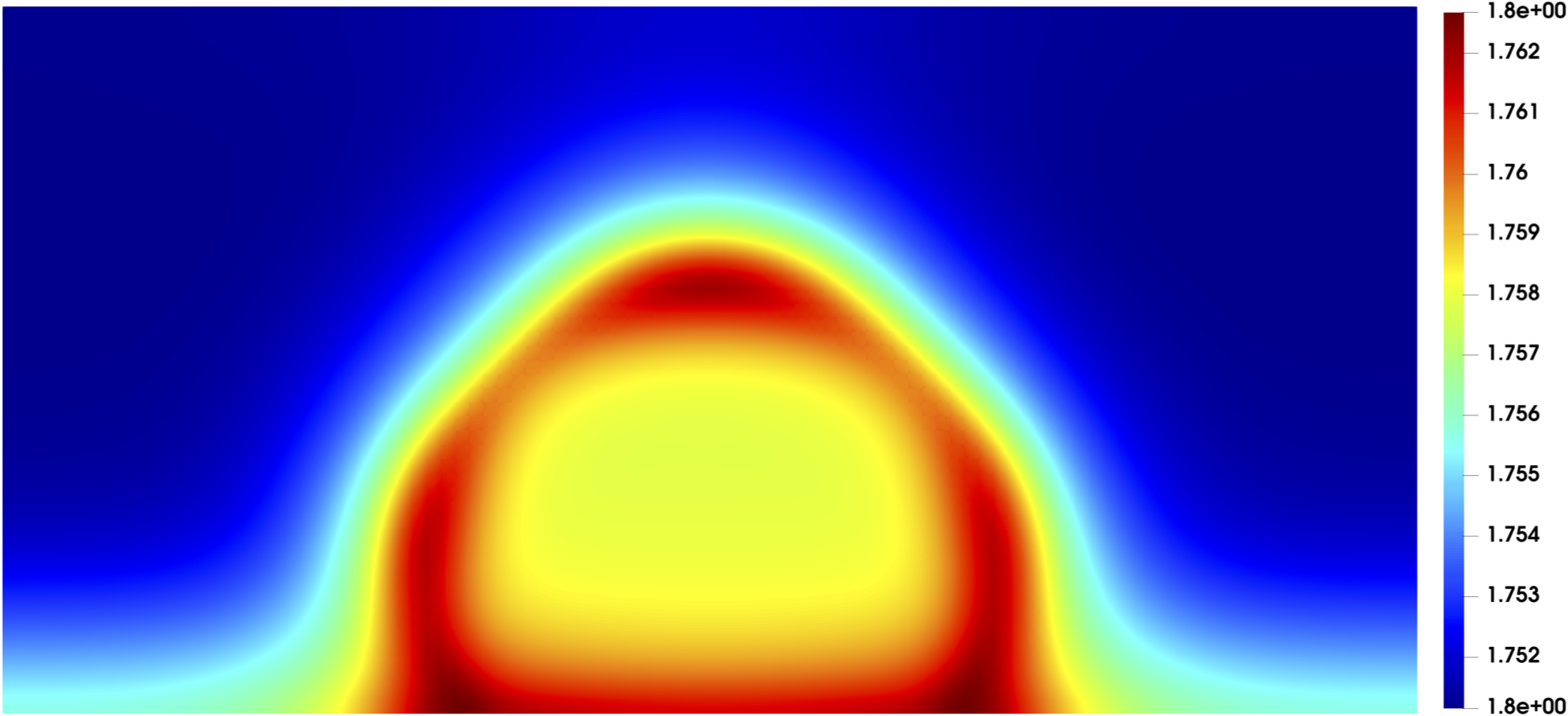}}\\
\subfloat[$n$ (color) and $\mathbf{v}$ (arrow)]{\includegraphics[width=0.325\textwidth]{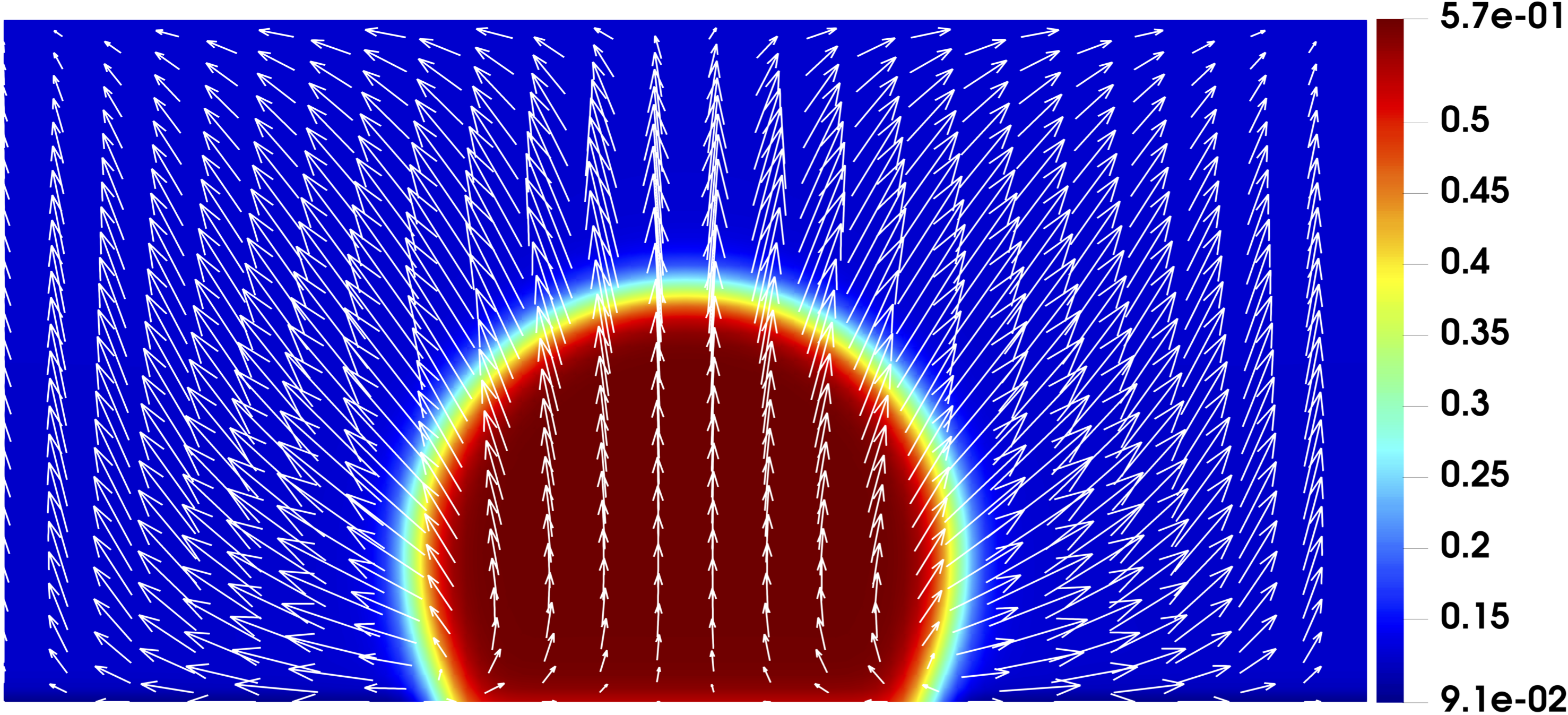}}
\hspace{0.5mm}
\subfloat[$T$]{\includegraphics[width=0.325\textwidth]{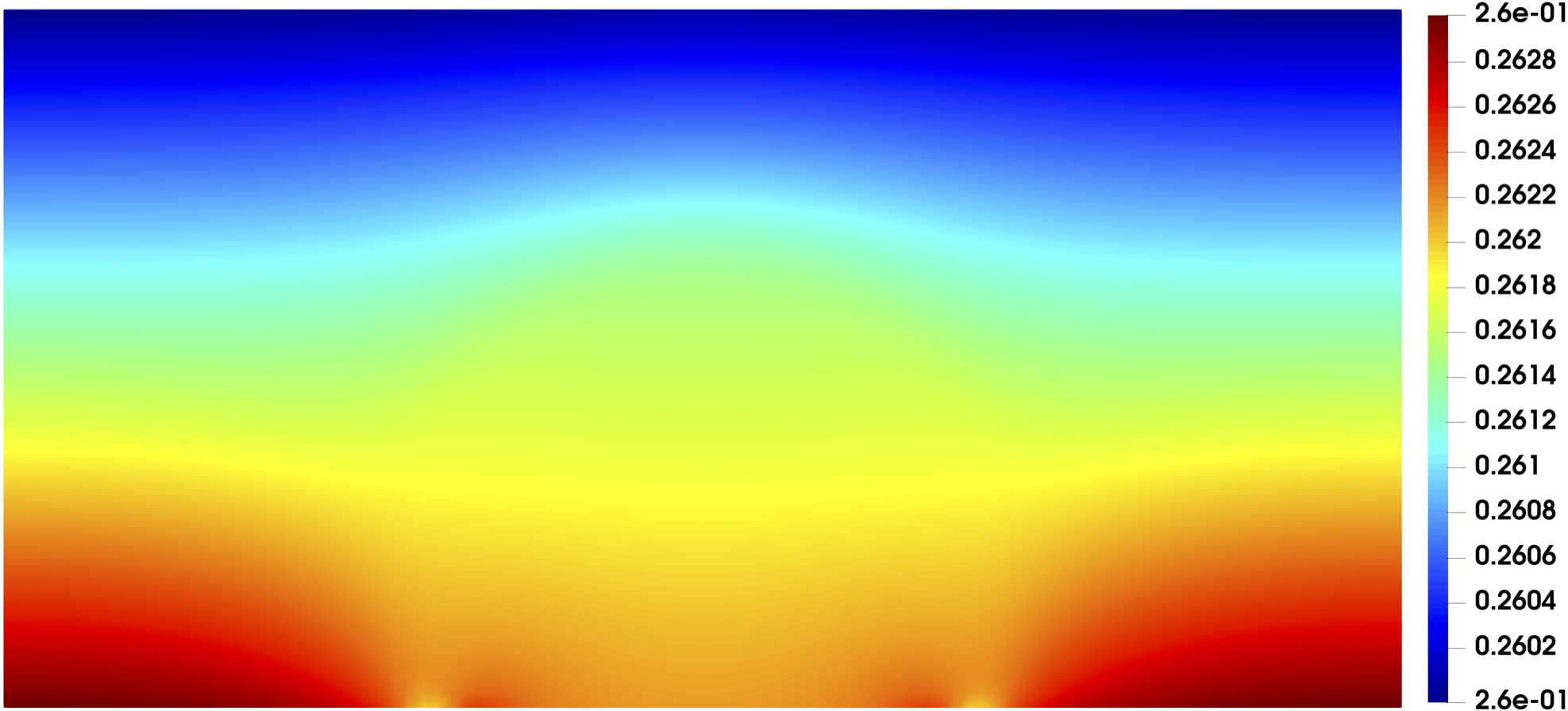}}
\hspace{0.5mm}
\subfloat[$\tilde{\mu}$]{\includegraphics[width=0.325\textwidth]{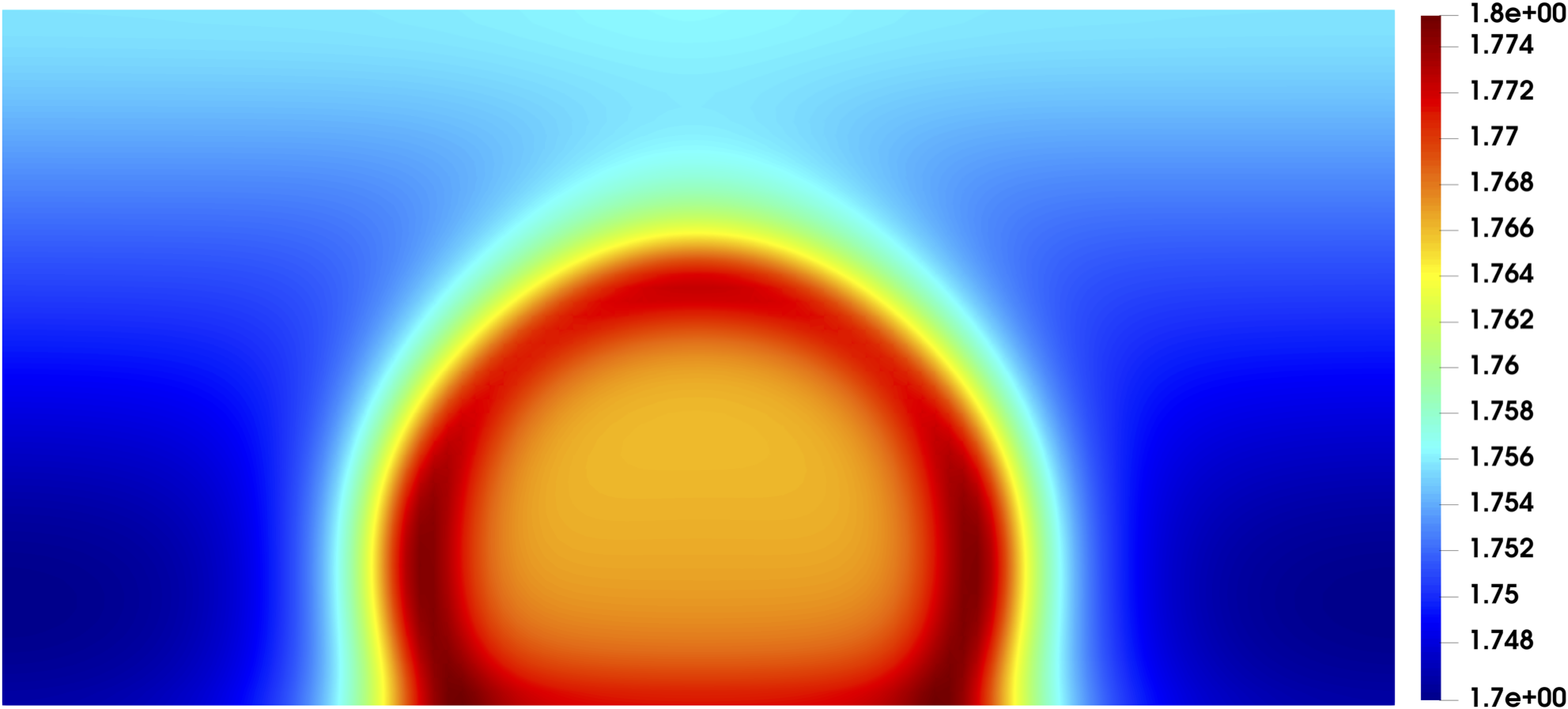}}
\caption{Final-time profiles: (a)-(c) the cooled substrate; (d)-(f) the ordinary substrate; and (g)-(i) the heated substrate.} 
\label{fig:case3:droplet T nv}
\end{figure}
When hydrodynamic boundary conditions are imposed, the outgoing entropy (denoted by $S_{\text{out}}$) should be considered. 
We define $S_{\text{out}}^{k}=\sum_{i=1}^{k}\delta t\int_{\Gamma} (\mathcal{R}_{e}^{-1}\partial_{\tau}(q_{s}^{i}/T^{i})+\mathbf{q}_{w}^{i-1}\cdot\bm{\gamma}/\mathcal{R}_{e}T_{w}+q^{i-1\ast}\mathcal{W}\partial_{\tau}(\sigma_{s}^{i}v_{\tau}^{i-1}))dA$.
As shown in \cref{thm:secondlaw_dc}, $S_{\text{sav}}+S_{\text{out}}$ should increase over time.
\cref{fig:case3:droplet T entropy Mass} shows the evolution of total entropy and total mass error.
\begin{figure}[htbp]
\centering
\subfloat[cooled substrate]{\includegraphics[width=0.24\textwidth]{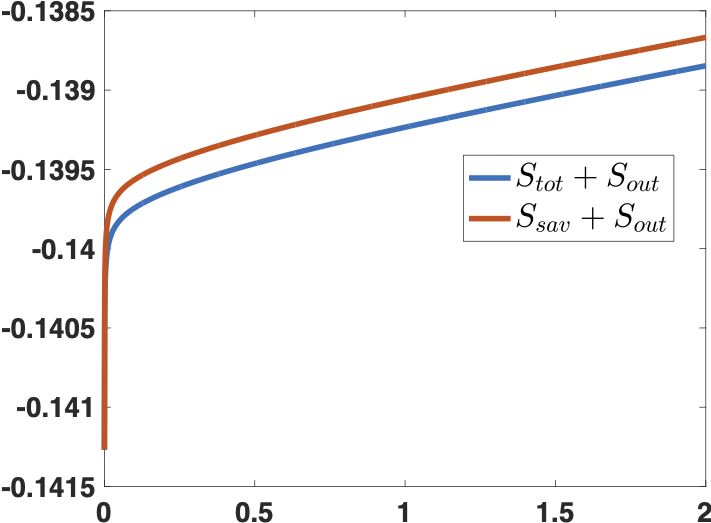}}
\hspace{0.5mm}
\subfloat[ordinary substrate]{\includegraphics[width=0.24\textwidth]{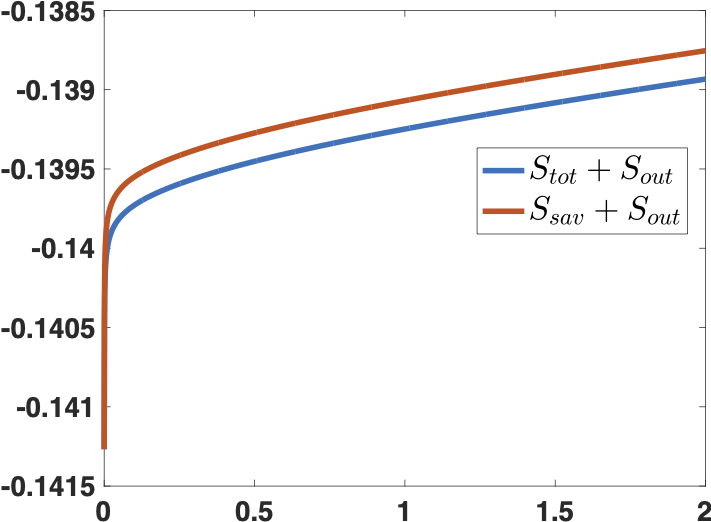}}
\hspace{0.5mm}
\subfloat[heated substrate]{\includegraphics[width=0.24\textwidth]{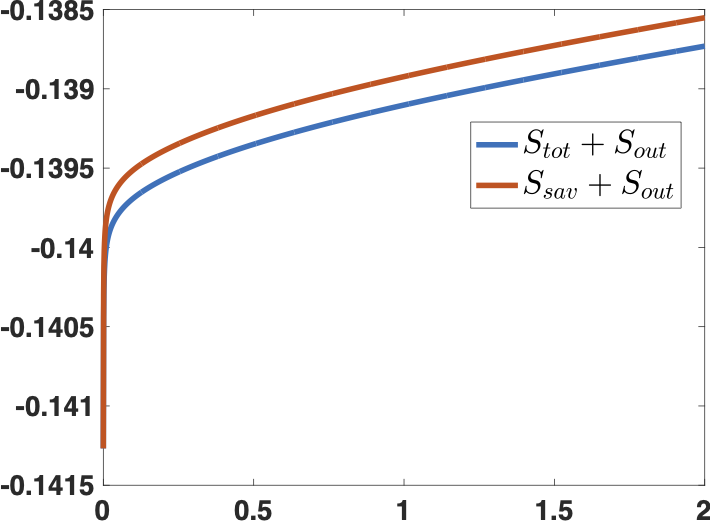}}
\hspace{0.5mm}
\subfloat[total mass error]{\includegraphics[width=0.24\textwidth]{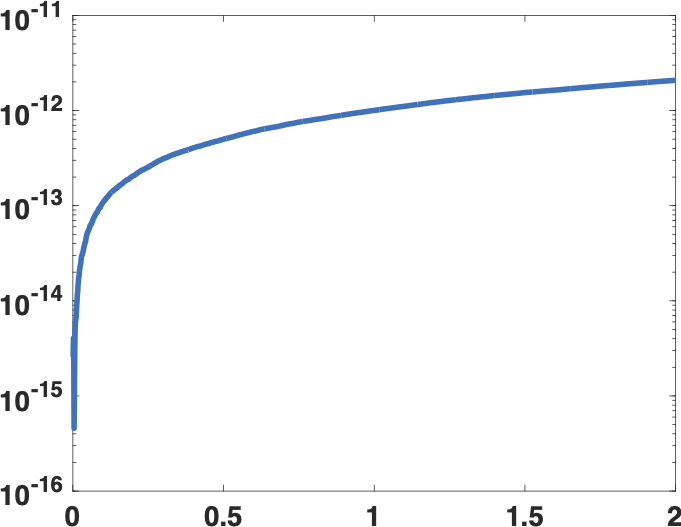}}
\caption{Evolution of total entropy and total mass error.} 
\label{fig:case3:droplet T entropy Mass}
\end{figure}

\textbf{(\romannumeral 2) Solid substrate with a wettability gradient:}
Next, we present and discuss the simulation results of droplet migration on the solid substrate with a wettability gradient, considering both hydrophobic and hydrophilic cases \cite{xu2012droplet}.
As depicted in \cref{fig:case3:solid}(b) and (c), we set the wettability number on the bottom substrate as $\mathcal{W} = 0.0015(1-x)$ in the hydrophobic case and $\mathcal{W} = -0.0015x$ in the hydrophilic case, while $\mathcal{W} = 0$ on the top substrate in both cases.
The initial droplet radius $R_{0}=0.1$ in the hydrophobic case and $0.3$ in the hydrophilic case.
The initial $T = 0.875T_{c}$, $T_{w}=0.875T_{c}$, and final time $T_{f}=30$ in both cases.
\cref{fig:case3:droplet gradW nv} shows the profiles of $n$ and $\mathbf{v}$ at $t = 0$ and $30$, as well as the final-time profile of $\tilde{\mu}$ for both cases.
It is found that, due to the driving force arising from the wettability gradient in the $+x$-direction at the bottom substrate, the droplet at $t = 30$ shows significant movement and deformation in the $x$-direction.
For clarity, the level curves of $(n_g + n_l)/2$ at different times are plotted in \cref{fig:case3:droplet gradW angle entropy mass}(a) and (b), showing that the droplet moves faster in the hydrophilic case, which is quantitatively analyzed in \cite{xu2012droplet}.
\cref{fig:case3:droplet gradW angle entropy mass}(c) and (d) show the evolution of total entropy and total mass error for the hydrophobic case.
\begin{figure}[htbp]
\centering
\subfloat[$t = 0$]{\includegraphics[width=0.33\textwidth]{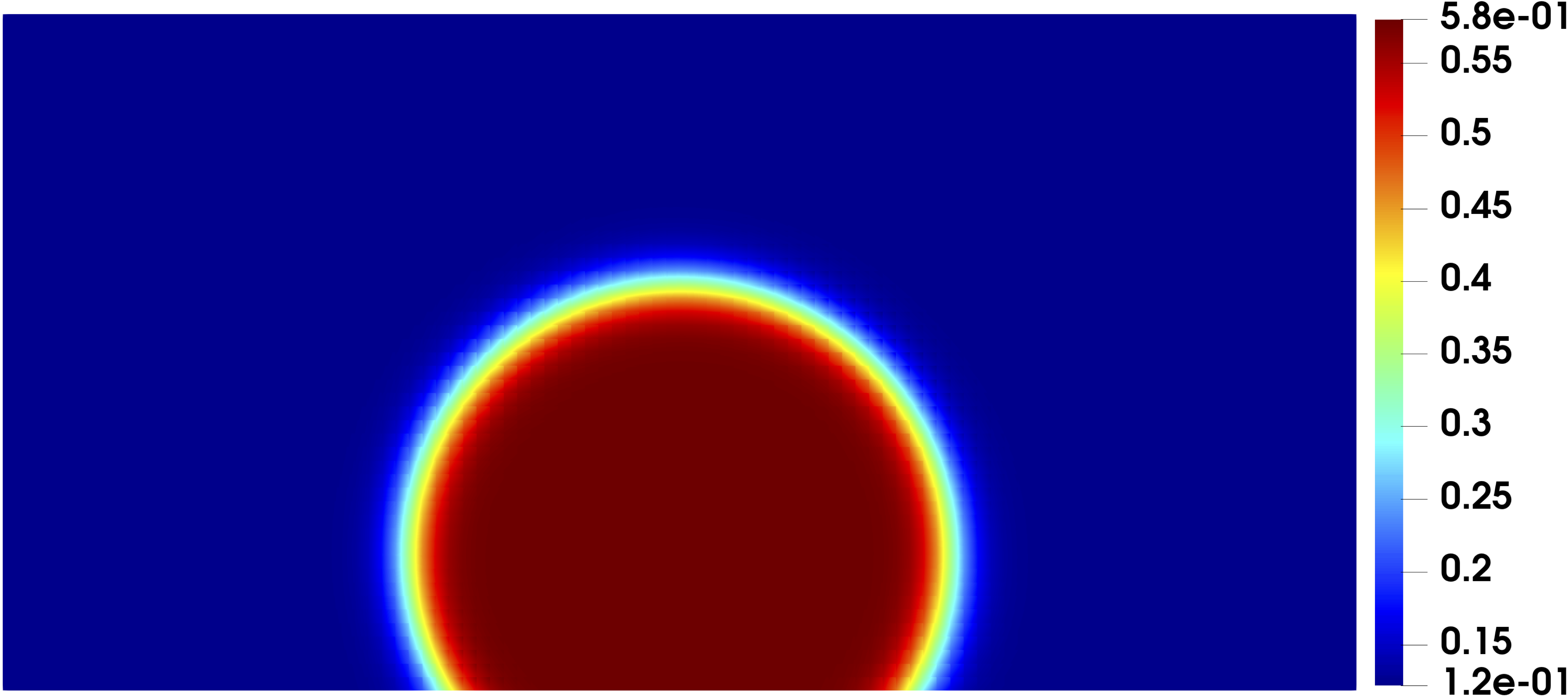}}
\subfloat[$t = 30$]{\includegraphics[width=0.33\textwidth]{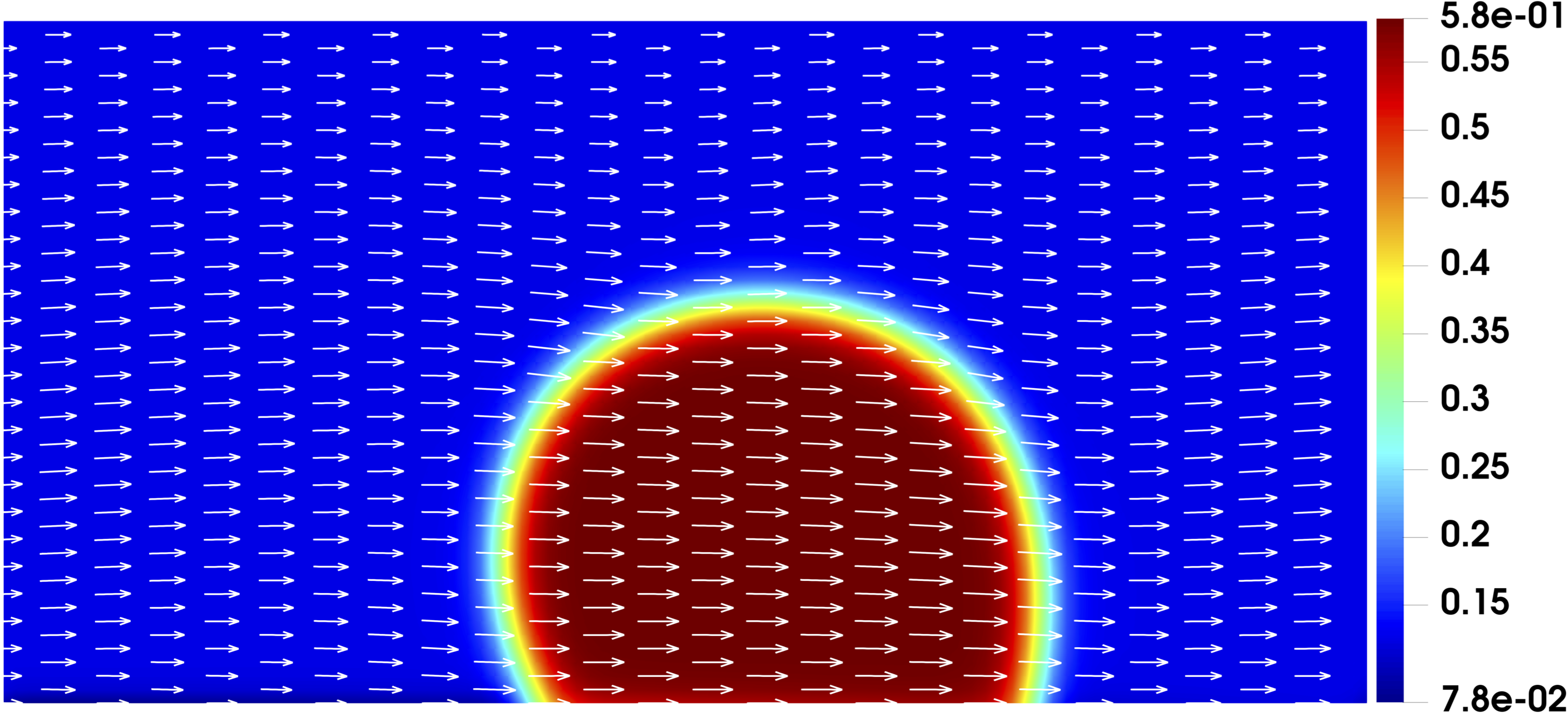}}
\subfloat[$\tilde{\mu}$]{\includegraphics[width=0.33\textwidth]{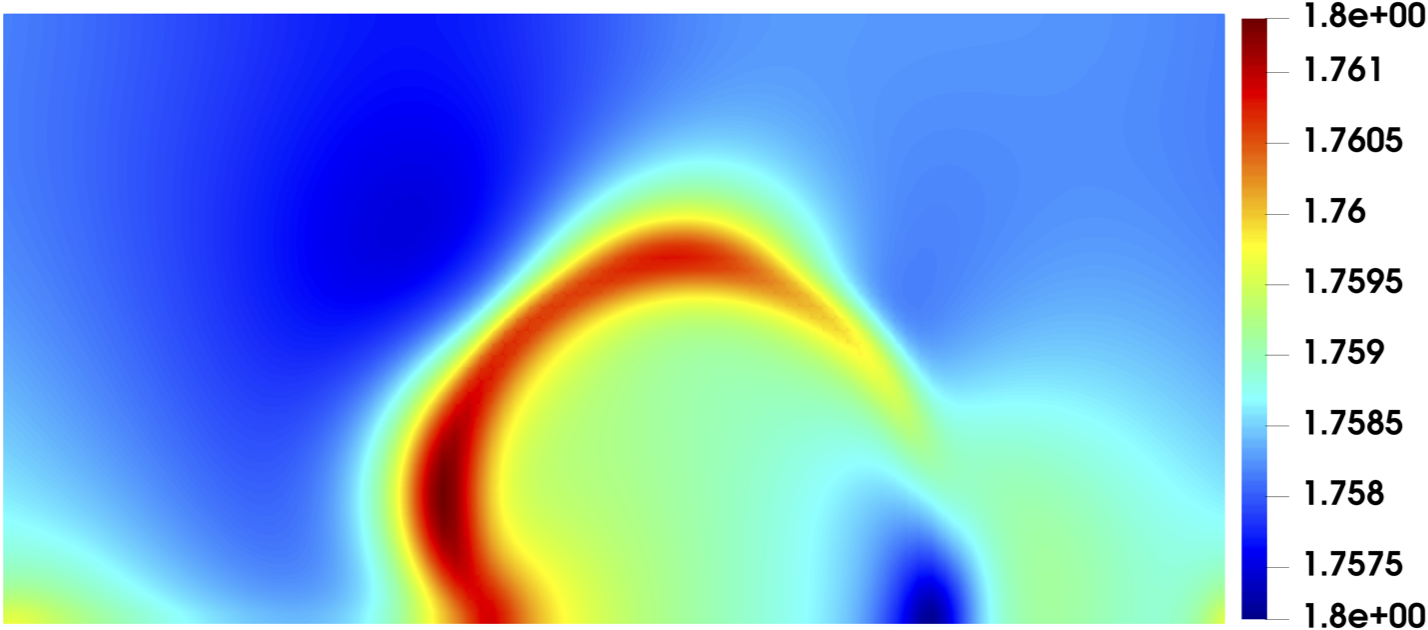}}\\
\subfloat[$t = 0$]{\includegraphics[width=0.33\textwidth]{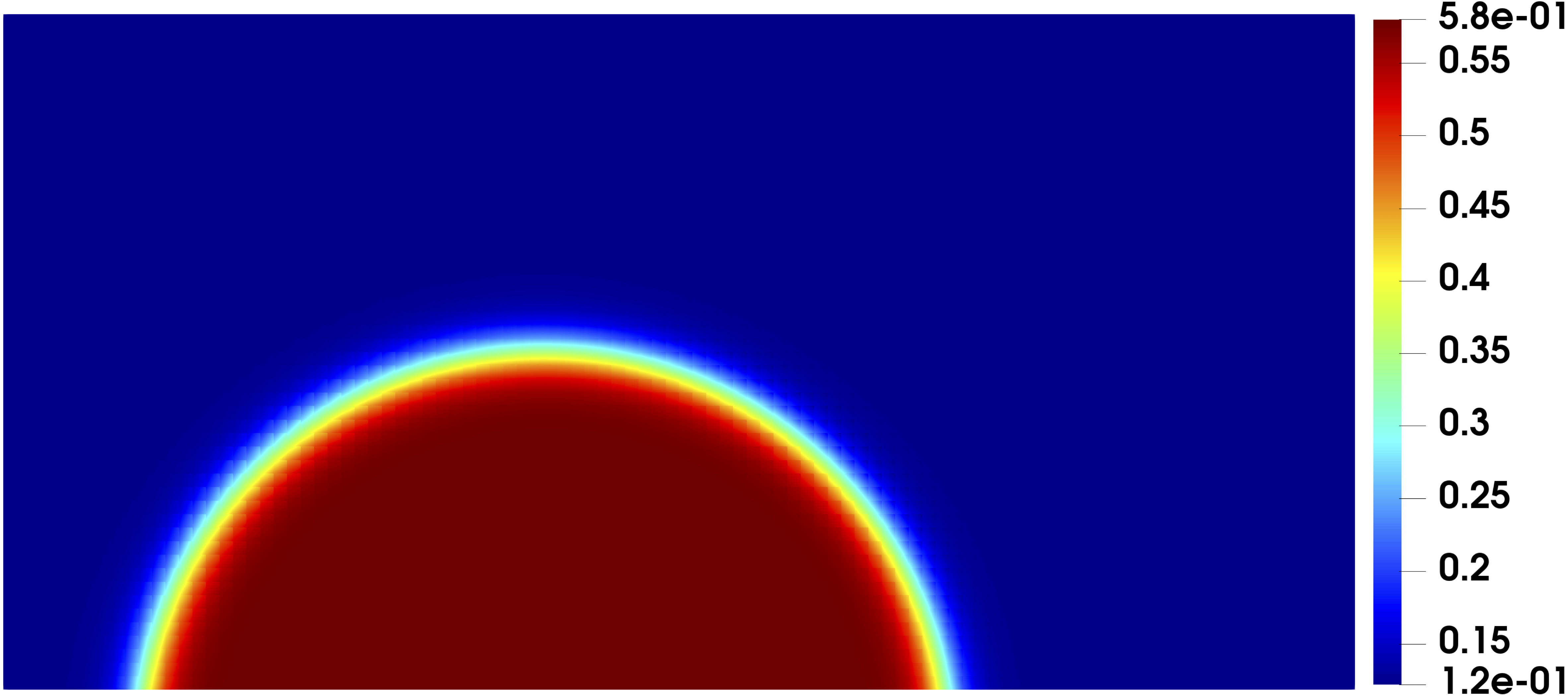}}
\subfloat[$t = 30$]{\includegraphics[width=0.33\textwidth]{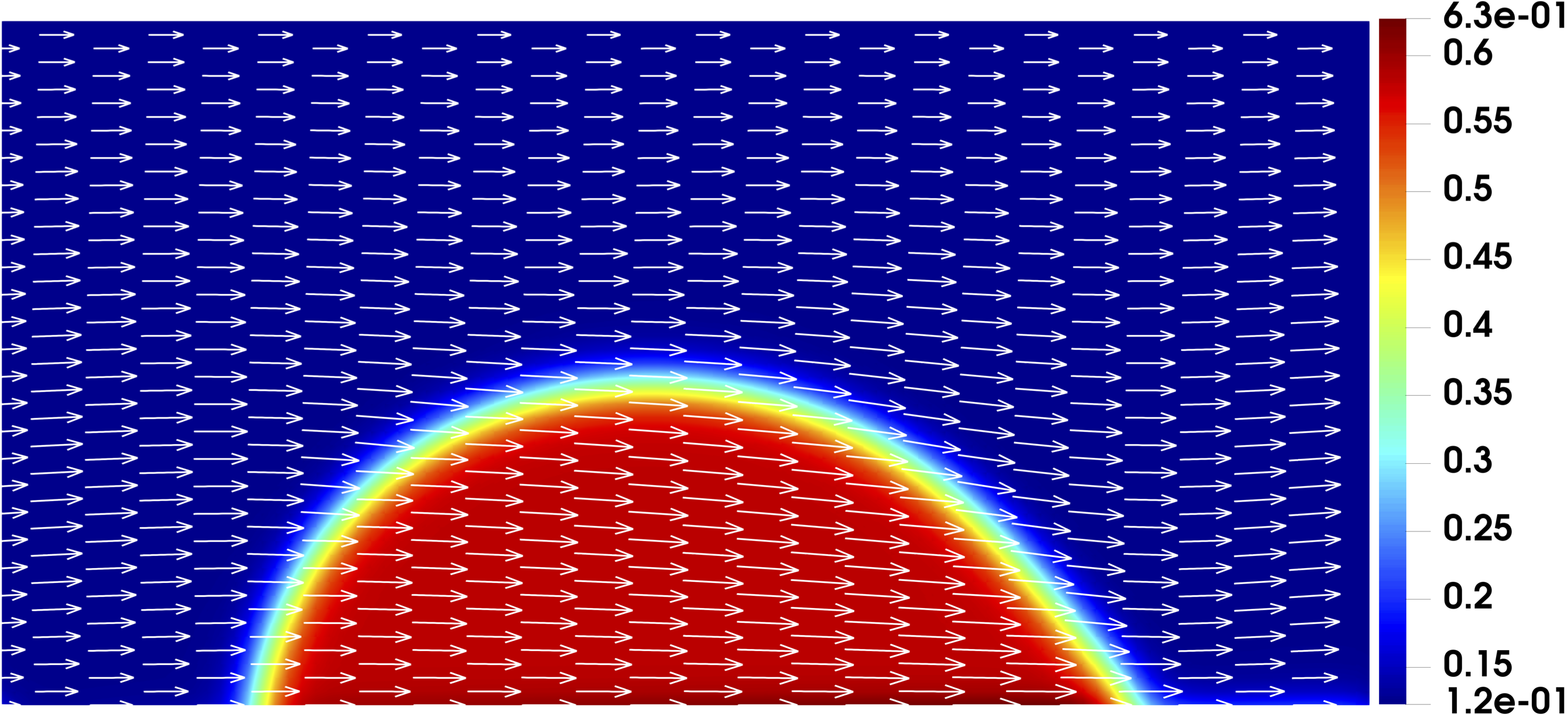}}
\subfloat[$\tilde{\mu}$]{\includegraphics[width=0.33\textwidth]{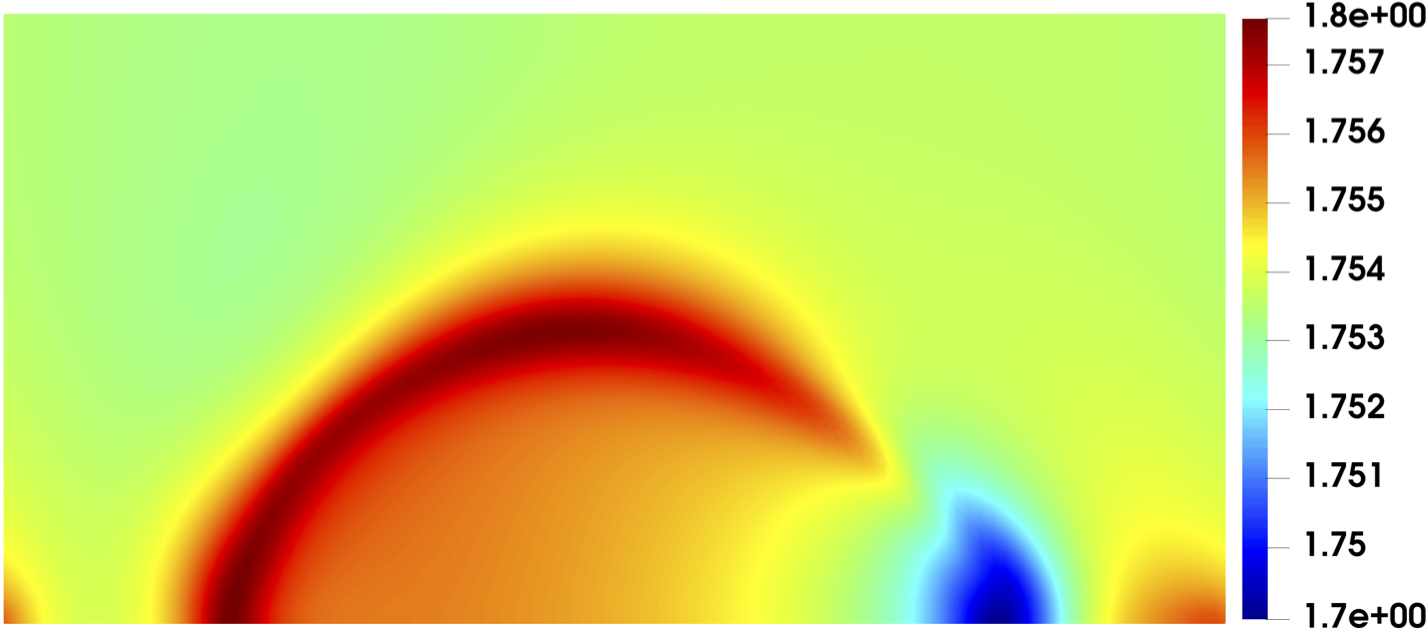}}
\caption{Profiles of $n$ (color) and $\mathbf{v}$ (arrow), and the final-time profile of $\tilde{\mu}$: (a)-(c) the hydrophobic case; (d)-(f) the hydrophilic case.} 
\label{fig:case3:droplet gradW nv}
\end{figure}
\begin{figure}[htbp]
\centering
\subfloat[hydrophobic case]{\includegraphics[width=0.25\textwidth]{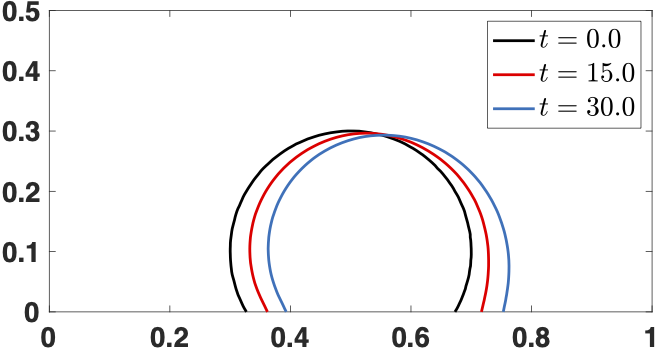}}
\hspace{0.1mm}
\subfloat[hydrophilic case]{\includegraphics[width=0.25\textwidth]{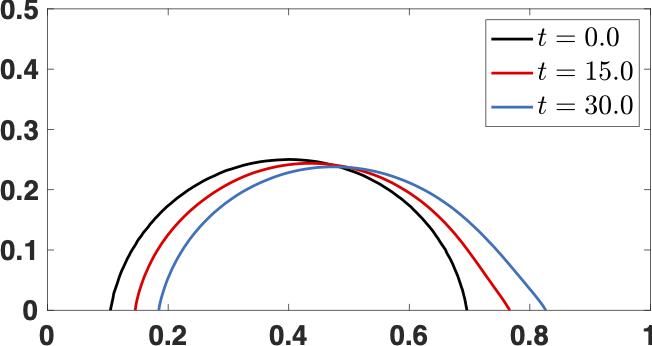}}
\hspace{0.1mm}
\subfloat[total entropy]{\includegraphics[width=0.235\textwidth]{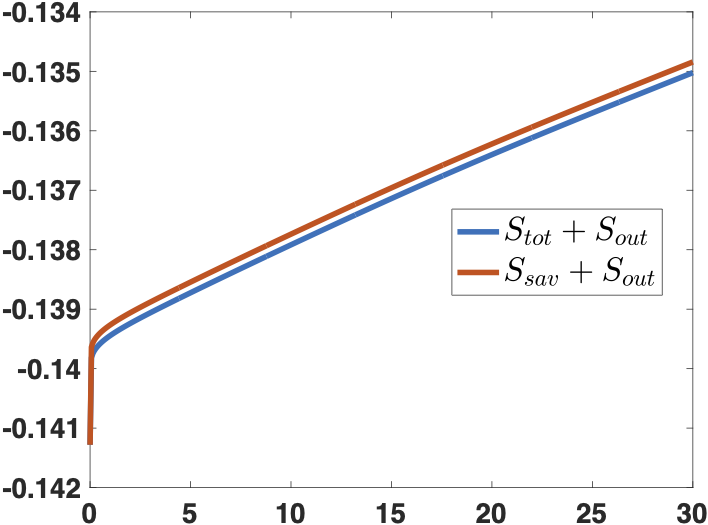}}
\hspace{0.1mm}
\subfloat[total mass error]{\includegraphics[width=0.235\textwidth]{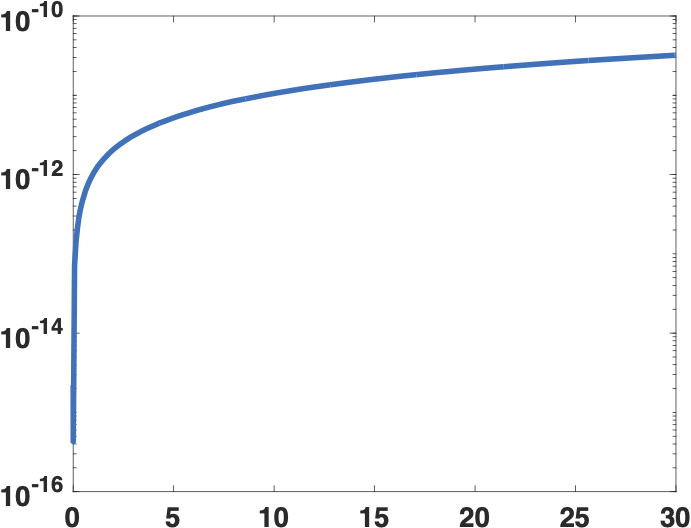}}
\caption{Level curves of $(n_g + n_l)/2$ at different times, and the evolution of total entropy and total mass error.} 
\label{fig:case3:droplet gradW angle entropy mass}
\end{figure}

\textbf{(\romannumeral 3) Solid substrate with a thermal gradient:}
In this case, we simulate the migration of a droplet on the solid substrate with a thermal gradient \cite{xu2012thermal}.
As shown in \cref{fig:case3:solid}(d), we set $R_{0}=0.1$, initial $T = 0.875T_{c}$, $\mathcal{W} = 0.001$ on the bottom substrate and $\mathcal{W} = 0$ on the top substrate.
The bottom substrate temperature varies as $T_{w}=(0.9-0.05x)T_{c}$, while the top substrate temperature remains at $0.875 T_{c}$.
The final time $T_{f}=30$.
\cref{fig:case3:droplet gradT nv angle} shows the profiles of $n$ and $\mathbf{v}$, and the level curves of $(n_g + n_l)/2$ at different times.
It is observed that, induced by the thermal gradient, the droplet spontaneously migrates from the hot region toward the cold region, and finally attains a steady state with an almost constant migration velocity \cite{xu2012thermal}.
\begin{figure}[htbp]
\centering
\subfloat[$t = 0$]{\includegraphics[width=0.33\textwidth]{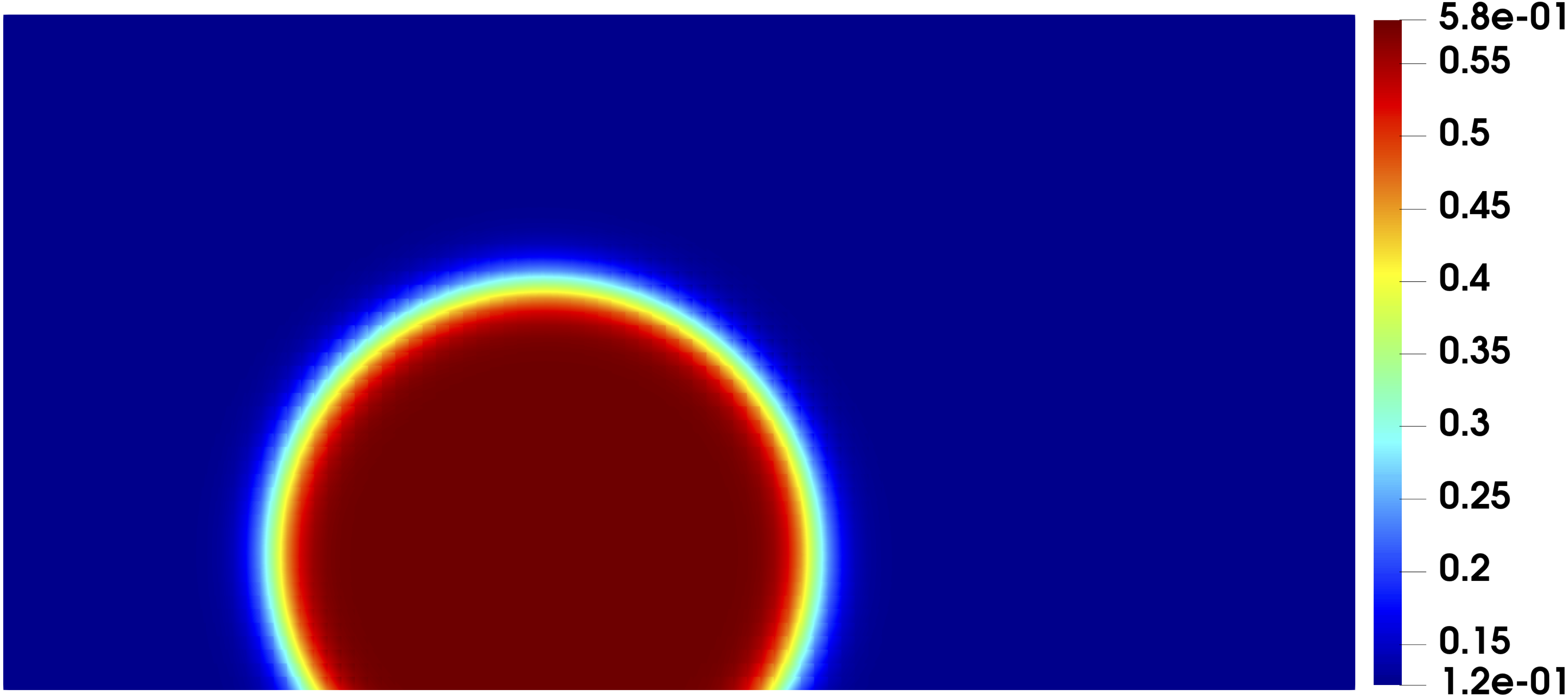}}
\hspace{0.5mm}
\subfloat[$t = 30$]{\includegraphics[width=0.33\textwidth]{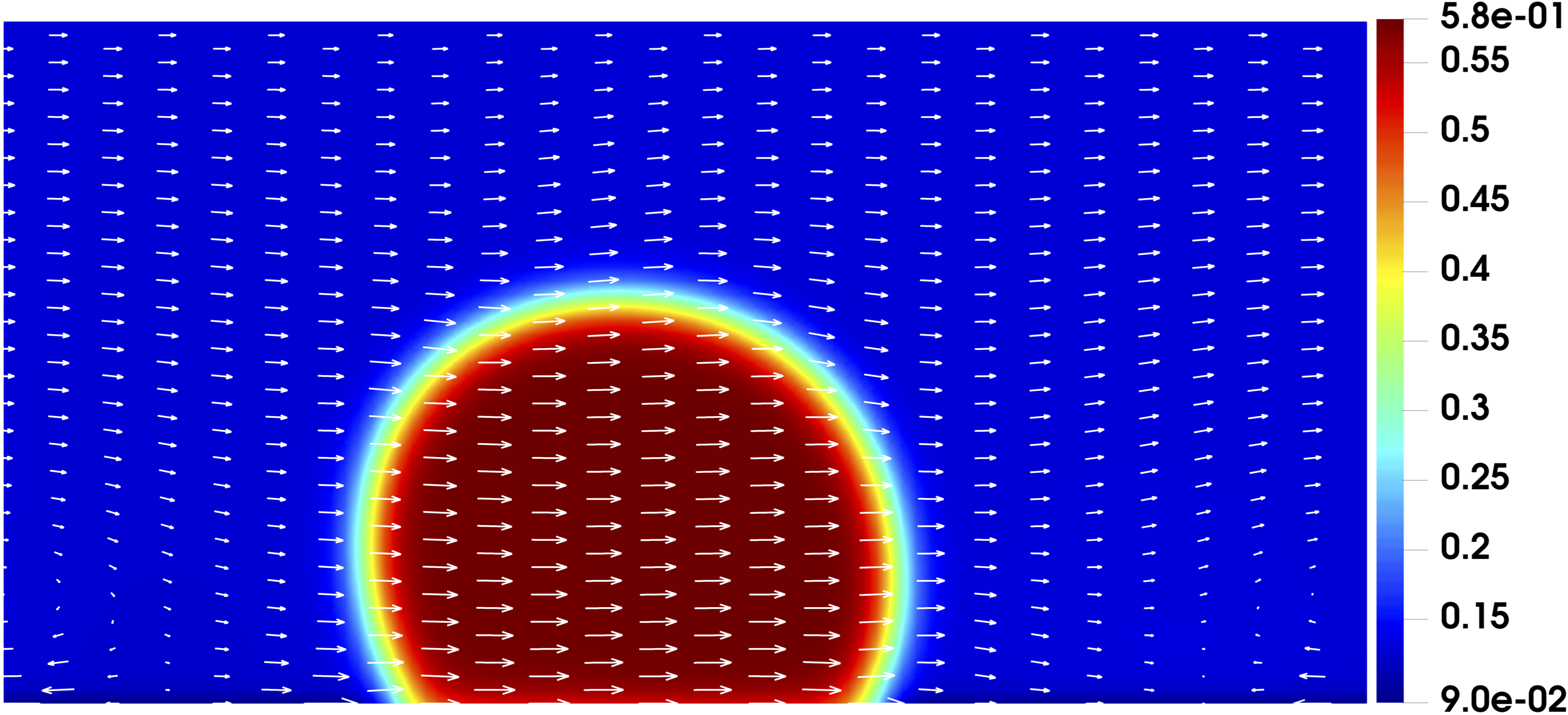}}
\hspace{0.5mm}
\subfloat[level curves of $(n_g + n_l)/2$]{\includegraphics[width=0.315\textwidth]{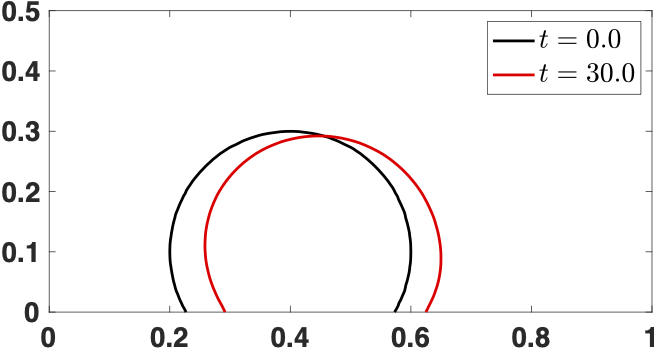}}
\caption{Profiles of $n$ (color) and $\mathbf{v}$ (arrow), and level curves of $(n_g + n_l)/2$ at different times.} 
\label{fig:case3:droplet gradT nv angle}
\end{figure}
The final-time profiles of $T$ and $\tilde{\mu}$ are plotted in \cref{fig:case3:droplet gradT T mu entropy mass}(a) and (b).
Due to the decrease of $T_{w}$ along the bottom substrate, the local evaporation on the hot region and the local condensation on the cold region can be clearly observed.
\cref{fig:case3:droplet gradT T mu entropy mass}(c) and (d) present the evolution of the total entropy and the total mass error, respectively.
\begin{figure}[htbp]
\centering
\subfloat[$T$]{\includegraphics[width=0.28\textwidth]{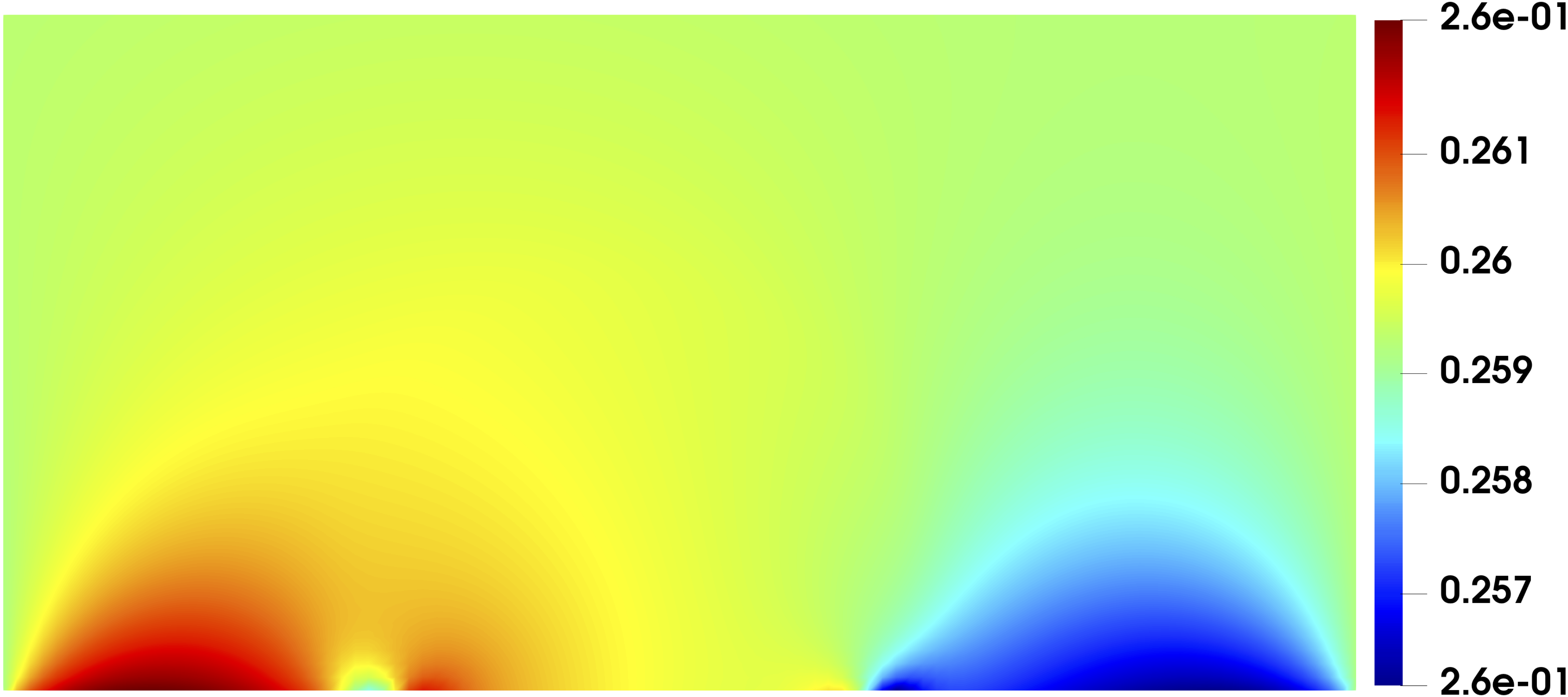}}
\subfloat[$\tilde{\mu}$]{\includegraphics[width=0.28\textwidth]{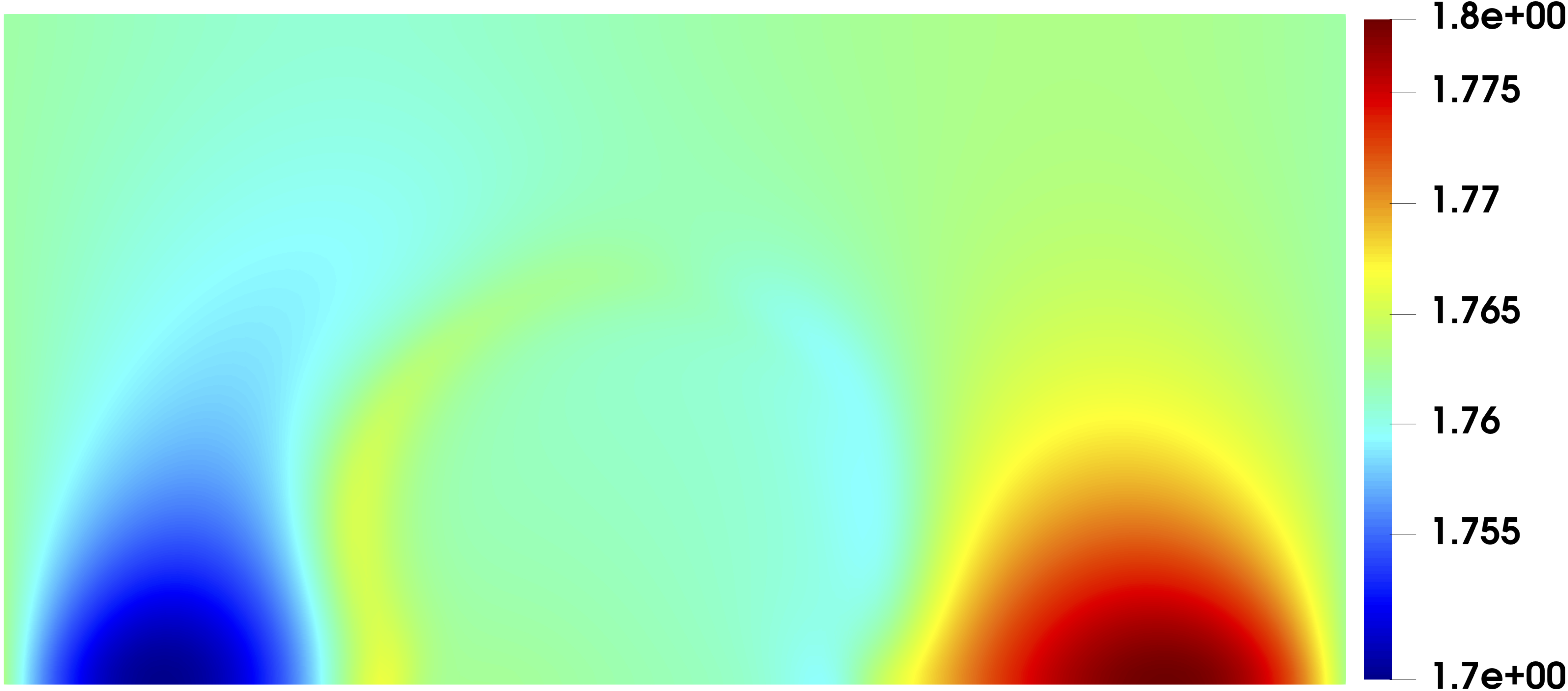}}
\subfloat[total entropy]{\includegraphics[width=0.22\textwidth]{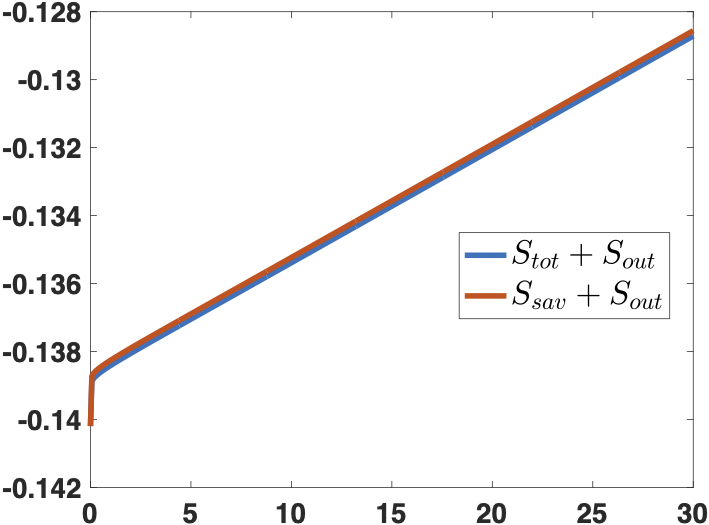}}
\subfloat[total mass error]{\includegraphics[width=0.22\textwidth]{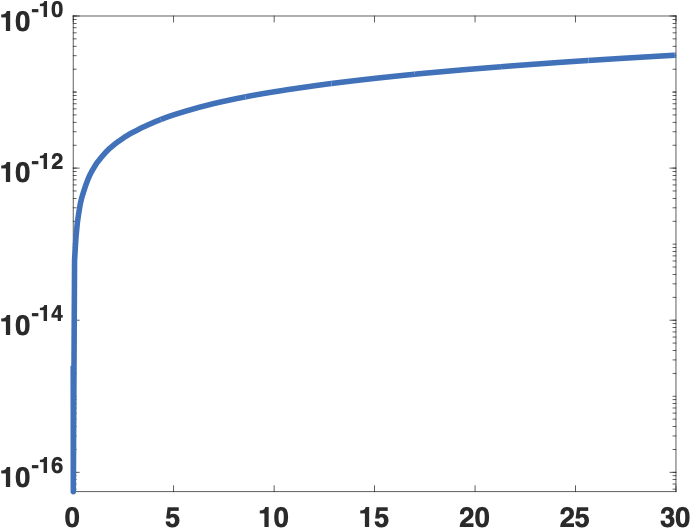}}
\caption{Final-time profiles of $T$ and $\tilde{\mu}$, and the evolution of total entropy and total mass error.} 
\label{fig:case3:droplet gradT T mu entropy mass}
\end{figure}

\section{Conclusions}
\label{sec:conclusions}
In this paper, we present a new thermodynamically consistent model for non-isothermal compressible two-phase flows with MCL.
The proposed model, based on the DVDWT, consists of equations for the number density, fluid velocity, and fluid temperature.
Since the fluid-solid interfaces are assumed to be non-isothermal and heterogeneous, we impose the hydrodynamic boundary conditions on the proposed model, which are able to describe both velocity slip and temperature slip.
The dimensionless form of the proposed model is developed and proved to strictly satisfy the fundamental laws of thermodynamics.
Based on the dimensionless model, we construct two numerical schemes.
The first is fully coupled and thermodynamically consistent, and is rigorously proved to satisfy the temporally discrete first and second laws of thermodynamics. 
By extending the MSAV approach to entropy production, we design the second scheme, which is decoupled, linear, and unconditionally entropy-stable.
It is also proved that this scheme satisfies the temporally discrete second law of thermodynamics.
The fully discrete scheme, constructed by combining the DG and CG methods, is proved to preserve the mass conservation law and the boundedness of the number density.
We conduct numerical simulations to verify the efficiency and stability of the proposed scheme.



\bibliographystyle{siamplain}
\bibliography{references}
\end{document}

%% file: non_iso_mcl_shared.tex
\usepackage{lipsum}
\usepackage{amsfonts}
\usepackage{graphicx}
\usepackage{epstopdf}
\usepackage{algorithmic}
\ifpdf
  \DeclareGraphicsExtensions{.eps,.pdf,.png,.jpg}
\else
  \DeclareGraphicsExtensions{.eps}
\fi

\newsiamremark{remark}{Remark}
\newsiamremark{hypothesis}{Hypothesis}
\crefname{hypothesis}{Hypothesis}{Hypotheses}
\newsiamthm{claim}{Claim}
\newsiamremark{fact}{Fact}
\crefname{fact}{Fact}{Facts}

\headers{Modelling and simulation of the MCL in non-isothermal flows}{Junkai Wang and Qiaolin He}

\title{Thermodynamically consistent modelling and simulation of the moving contact line problem in non-isothermal compressible two-phase flows\thanks{Submitted to the editors DATE.
\funding{This research is supported part by the National Key R \& D Program of China Under Grant (No. 2022YFE03040002), the National Natural Science Foundation of China (No.12371434).}}}

\author{
Junkai Wang\thanks{School of Mathematics, Sichuan University, Chengdu, 610064, P.R. China 
  (\email{wangjk0520@gmail.com}, \email{qlhejenny@scu.edu.cn})}
\and Qiaolin He\footnotemark[2]
}

\usepackage{amsopn}

\usepackage{bm}
\usepackage[caption=false]{subfig}
\allowdisplaybreaks[4]
